\documentclass[sigconf, nonacm]{acmart}
\usepackage{floatflt}
\usepackage{lipsum}
\usepackage{multirow,booktabs}
\usepackage{pgf-pie}
\usepackage{rotating}
\usepackage{wrapfig}
\usepackage{soul}
\usepackage{mdframed}
\usepackage{graphicx}
\usepackage{amsmath,amsfonts}
\usepackage{bbm}
\usepackage{comment}
\usepackage{subcaption}
\usepackage{enumitem}
\usepackage{csquotes}
\usepackage{tabularx}
\usepackage{enumitem}
\usepackage{tikz}
\usetikzlibrary{positioning}
\usepackage{mathtools}
\usepackage[a-1b]{pdfx}  % load first if possible
\usepackage[T1]{fontenc}
\usepackage{microtype}
\usepackage{fnpct}  % adjusts kerning of footnote marks (superscripts) before commas or periods
\usepackage{color,tcolorbox, xcolor}
\usepackage{xspace}

\usepackage{pifont}% http://ctan.org/pkg/pifont
\newcommand{\cmark}{\ding{51}}%
\newcommand{\xmark}{\ding{55}}%

\graphicspath{{./figs/}}

\usepackage{algorithm}
\usepackage[noend]{algpseudocode}
\DeclarePairedDelimiterX\set[1]\lbrace\rbrace{\,#1\,}
\usepackage{tikz}
\usetikzlibrary{positioning}
\newcommand*{\Reals}{\mathbb{R}}

\DeclareMathOperator*{\argmin}{arg\,min}

\newcommand{\stitle}[1]{\vspace{1mm}\noindent{\textbf{#1}}.}

\newcommand{\dee}{\mathcal{D}}

\newcommand{\Gee}{\mathcal{G}}
\newcommand{\gee}{\mathbf{g}}

\newcommand{\tee}{\mathcal{T}}
\newcommand{\eps}{\varepsilon}
\newcommand{\opt}{\mathrm{OPT}}

\newcommand{\at}[1]{{\tt \small #1}\xspace}
\newcommand{\eat}[1]{}

\newcommand{\DP}{\mathsf{DP}}
\newcommand{\LS}{\mathsf{LS}}
\newcommand{\DnC}{\mathsf{D\&C}}

\newtheorem{theorem}{Theorem} 
\newtheorem{lemma}[theorem]{Lemma}
\newtheorem{definition}{Definition}
\newtheorem{example}{Example}
\newtheorem{problem}{Problem}

\newtcolorbox{highlightbox}{
  colback=blue!10,
  colframe=blue!20,
  arc=0.5mm, % Rounder edges
  fonttitle=\bfseries,
  boxrule=0mm,
  boxsep=0mm,
  left=0mm,
  right=0mm,
  top=0mm,
  bottom=0mm
}

\newtcolorbox{examplebox}{
  colback=blue!10,
  colframe=blue!20,
  arc=2mm, % Rounder edges
  fonttitle=\bfseries,
  boxrule=0mm,
  boxsep=1mm,
  left=0mm,
  right=0mm,
  top=0mm,
  bottom=0mm
}

\newtcolorbox{defbox}{
  colback=orange!10,
  colframe=orange!20,
  arc=2mm, % Rounder edges
  fonttitle=\bfseries,
  boxrule=0mm,
  boxsep=1mm,
  left=0mm,
  right=0mm,
  top=0mm,
  bottom=0mm
}

\newtcolorbox{pbox}{
  colback=black!5,
  colframe=black!30,
  arc=2mm, % Rounder edges
  fonttitle=\bfseries,
  boxrule=0mm,
  boxsep=1mm,
  left=0mm,
  right=0mm,
  top=0mm,
  bottom=0mm
}

\newtcolorbox{visionbox}{ % vision box
  colback=green!10,
  colframe=green!20,
  arc=2mm, % Rounder edges
  fonttitle=\bfseries,
  boxrule=0mm,
  boxsep=1mm,
  left=0mm,
  right=0mm,
  top=0mm,
  bottom=0mm
}
\newcommand{\technicalreport}[1]{#1\xspace}
\newcommand{\submit}[1]{}

\newcommand\vldbdoi{10.14778/3828612.3828619}
\newcommand\vldbpages{2617 - 2629}
\newcommand\vldbvolume{19}
\newcommand\vldbissue{X}
\newcommand\vldbyear{2026}
\newcommand\vldbtitle{\shorttitle} 
\newcommand\vldbavailabilityurl{https://github.com/asudeh/UnbiasedBinning}
\newcommand\vldbpagestyle{plain} 

\begin{document}

\title{Unbiased Binning for Fairness-aware Attribute Representation}

\author{Abolfazl Asudeh}
\orcid{0000-0002-5251-6186}
\affiliation{%
  \institution{University of Illinois Chicago}
  % \city{Chicago}
  % \state{IL}
  \country{}
}
\email{asudeh@uic.edu}

\author{Zeinab (Mila) Asoodeh}
\orcid{0009-0002-2087-4947}
\affiliation{%
  \institution{Independent Researcher}
  \country{}
}
\email{za.asudeh@gmail.com}

\author{Bita Asoodeh}
\orcid{0009-0007-8835-2746}
\affiliation{%
  \institution{University of Edinburgh}
  \country{}
}
\email{bita.asoodeh@ed.ac.uk}

\author{Omid Asudeh}
\orcid{0000-0002-5737-1780}
\affiliation{%
  \institution{University of Utah}
  \country{}
}
\email{asudeh@cs.utah.edu}

\begin{abstract}
Discretizing raw features into bucketized attributes is a common step before sharing a dataset. However, this process can inadvertently introduce bias and amplify unfairness in downstream tasks.

In this paper, we address this issue by formulating the unbiased binning problem, which seeks bucketized attributes that satisfy group parity. We develop an efficient dynamic programming algorithm to solve this problem for equal-size binning.
In practice, however, an unbiased binning may incur a high price of fairness or may not exist at all, particularly when group distributions differ substantially. To accommodate settings in which small deviations from perfect parity are acceptable, we introduce $\eps$-biased binning, which restricts group disparities across buckets to at most $\eps$. 
We first present a dynamic programming algorithm, $\DP$, that computes the optimal solution in quadratic time. Although polynomial, $\DP$ does not scale to large datasets. To address this limitation, we propose a practically scalable algorithm based on a local search strategy ($\LS$).
A key component of $\LS$ is a divide-and-conquer algorithm ($\DnC$) that finds a solution in near-linear time. We prove that $\DnC$ always returns a valid solution whenever one exists. The $\LS$ algorithm then performs a local search, using the $\DnC$ solution as an upper bound, to find the optimal solution. Our $\LS$ and $\DnC$ algorithms are general and are not restricted to equal-size binning.

To complement our theoretical analysis, we conduct extensive experiments on real-world and synthetic datasets. Besides confirming the efficiency of the algorithms, our experiments verify that while fairness-unaware binning can generate biased attribute representations, this bias can be significantly reduced with a negligible price of fairness. 
\end{abstract}

\maketitle

%%% do not modify the following VLDB block %%
%%% VLDB block start %%%
\pagestyle{\vldbpagestyle}
\begingroup\small\noindent\raggedright\textbf{PVLDB Reference Format:}
This is the extended version (technical report) of the following publication.
\\
Abolfazl Asudeh, Zeinab Asoodeh, Bita Asoodeh, and Omid Asudeh. \vldbtitle. PVLDB, \vldbvolume(\vldbissue): \vldbpages, \vldbyear.\\
\href{https://doi.org/\vldbdoi}{doi:\vldbdoi}
\endgroup
\begingroup
\renewcommand\thefootnote{}\footnote{\noindent
This work is licensed under the Creative Commons BY-NC-ND 4.0 International License. Visit \url{https://creativecommons.org/licenses/by-nc-nd/4.0/} to view a copy of this license. For any use beyond those covered by this license, obtain permission by emailing \href{mailto:info@vldb.org}{info@vldb.org}. Copyright is held by the owner/author(s). Publication rights licensed to the VLDB Endowment. \\
\raggedright Proceedings of the VLDB Endowment, Vol. \vldbvolume, No. \vldbissue\ %
ISSN 2150-8097. \\
\href{https://doi.org/\vldbdoi}{doi:\vldbdoi} \\
}\addtocounter{footnote}{-1}\endgroup
%%% VLDB block end %%%

%%% do not modify the following VLDB block %%
%%% VLDB block start %%%
\ifdefempty{\vldbavailabilityurl}{}{
\vspace{.3cm}
\begingroup\small\noindent\raggedright\textbf{PVLDB Artifact Availability:}\\
The source code, data, and/or other artifacts have been made available at \url{\vldbavailabilityurl}.
\endgroup
}
%%% VLDB block end %%%

\section{Introduction}\label{sec:intro}

For various reasons such as augmenting privacy preservation and minimizing the impact of noise~\cite{mironchyk2017monotone}, 
data sharing platforms often apply attribute binning (also known as bucketization) in domains such as credit scoring, {\em before} (publicly) sharing a dataset~\cite{oliveira2008rigorous,zeng2014necessary}.
Binning partitions the raw attribute values into discrete intervals, known as bins or buckets\footnote{We use the terms ``bin'' and ``bucket'' interchangeably.}, each representing a category of values. 
The original attribute values that fall into a given interval are replaced by the bin value.
For instance, equal-size (aka. equal-frequency) binning ensures an equal number of tuples fall inside each bucket~\cite{kotsiantis2006discretization,boulle2005optimal}.

Despite its conceptual simplicity and computational appeal, the process of attribute bucketization frequently induces unintended biases, thereby amplifying preexisting disparities embedded within the dataset. 
% \textcolor{blue}{ \marginpar{Meta4}
Since downstream analytics depend critically on the statistical properties of the input data, such biases can directly translate into unfairnesses in the outcomes of learned models. 
% }
To further clarify this issue, let us present the following motivation example.

\begin{visionbox}
{\small
\begin{example}\label{ex-binning}
% Consider a credit scoring application that stratifies individuals based on income.
As a common practice before sharing the datasets, a data sharing platform, such as \at{Chicago Open Data Portal}~\cite{kassen2013promising}, 
may bucketize attributes such as income.
However, income disparities across different demographic groups, when combined with binning, may inadvertently cluster marginalized populations predominantly into lower-tier buckets. 
Such biased attributes can propagate biases throughout the data analytics pipeline, further exacerbating unfairness in downstream decisions.
% \textcolor{blue}{
This propagation is not merely theoretical: as we shall experimentally demonstrate in Section~\ref{exp:validation}, biases introduced during binning directly translate into measurable disparities in downstream models. This highlights the importance of addressing bias at the data preparation stage, prior to model training.
% }
\end{example}
}
\end{visionbox}

The bias issues become critically important in sensitive domains such as finance, healthcare, employment screening, and criminal justice, wherein biased data transformations can yield discriminatory and inequitable outcomes~\cite{machinebias,barocas2016big}. For example, diagnostic healthcare models trained on bucketized attributes exhibiting demographic bias might disproportionately recommend preventive care measures to certain populations while neglecting others~\cite{obermeyer2019dissecting}. Similarly, credit scoring systems employing improperly bucketized financial attributes may lead to unfair loan application denials, particularly disadvantaging marginalized communities~\cite{hurley2016credit}.

\begin{table*}[!tb]
    \centering
    \small
    \caption{Key aspects of the algorithms presented in the paper.}
    \label{tab:results}
    % \vspace{-3mm}
    \begin{tabular}{c|ccccc|c}
        % \hline
         &  & & \textbf{Practically} & \textbf{Time} & \textbf{Space}& \textbf{Problem} \\
        \textbf{Objective} & \textbf{Algorithm} & \textbf{Optimal?} & \textbf{Scalable?} & \textbf{Complexity} & \textbf{Complexity}& \textbf{Scope}\\
        \hline\hline
        % Unbiased & {\sc UnbiasedBinning} &\cmark &\cmark &$O(n \log n + m^2 k)$&$O(mk)$&Equal-size Binning\\
        % Binning &(Algorithm~\ref{alg:step2})&&&&& \\
        Unbiased Binning & {\sc UnbiasedBinning} (Alg.~\ref{alg:step2}) &\cmark &\cmark &$O(n \log n + m^2 k)$&$O(mk)$&Equal-size Binning\\
        \hline 
        $\eps$-biased & $\DP$ (Alg.~\ref{alg:ebiased}) &\cmark & \xmark &$O(n^2k)$&$O(n^2)$&Equal-size Binning\\
        \cline{2-7}
        Binning &$\DnC$ (Alg.~\ref{alg:dnc})&\xmark&\cmark& $O(n\log n)$&$O(n)$& General \\
        \cline{2-7}
                &$\LS$ \technicalreport{(Alg.~\ref{alg:ls})}&\cmark&\cmark&$O(n \log n + n^{k-1})$&$O(n)$&General \\
        % \hline
    \end{tabular}
    % \vspace{-3mm}
\end{table*}

Despite its substantial implications for fairness in data-driven systems, this problem has largely been overlooked in the literature, which has instead focused primarily on training fair models using potentially biased attributes. 
To address this gap, we introduce the {\em unbiased binning} problem, whose objective is to ensure demographic group parity across buckets while minimizing deviations from the original bucketization. 

Unbiased binning operates early in the data analytics pipeline~\cite{jagadish2014big} and complements, rather than replaces, existing algorithmic fairness approaches, including FairML~\cite{barocas2023fairness}. As a preprocessing step applied {\em before} a dataset is shared and subsequently used for a variety of data analytics tasks, such as training machine learning models, unbiased binning targets a distinct objective and {\em does not, by itself, guarantee fairness in downstream tasks}.
Rather, it is a necessary (but insufficient) component of the responsible data analytics pipeline~\cite{jagadish2014big, shahbazi2023representation,nargesian2022responsible}. Ensuring fairness-aware dataset preparation must therefore be coupled with the development of fair algorithms and other downstream interventions, including fair model development techniques~\cite{barocas2023fairness}, to achieve fair models and outcomes (see Sections~\ref{sec:discussion} and \ref{sec:relatedwork}).

We develop an efficient dynamic programming ($\DP$) to solve the problem for equal-size binning.
Specifically, we first define a small set of ``boundary candidates'', and prove that the optimal unbiased binning must draw the bucket boundaries only from this set. We then develop our $\DP$ algorithm on top of the boundary candidates to find the optimal unbiased binning.

Due to several reasons, such as distribution mean differences between various groups, unbiased binning may result in a high price of fairness, significantly deviating from the initial binning, while sometimes unbiased binning may not even exist.
% On the other hand, a small amount of bias might be tolerable in binning in such settings.
Therefore, we next introduce the problem of {\em $\eps$-biased binning} that bounds the group disparity across buckets to $\eps$.
Our solution for unbiased binning is not applicable for $\eps$-biased binning, since the concept of boundary candidates does not extend to the new setting.
To address this issue, we introduce a matrix that specifies the set of possible $\eps$-biased buckets, and use it to extend our dynamic programming ($\DP$) approach for $\eps$-biased binning.
Our DP solution, however, is not scalable to large settings due to its quadratic time complexity.

Thus, we propose a practically scalable solution for $\eps$-biased binning based on a local search around the initial binning.
A key component of our practical solution is a divide-and-conquer ($\DnC$) algorithm that finds a valid $\eps$-biased binning in near-linear time. 
We prove that when $\DnC$ fails to find a valid solution, the problem is infeasible. 
$\DnC$ can serve as an efficient and scalable algorithm for finding a practical (not necessarily optimal) solution for the problem.
The local search $\LS$ algorithm uses the output of $\DnC$ to limit its search space to only the combinations with smaller objective values.
Unlike our $\DP$ algorithms, the $\DnC$ and $\LS$ algorithms are general, and not limited to equal-size binning.
Table~\ref{tab:results} summarizes the key aspects of our proposed algorithms.

Besides theoretical analysis, we conduct an extensive set of experiments on both real-world and synthetic datasets to
empirically validate the practical relevance of our problem and the effectiveness of our algorithms, and evaluate the performance of our algorithms across a wide range of settings.
Our experimental results demonstrate not only the efficiency and scalability of our algorithms but also underscore that traditional fairness-unaware binning approaches can lead to biased attribute representations, which may result in unfair decisions in downstream tasks.
On the other hand, our experiments confirm that unbiased binning can significantly reduce group unfairness at a small cost. 

\section{Preliminary}\label{sec:pre}

\subsection{Data Model}

Let \(\dee=\{t_i\}_{i=1}^n\) be a set of $n$ tuples defined over the attributes $X=\{x_j\}_{j=1}^d$.
Let $\Gee=\{\gee_1,\cdots,\gee_\ell\}$ be the set of (demographic) groups\footnote{$\Gee$ can be the set of subgroups (e.g., \at{black-female}) defined as the intersection of multiple sensitive attributes such as \at{race} and \at{gender}.}, where $\ell$ is a small constant.

Each tuple $t\in\dee$ belongs to a group $\gee(t)\in \Gee$.
We use $\Gee_l$ to refer to the set of tuples that belong to the group $\gee_l\in \Gee$. That is, $\Gee_l = \{t\in\dee~|~\gee(t)=\gee_l\}$.
In our examples, we use colors to distinguish the groups. Particularly, we use \at{red} and \at{blue} to refer to binary groups. %(such as \at{male} and \at{female}).

% \vspace{-3mm}
\subsection{Problem Definition}
We study the attribute binning (aka bucketization) problem.

\stitle{$k$-binning}
Given an integer value $k\ll n$ and an attribute $x\in \Reals$ from the set $X$, a $k$-binning is a partitioning of $x$ into $k$ ordered buckets $\mathcal{B}=\{B_1,\cdots,B_k\}$ by specifying $k-1$ boundaries $\{v_1,\cdots,v_{k-1}\}$.
The size of a bucket $B_j$ is the number of tuples with their value of attribute $x$ being in the range $(v_{j-1},v_j]$\footnote{$B_1$ is $x\leq v_1$ and $B_k$ is $x>v_{k-1}$.}. That is,
\[ |B_j\cap \dee| = \left|\{t_i\in \dee~|~ t_i[x]\in (v_{j-1}, v_j]\}\right|.\]

Equal-size (aka equal-frequency) $k$-binning is the $k$-binning of $x$, where all buckets contain the same number of tuples. That is,
\(|
B_j\cap \dee|=\frac{n}{k}, \quad \forall j\in[k]
\).

\technicalreport{
\begin{table}[!tb]
    \centering
    \caption{Table of Notations}
    \label{tab:notations}
    % \vspace{-3mm}
    \begin{tabular}{c l}
        \toprule
        {\bf notation}& {\bf description}
        \\ \toprule
        $\dee$ & the dataset 
        \\ \hline
        $n$ & size of $\dee$
        \\ \hline
        $\gee_i$ & group $i$ in the set of groups $\Gee$ 
        \\ \hline
        $\ell$ & The total number of groups, i.e., $|\Gee|$
        \\ \hline
        $\Gee_i$ & The set of tuples in group $\gee_i$ 
        \\ \hline
        $k$ & the total number of buckets, i.e., $|\mathcal{B}|$  
        \\ \hline
        $B_i$ & the $i$-th bucket   
        \\ \hline
        $T$ & the set of boundary candidates 
        \\ \hline
        $m$ & the number of the boundary candidates
        \\ \hline
        $\eps$ & maximum-bias threshold in $\eps$-biased binning
        \\ \hline
        $\beta(~)$ & The bias of a bucket $B$ or a binning $\mathcal{B}$ 
        \\ \bottomrule
    \end{tabular}
\end{table}
}

\stitle{Unbiased binning}
Our goal is to ensure group parity in
$k$-binning; that is, the group ratios are equal across all buckets.
Formally,

\begin{defbox}
\begin{definition}[Unbiased Binning]\label{def:unbiased}
    Given a dataset $\dee$, a binning $\mathcal{B}=\{B_i\}_{i=1}^k$ of an attribute $x$ is unbiased if
    \[
    \frac{|\Gee_l|}{|\dee|} = \frac{|B_j\cap \Gee_l|}{|B_j\cap \dee|}, \quad \forall l\in[\ell], j\in[k]
    \]
\end{definition}
\end{defbox}

Due to various reasons, such as social and historical discrimination, equal-size binning without considering group ratios may not achieve group parity. %\footnote{We extensively evaluate this in our experiments in Section~\ref{sec:exp}.}.
Therefore, to ensure unbiased binning, Problem~\ref{problem} aims for the closest unbiased binning to the equal-size binning.
Formally,

\begin{pbox}
    \begin{problem}\label{problem}
        Given a dataset $\dee$, an attribute $x\in X$, and a positive integer value $k$,
        partition $x$ into $k$ bins $\mathcal{B}=\{B_i\}_{i=1}^k$, such that $\mathcal{B}$ is unbiased, and 
        \(
        \big( \max_{i\in[k]} |B_i\cap \dee| - \min_{j\in[k]} |B_j\cap \dee|~\big)
        \)
        is minimized.
    \end{problem}
\end{pbox}

% \marginpar{\textcolor{blue}{R1.D1}}
\begin{visionbox}
% \color{blue}
\small\it
\noindent{\sc Example~\ref{ex-binning} (continued).}
To illustrate the notion of unbiased binning,
consider a toy dataset of $n=16$ tuples belonging to two groups,
$\Gee_1$ (blue) and $\Gee_2$ (red), where $|\Gee_1|=8$ and $|\Gee_2|=8$.
Suppose the tuples are sorted by an attribute (income), as shown in the first row of Figure~\ref{fig:step1}.
Thus, the overall ratio of group $\Gee_1$ is $\frac{|\Gee_1|}{|\dee|}=0.5$.

Let $k=4$.
The fairness-unaware equal-size binning selects boundaries at (the income value of) the 4th, 8th, and 12th tuples, producing four buckets of size $4$.
However, the resulting buckets are biased -- for example, the first bucket contains three red and one blue tuple.

\noindent\emph{Unbiased binning.}
An unbiased binning requires each bucket to preserve the global ratio (i.e., $0.5$).
One valid partition is:
\[
B_1 = \{1,\ldots,6\},\;
B_2 = \{7,8\},\;
B_3 = \{9,10,11,12\},\;
B_4 = \{13,14,15,16\}.
\]
\end{visionbox}
% \color{black}

While Problem~\ref{problem} is defined on top of equal-size binning, the Unbiased Binning problem, in general, can be defined on top of any $k$-binning. 
Specifically, suppose an algorithm has generated an initial binning $\hat{\mathcal{B}}=\{\hat{B}_i\}_{i=1}^k$ of an attribute $x$ on a dataset $\dee$.
The extension of Problem~\ref{problem} for general binning, aims to minimally change the initial boundaries to generate an unbiased binning $\mathcal{B}=\{B_i\}_{i=1}^k$. That is, to minimize $\vert\hat{B}_i - B_i\vert$, $\forall i\in[k]$.

In the following sections, we first provide an efficient optimal solution for Problem~\ref{problem}. Then, in Section~\ref{sec:localsearch}, we extend our scope to general binning, where we provide practically efficient and scalable algorithms.

\begin{figure*}[!tb]
    \centering
    \resizebox{.8\textwidth}{!}{\begin{tikzpicture}
    \def\labelwidth{2.2}
    \def\ncols{16}

    % \node[font=\bfseries] at (\labelwidth/2, -0.5) {Sorted\\ Indices};
    \node[font=\bfseries, align=center] at (\labelwidth/2, -.5) {Sorted\\Indices};
    % \node[font=\bfseries] at (\labelwidth/2, -1.5) {Groups};
    \node[font=\bfseries] at (\labelwidth/2, -1.5) {$\frac{\text{\bf Blue}}{\text{\bf All}}$ Ratios};
    \node[font=\bfseries, align=center] at (\labelwidth/2, -2.5) {Boundary\\Candidates};

    % \foreach \i in {1,...,16} {
    %     \node at ({\labelwidth + \i - 0.5}, -0.5) {$\i$};
    % }
    \def\groups{{"R","R","B","R","B","B","R","B","R","R","B","B","B","R","R","B"}}
    \foreach \i in {1,...,16} {
        \pgfmathtruncatemacro{\idx}{\i - 1}
        \pgfmathsetmacro{\xpos}{\labelwidth + \idx}
        \pgfmathparse{\groups[\idx]}
        \edef\groupval{\pgfmathresult}

        \ifthenelse{\equal{\groupval}{R}}{
            \fill[red!30] (\xpos,0) rectangle ++(1,-1);
        }{
            \fill[blue!30] (\xpos,0) rectangle ++(1,-1);
        }

        \node at (\xpos + 0.5, -0.5) {$\i$};
    }
    \def\ratios{{"0","0","1/3","1/4","2/5","1/2","3/7","1/2","4/9","2/5","5/11","1/2","7/13","1/2","7/15","1/2"}}
    \foreach \i in {1,...,16} {
        \pgfmathtruncatemacro{\idx}{\i - 1}
        \pgfmathsetmacro{\xpos}{\labelwidth + \idx}
        \pgfmathparse{\ratios[\idx]}
        \edef\ratioval{\pgfmathresult}
        \node at (\xpos + 0.5, -1.5) {$\ratioval$};
    }

    \foreach \i in {1,...,16} {
        \pgfmathtruncatemacro{\idx}{\i - 1}
        \pgfmathsetmacro{\xpos}{\labelwidth + \idx}
        \pgfmathparse{\ratios[\idx]}
        \edef\ratioval{\pgfmathresult}
        \ifthenelse{\equal{\ratioval}{1/2}}{
            \node at (\xpos + 0.5, -2.5) {$\checkmark$};
        }{}
    }

    \draw[->, thick, opacity=0.8]
    ({\labelwidth + 0.5}, -1.7) -- ({\labelwidth + 5.2}, -1.7);
    \node[] at ({\labelwidth + 3}, -1.85) {\tt \small First Pass};
    \draw[->, thick, opacity=0.8]
    ({\labelwidth + .5}, -2.7) -- ({\labelwidth + 5.2}, -2.7);
    \node[] at ({\labelwidth + 3}, -2.85) {\tt\small Second Pass};

    \foreach \row in {0,...,2} {
        \draw (0,-\row) rectangle ++(\labelwidth,-1); 
        \foreach \col in {1,...,16} {
            \draw ({\labelwidth + \col - 1}, -\row) rectangle ++(1,-1);
        }
    }
    \draw[thick] (0,-1) -- ({\labelwidth + 16}, -1);
    \draw[thick] (0,-1.05) -- ({\labelwidth + 16}, -1.05);
\end{tikzpicture}}
    %\vspace{-3mm}
    \caption{Illustration of Boundary-candidates identification based on the Blue/all ratios.}
    \label{fig:step1}
    %\vspace{-5mm}
\end{figure*}

%\vspace{-3mm}
\section{Unbiased Binning Algorithm}\label{sec:solution}
There are $n \choose k$ possible ways to partition an attribute $x$ into $k$ buckets.
Checking if a $k$-binning is unbiased and computing the objective value takes $O(n)$ time.
As a result, the brute-force approach for solving Problem~\ref{problem} has a time complexity of $O\big(n {n\choose k}\big)$.
Alternatively, one can model the problem as an Integer Programming (IP) problem and use an existing solver to solve the problem. This approach also is not practical, since IP$\in${\sc NP}-complete.

Therefore, in this section, we instead develop an efficient algorithm for the unbiased binning problem (Problem~\ref{problem}).
Our algorithm has two sequential steps, where in step 1, a set of ``boundary candidates'' is identified. In step 2, we use the boundary candidates and develop a dynamic programming (DP) algorithm that quickly finds the optimal unbiased binning.
Without loss of generality, to simplify the explanations, we first consider the binary-group case. %(e.g., \at{male/female}).
We then extend our approach to non-binary grouping in Section~\ref{sec:extension}.

%\vspace{-3mm}
\subsection{Step 1: Finding the boundary candidates}\label{sec:unbiased:step1}

Let us sort and index the tuples in $\dee$ based on their values on $x$, such that $\forall i\in[n]$, $t_i[x]\leq t_{i+1}[x]$.
{For every index $i\in [n]$ such that $t_i[x]<t_{i+1}[x]$\footnote{Since boundaries are defined by value thresholds, all tuples with identical values are assigned to the same bucket. Hence, if $t_i[x]=t_{i+1}[x]$, they cannot be separated.}, let $r_i$ be the ratio of the $\gee_1$ (blue) tuples in the set $\dee[1:i] = \{t_1, t_2, \cdots, t_i\}$.} That is,
\(
r_i = \frac{1}{i}\big|\dee[1:i]\cap \Gee_{1}\big|
\).
Note that $r_n$ is the overall ratio of the blue tuples in the dataset, i.e., $r_n = \frac{|\Gee_{1}|}{n}$.

We say an index $i\in [n]$ is a {\bf boundary candidate} if its blue ratio is equal to the overall ratio of blue tuples, i.e., $r_i=r_{n}$.
Accordingly, the set of boundary candidates is,
\(
T = \{i\in[n]~|~ r_i = r_n\}
\).
We use $m=|T|$ to refer to the size of the boundary candidates set. 

In order to efficiently identify the set of boundary candidates, we make two passes over $\dee$ (sorted on $x$) -- Figure~\ref{fig:step1} provides a toy example with $n=16$ tuples.

During the first pass, we maintain a counter $c$ that counts the number of blue tuples observed so far. For each index $i$, we increment the counter if $\gee(t_i)=\gee_1$. The ratio $r_i$ is hence $\frac{c}{i}$.
After computing all ratios, the second pass adds the indices where $r_i=r_n$ to the set of boundary candidates.

The pseudo-code of the function {\sc BCandidates} for finding the boundary candidates is provided in \submit{the technical report~\cite{techrep}}\technicalreport{Algorithm~\ref{alg:step1}}. 

\technicalreport{
\begin{algorithm}[ht]
\caption{Finding the Boundary Candidates}
\label{alg:step1}
\begin{algorithmic}[1]
\Require The tuples $\{t_1,\cdots,t_n\}$ sorted on $x$.
\Ensure the boundary candidates array $T$.
\Function{\sc BCandidates}{$\{t_1,\cdots,t_n\}$} 
    \Statex {\small /* First Pass: computing the ratios */ }
    \State $c\gets [0]_{l=1}^{(\ell-1)}$; $r\gets [0]_{n\times (\ell-1)}$ \Comment{\small initialize the variables}
    \For{$i\gets 1$ to $n$}
        \If{$\gee(t_i)\neq l$}
            $c[\gee(t_i)]\gets c[\gee(t_i)]+1$
        \EndIf
        \State $r[i,l]\gets c[l]/i,~\forall l\in[\ell-1]$
    \EndFor
    \Statex {\small /* Second Pass: finding the boundary candidates */ }
    \State $T\gets [~]$ \Comment{\small initialize the array of boundary candidates}
    \For{$i\gets 1$ to $n$}
        \If{$r[i,l]=r[n,l],~\quad \forall l\in[\ell-1]$}
            \State $T.append(i)$ \Comment{\small add $i$ to the end of $T$}
        \EndIf
    \EndFor
    \State {\bf Return} $T$
\EndFunction
\end{algorithmic}
\end{algorithm}
}

The following equation, (\ref{eq:buckets}), shows the corresponding binning of $x$ for a subset of boundary candidates $T'\subseteq T$ , where $|T'|=k-1$.

\begin{examplebox}
%\vspace{-2.6mm}
\begin{align} \label{eq:buckets}
    \mathcal{B}_{T'} = \big\{~~&B_1~:\hspace{19mm} x\leq t_{T'[1]}[x],
    \\
    \nonumber &B_2~:\hspace{3.5mm} t_{T'[1]}[x]<x\leq t_{T'[2]}[x],
    \\
    \nonumber&\cdots,
    \\
    \nonumber& B_k~: t_{T'[k-1]}[x]<x\hspace{18mm}
    \big\}
\end{align}
%\vspace{-5mm}
\end{examplebox}

Theorem~\ref{th:cboundary} plays a key role in the design of our algorithm.  

\begin{theorem}\label{th:cboundary}
    Let $\mathcal{B}$ be an unbiased $k$-binning for the attribute $x$ of a dataset $\dee$ (Definition~\ref{def:unbiased}).
    The boundaries of $\mathcal{B}$ must only belong to the set of boundary candidates $T$.
\end{theorem}

\begin{proof}
    Let $m=|T|$.
    First, let us prove that any pair of the boundary candidates forms an unbiased bucket. That is, 
    \[
        \frac{|\{i'\in (i,j]~|~t_{i'}\in \Gee_{1}\}|}{j-i}=r_n, \quad \forall i,j\in T, i<j
    \]
    This is easy to prove, following the definition of boundary candidates:
    \[
    r_i = r_n \Rightarrow \big|\{\dee[1:i]\cap \Gee_{1}\}\big| = i\times r_n
    \]
    \begin{align*}
        \frac{|\{i'\in (i,j]~|~t_{i'}\in \Gee_{1}\}|}{j-i}  &= \frac{|\{\dee[1:j]\cap \Gee_{1}\}|-|\{\dee[1:i]\cap \Gee_{1}\}|}{j-i}\\
        &= \frac{r_n(j-i)}{j-i} = r_n %\hspace{32mm}\square
    \end{align*}
    As a result, any selection of $k-1$ boundaries from the set of boundary candidates $T[1,\cdots,m-1]$ produces an unbiased binning of size $k$ (recall that $T[m]=n$).
    This proves that any binning $\mathcal{B}$ whose boundaries belong to $T$ is unbiased.

    Next, we prove that the boundaries of any unbiased binning must only belong to $T$. In other words, any binning with at least one boundary that does not belong to $T$ is biased.
    
    Following proof-by-contradiction, suppose there exists an unbiased binning $\mathcal{B}$ with at least one boundary that does not belong to $T$.
    Let $\{{i_1}, \cdots,{i_{k-1}}\}$ be the indices of the boundaries of $\mathcal{B}$.
    Let $i_j\notin T$ be the last index in the boundaries that do not belong to $T$.
    Consider the bucket $B$ with this boundary as its beginning. If this is the last boundary, i.e., $j=k-1$, let us define the dummy index $i_{j+1}=n$. The bucket $B$ is then $\big(t_{i_j}[x], t_{i_{j+1}}[x]\big]$. Note that $i_{j+1}\in T$ since $i_j$ is the last boundary that does not belong to $T$. Also, recall that $T[m]=n$.
    In the following, we show that $B$ is biased, contradicting with the assumption that $\mathcal{B}$ is unbiased.

    Let $\rho = r_{i_j}-r_n$. Then $r_{i_j} = \rho+r_n$.
    Since $i_j\notin T$, $\rho \neq 0$. The ratio of group $\gee_1$ in $B$ is:
    \begin{align*}
        \frac{|\{i\in (i_j,i_{j+1}]~|~t_i\in \Gee_{1}\}|}{i_{j+1}-i_j}= \frac{r_{j+1}-r_{i_j}}{i_{j+1}-i_j}
        &= \frac{r_n\times i_{j+1} - (r_n+\rho)i_j}{i_{j+1}-i_j}\\
        &= r_n - \frac{\rho\times i_j}{i_{j+1}-i_j} \neq r_n
    \end{align*}
    % \vspace{-2mm}
\end{proof}

% \vspace{-10mm}
%\vspace{-3mm}
\subsection{Step 2: Dynamic Programming}

\begin{figure}[tb!]
    \centering
    \begin{tikzpicture}[scale=.4\textwidth/5cm]

    % Define colors
    \definecolor{coolblue}{RGB}{160,200,255}
    \definecolor{warmyellow}{RGB}{255,225,140}

    % Draw cells with reduced height
    \foreach \i in {0,...,4} {
        \ifnum\i<4
            \fill[coolblue] (\i,0) rectangle ++(1,1/3);
        \else
            \fill[warmyellow] (\i,0) rectangle ++(1,1/3);
        \fi
        \draw (\i,0) rectangle ++(1,1/3);
    }

    % Labels below cells
    \node at (0,-.2) {$0$};
    \node at (1,-.2) {$1$};
    \node at (2,-.2) {$2$};
    \node at (2.5,-.2) {$\cdots$};
    \node at (3,-.2) {$T[i]{-}1$};
    \node at (4,-.2) {$T[i]$};
    \node at (4.5,-.2) {$\cdots$};
    \node at (5,-.2) {$T[j]$};

    % Arrows ABOVE (adjusted for new height)
    \draw[<->, thick] (0.1,0.5) -- (3.98,0.5);
    \node at (2,0.7) {OPT$(i, \kappa-1)$};

    \draw[<->, thick] (4,0.5) -- (5,0.5);
    \node at (4.5,0.7) {$B[\kappa]$};

    % Curved arrow
    \draw[->, thick]
        (3.2,-.45) .. controls (3.7,-.46) .. (3.9,-0.3);
    \node at (1.6,-.45) {selecting $T[i]$ as the last boundary};

\end{tikzpicture}
    %\vspace{-4mm}
    \caption{illustration of selecting $T[i]$ as the last boundary of the subproblem $\opt(j,\kappa)$ for unbiased binning.}
  \label{fig:dp-unbiased}
    %\vspace{-4mm}
\end{figure}

Based on Theorem~\ref{th:cboundary}, the optimal solution for Problem~\ref{problem} is a subset of $k-1$ boundary candidates that minimizes the maximum bucket-size difference between the buckets.
Following this observation, we develop our dynamic programming (DP) algorithm as follows.

\stitle{The DP formulation}
We define a recursive function \(\mathrm{OPT}(j, \kappa)\) to represent the optimal solution for the subproblem partitioning $\{t_1,t_2,\cdots,t_{T[j]}\}$, into $\kappa$ unbiased buckets.
\(\opt(j, \kappa)\) returns the maximum width ($w^\uparrow$) and the minimum width ($w^\downarrow$) of the buckets for the optimal solution of this subproblem.
Following this formulation, the optimal solution for our problem is $\opt(m,k)$.

\stitle{The boundary condition} 
The number of buckets for any valid subproblem is at most $\kappa\leq j$.
When $\kappa=1$, the optimal solution for $\opt(j,1)$ generates the bucket $B_1: x\leq t_{T[j]}[x]$, where the width of the bucket is $|B_1\cap \dee| = T[j]$.
Hence,
\[
\opt(j,1) = (T[j],T[j]) \quad \text{\small \tt // $w^\uparrow=w^\downarrow=T[j]$}
\]

\subsubsection{Recursive formula (general case)}
% Recall that in a DP algorithm our goal is to break up a problem into a series of overlapping sub-problems, and build up solutions to more general sub-problems.
At any step $j\in[n]$, the optimal solution for $\opt(j,\kappa)$ should decide 
where to put {\em the boundary of the last bucket}, i.e., $B_\kappa$.
According to Theorem~\ref{th:cboundary}, it must select one of $\{T[1],\cdots,T[j-1]\}$ as the boundary of the last bucket $B_\kappa$ since choosing any other index $i'\notin T$ will make $B_\kappa$ biased. 

\stitle{Selecting $T[i]$ (Figure~\ref{fig:dp-unbiased})}
selecting $T[i]$, where $i<j$, is equivalent of forming a new bucket $B_\kappa = \big(t_{T[i]}[x],t_{T[j]}[x]\big]$. Then, the subproblem to solve is $\opt(i,\kappa-1)$: to form $\kappa-1$ bins, with a new end index $i$.
Let $(wp^\uparrow,wp^\downarrow) = \opt(i,\kappa-1)$.
Then after adding $B_\kappa$,
\(w_{i}^\uparrow = \max (wp^\uparrow, T[j]-T[i])\) and
\(w_{i}^\downarrow = \min (wp^\downarrow, T[j]-T[i])\), and the objective value (the max bucket-size difference) is
\(
obj_{i} = w_{i}^\uparrow - w_{i}^\downarrow
\).

\stitle{Choosing between the valid options}
Between various valid options, the optimal solution selects the one with the smallest objective value. As a result, 
\begin{align}\label{eq:dp}
\opt(j,\kappa) = (w^\downarrow_{i^*},w^\uparrow_{i^*}) \text{, where } i^*=\argmin\limits_{i\in [j-1]} {w_{i}^\uparrow - w_{i}^\downarrow}
\end{align}

\begin{algorithm}[t]
\caption{Unbiased Binning}
\label{alg:step2}
\begin{algorithmic}[1]
\Require The dataset $\dee$, the attribute x, the value $k$.
\Ensure The optimal unbiased $k$-binning.
% \Function{\sc UnbiasedBinning}{$\dee,k$} 
    \State $\{t_1,\cdots,t_n\}\gets$ {\sc sort}$(\dee,x)$ \Comment{\small sort $\dee$ on $x$}
    \State $T\gets$ {\sc BCandidates}$(\{t_1,\cdots,t_n\})$  \Comment{\small \technicalreport{Algorithm~\ref{alg:step1}}\submit{\textcolor{blue}{details in the technical report~\cite{techrep}}}}
    \State $m\gets |T|$
    \If{$m<k$}
        \State {\bf Return} infeasible
    \EndIf
    \State $M\gets [~]_{m\times k}$\Comment{\small initialize $M$}
    \For{$j\gets 1$ to $m$} \Comment{\small the boundary condition}
        \State $M[j,1] = \big\langle (T[j],T[j])~,~ 0~ \big\rangle$ 
    \EndFor
    \Statex {\small /*Filling the matrix $M$*/}
    \For{$\kappa\gets 2$ to $k$}
    \For{$j\gets \kappa$ to $m$}
        \State $\langle (w^\downarrow,w^\uparrow),i^*\rangle\gets \langle(0,\infty),null \rangle$
        \For{$i\gets 1$ to $j-1$}
            \Statex \hspace{9mm}{\small /*evaluating the selection $T[i]$ for the last bucket*/}
            \State $\langle (wp^\downarrow,wp^\uparrow),\square \rangle\gets M[i,\kappa-1]$
            \State $w_{i}^\uparrow = \max (wp^\uparrow, T[j]-T[i])$
            \State $w_{i}^\downarrow = \min (wp^\downarrow, T[j]-T[i])$
            \If{$(w_i^\uparrow - w_i^\downarrow)<(w^\uparrow - w^\downarrow)$}
                \State $\langle (w^\downarrow,w^\uparrow),i^*\rangle\gets \langle(w_i^\downarrow,w_i^\uparrow),i \rangle$
            \EndIf
        \EndFor
        $M[j,\kappa]\gets \langle (w^\downarrow,w^\uparrow),i^*\rangle$
    \EndFor
    \EndFor
    
    \State {\bf Return} {\sc TraceBack}$(M)$\Comment{\small \technicalreport{Algorithm~\ref{alg:traceback}}\submit{\textcolor{blue}{details in the technical report~\cite{techrep}}}}
% \EndFunction
\end{algorithmic}
\end{algorithm}

%\vspace{-3mm}
\subsubsection{Efficient implementation}
The pseudo-code of our algorithm for solving Problem~\ref{problem} is provided in Algorithm~\ref{alg:step2}.
Following a matrix-filling approach for the implementation of our dynamic programming formulation, we define the matrix $M$ of size $m \times k$.
Each cell $M[j,\kappa]$ contains two components: the pair $(w^\uparrow,w^\downarrow) = OPT\left(j,\kappa \right)$ and the left boundary $i^*$ of the last bucket -- the index $i^*$ that minimizes the objective value in Equation~\ref{eq:dp}.

Using the boundary condition, we fill the first column of the matrix as 
\[
M[j,1] = \big\langle (T[j],T[j])~,~ 0~ \text{\tt\small /*choosing $B_1:x\leq t_{T[j]}[x]$*/ } \big\rangle
\]
Next, we fill the matrix using Equation~\ref{eq:dp}.
Once $M$ is filled, we apply a backward tracing of the corresponding cells for the optimal solution to identify the generated buckets $\mathcal{B} = \{B_1, \cdots, B_k\}$.

\technicalreport{
\begin{algorithm}[t]
\caption{Trace back}
\label{alg:traceback}
\begin{algorithmic}[1]
\Require The matrix $M$ and the boundary candidates $T$.
\Ensure The selected boundaries.
\Function{\sc TraceBack}{$M,T$}
    \State $S\gets [~]$ \Comment{\small initialize the output}
    \State $(m,k)\gets shape(M)$
    \State $j\gets m; \kappa \gets k$
    \While{$\kappa>0$}
        \State $\langle (w^\uparrow,w^\downarrow),i \rangle \gets M[j,\kappa]$
        \State $S.push(t_{T[i]}[x])$ \Comment{\small add $t_{T[i]}[x]$ to the beginning of $S$}
        \State $j\gets i$; $\kappa \gets \kappa-1$
    \EndWhile
    % \State $S.push(t_{1}[x])$ \Comment{\small add $t_{1}[x]$ to the beginning of $S$}
    \State {\bf Return} $~S$
\EndFunction
\end{algorithmic}
\end{algorithm}
}

\begin{lemma}\label{lem:unbiasedcomplexity}
    The unbiased binning algorithm (Algorithm~\ref{alg:step2}) has a space complexity of $O(mk)$ and time complexity of $O(n\log n + m^2k)$.
\end{lemma}

\begin{proof}
{\em Space complexity:}
The algorithm stores (and fills) the matrix $M$ of size $m\times k$, hence having the space complexity of $O(mk)$.

\vspace{2mm}\noindent{\em Time complexity:}
\begin{itemize}[leftmargin=*]
    \item Sorting $\dee$ based on the attribute $x$ is done in $O(n\log n)$.
    \item The function {\sc BCandidates} (Algorithm~\ref{alg:step1}) makes two linear passes over $\dee$ and, hence, is in $O(n)$.
    \item Iterating through the indices $i<j$ to fill each cell of matrix $M[j,\kappa]$ is in $O(m)$. Therefore, filling the matrix $M$ is done in $O(m^2k)$.
    \item Finally, the {\sc TraceBack} function (Algorithm~\ref{alg:traceback} is in $O(k)$.
\end{itemize}
Therefore, putting all steps together, Algorithm~\ref{alg:step2} has the time complexity of $O(n\log n + m^2 k)$.
\end{proof}

\subsubsection{Extension beyond binary grouping}\label{sec:extension}

The extension of our algorithm beyond binary grouping is straightforward and only requires updating step 1.
Recall that in step 1, our algorithm identifies a set of boundary candidates. Then in step 2, it selects a subset of $k-1$ boundaries to minimize the maximum size difference between the buckets.

When $\ell=|\Gee|\geq 2$, we need to compute the group ratios for all but the last group~\footnote{When a binning is unbiased for the first $\ell-1$ groups, it is also unbiased for $\gee_\ell$.}. That is,
\[
r^{(l)}_i = \frac{1}{i}\big|~\dee[1:i]\cap \Gee_l\big|\, ,\quad \forall l\in [\ell-1], i\in [n]
\]

The boundary candidates, then, are the indices for which the overall ratios are equal to the group ratios:
\[
T = \{i\in [n]~|~ r^{(l)}_i = r^{(l)}_n,~\forall l\in [\ell-1]\}
\]

% \vspace{-3mm}
\section{$\eps$-biased Binning}
Problem~\ref{problem} aims at finding an unbiased binning.
This requirement, however, may be too restrictive in practice, particularly (a) when the expected values of $x$ largely vary for different groups, and (b) for non-binary groups, where only a small boundary candidates set may satisfy the overall ratios for all groups.
On the other hand, a small deviation from the unbiased binning may be acceptable in such settings.

Using Definition~\ref{def:unbiased}, we measure the bias $\beta$ of a bucket $B_j$ with respect to a group $\gee_l$ as the
difference in the ratio of $\gee_l$ in $B_j$ compared to its overall ratio\footnote{Alternatively, one could measure the bias using division, instead of subtraction. That is, for a bucket $B_j$ and group $\gee_l$,\\ $\beta^{div}_\dee(B_j,\gee_l) = 1-\big(\max(\frac{|\Gee_l\cap B_j|}{|\dee\cap B_j|}, \frac{|\Gee_l|}{n})/\min(\frac{|\Gee_l\cap B_j|}{|\dee\cap B_j|}, \frac{|\Gee_l|}{n})\big)$. Our algorithms are agnostic to the choice bias measure, and can use $\beta^{div}$ instead of Equation~\ref{eq:bias_bucket}.}. That is,
\begin{align}\label{eq:bias_bucket}
    \beta_\dee(B_j,\gee_l) = \left\vert
    \frac{|\Gee_l\cap B_j|}{|\dee\cap B_j|} - \frac{|\Gee_l|}{n}
    \right\vert
\end{align}

The bias of a binning $\mathcal{B}$ with respect to a group $\gee_l$ is the maximum bias of its buckets for $\gee_l$, and the overall bias of $\mathcal{B}$ is measured as the maximum of its bias across different groups. That is,
\begin{align}\label{eq:bias}
    \beta_\dee(\mathcal{B}) = \max_{l\in[\ell]}  \max_{j\in[k]} \beta_\dee(B_j,\gee_l)
\end{align}

\begin{defbox}
\begin{definition}[$\eps$-biased Binning]\label{def:epsbiased}
    Given a dataset $\dee$, a binning $\mathcal{B}=\{B_i\}_{i=1}^k$ of an attribute $x$ is $\eps$-biased, if \(\beta_\dee(\mathcal{B})\leq \eps\).
\end{definition}
\end{defbox}

Following this definition, we 
extend Problem~\ref{problem} to find an $\eps$-biased $k$-binning, instead of an unbiased one\footnote{Section~\ref{sec:discussion} provides guidelines for selecting of $\eps$.}.

% \marginpar{\textcolor{blue}{R1.D1}}
\begin{visionbox}
% \color{blue}
\small\it
\noindent{\sc Example~\ref{ex-binning} (continued).}
Let us continue with the toy example in Figure~\ref{fig:step1} to further explain the notion of $\eps$-biased binning.
In practice, an unbiased binning may not exist or may incur a high price of fairness.
In this example, the unbiased solution produces imbalanced bucket sizes (e.g., $|B_1|=6$ vs.\ $|B_2|=2$).
Allowing a bounded deviation, $\epsilon$-biased binning enables more balanced partitions.
For instance, with $\eps=0.17$, a valid $\eps$-biased binning is:
\[
B_1 = \{1,\ldots,5\},\;
B_2 = \{6,7,8\},\;
B_3 = \{9,10,11,12\},\;
B_4 = \{13,14,15,16\},
\]
which yields bucket sizes within $4\pm1$.
\end{visionbox}

The solution proposed for unbiased binning does not directly extend to $\eps$-biased binning, mainly since Theorem~\ref{th:cboundary} no longer holds for $\eps$-biased binning.
Therefore, we introduce new ideas that enable the development of $\DP$ for this setting.
The details of the dynamic programming algorithm are provided in Appendix~\ref{sec:dp2}.

\begin{lemma}\label{lem:ebiased-complexity}
    The DP algorithm for $\eps$-biased binning has a time complexity of $O(n^2k)$ and space complexity of $O(n^2)$.
\end{lemma}

% \input{ebiased-dp}

%\vspace{-3mm}
\subsection{A Practically Scalable Algorithm based on Local Search (LS)}\label{sec:localsearch}
The $\DP$ algorithm for $\eps$-biased binning has a quadratic time and space complexity, which prevents it from scaling to very large values of $n$.
Therefore, we propose a local search algorithm ($\LS$) as our practical approach for the large-scale settings.

At a high level, the $\LS$ algorithm follows two steps. Step 1 is an efficient divide-and-conquer ($\DnC$) algorithm that quickly finds a valid solution for the problem (if one exists). While $\DnC$ finds an $\eps$-biased binning, it may not be optimal, i.e., its solution may not minimize the maximum width difference between the buckets.
In other words, Step 1 quickly finds a valid solution for the problem.
Next, using the objective value of the solution found during Step 1 as an upper bound, Step 2 of the $\LS$ algorithm applies a local search around the equal-size binning to find the optimal solution.

Unlike our dynamic programming algorithms, the {\em $\DnC$ and $\LS$ algorithms are not limited to equal-size binning}, as they can be applied as a post-processing step on the output of any (black-box) bucketization approach.

% the boundaries of a bucketization 

\begin{figure}[!tb]
    \centering
    \begin{tikzpicture}
    \definecolor{highlight}{RGB}{255,220,120}

    % \fill[highlight] (2,1.9) rectangle ++(2,1);
    % \fill[highlight] (4,.5) rectangle ++(1,1);

    % \foreach \i in {0,...,1} {
    %     \draw (4*\i,3) rectangle ++(4,1);
    %     \node at (2+4*\i,3.5) [] {\small $\lceil k/2\rceil$};
    % }

    % Root
    \node at (.7,3.7) {$0$};
    \node at (6.7,3.7) {$n$};
    \node at (3.7,3.7) {$\mathsf{D\&C}(0,n,k,\eps)$};

    \fill[highlight] (3.6,2.5) rectangle ++(.2,1);
    \draw[dashed] (3.7,2.5) -- (3.7,3.5);
    \draw[thick] (3.8,2.5) -- (3.8,3.5);
    \draw (0.7,2.5) rectangle ++(6,1);

    \draw[dashed] (.7,2.5) -- (0,2);
    \draw[dashed] (3.8,2.5) -- (3,2);
    \draw[dashed] (3.8,2.5) -- (4.5,2);
    \draw[dashed] (6.7,2.5) -- (7.5,2);

    \node at (3.8,2.3) {$i_1$};
    \node at (1.8,2.3) {$\mathsf{D\&C}(0,i_1,\lceil \frac{k}{2}\rceil,\eps)$};
    \node at (5.7,2.3) {$\mathsf{D\&C}(i_1,n,\lfloor \frac{k}{2}\rfloor,\eps)$};

    % left child
    \fill[highlight] (1.35,1) rectangle ++(.3,1);
    \draw[dashed] (1.5,1) -- (1.5,2);
    \draw[thick] (1.65,1) -- (1.65,2);
    \node at (1.65,.8) {$i_2$};
    \draw (0,1) rectangle ++(3,1);

    \draw[dashed] (1.15,.5) -- (1.65,1);
    \draw[dashed] (-0.5,.5) -- (0,1);
    \draw (-0.5,-.5) rectangle ++(1.65,1);
    \node at (0.3, -0.9) [] {\large $\cdots$};

    \draw[dashed] (3.5,.5) -- (3,1);
    \draw[dashed] (2.15,.5) -- (1.65,1);
    \draw (2.15,-.5) rectangle ++(1.35,1);
    \node at (2.85, -0.9) [] {\large $\cdots$};

    % Right child
    \fill[highlight] (5.8,1) rectangle ++(.4,1);
    \draw[dashed] (6,1) -- (6,2);
    \draw[thick] (5.8,1) -- (5.8,2);
    \node at (5.8,.8) {$i_3$};
    \draw (4.5,1) rectangle ++(3,1);

    \draw[dashed] (5.3,.5) -- (5.8,1);
    \draw[dashed] (4,.5) -- (4.5,1);
    \draw (4,-.5) rectangle ++(1.3,1);
    \node at (4.8, -0.9) [] {\large $\cdots$};

    \draw[dashed] (8,.5) -- (7.5,1);
    \draw[dashed] (6.3,.5) -- (5.8,1);
    \draw (6.3,-.5) rectangle ++(1.7,1);
    \node at (7.1, -0.9) [] {\large $\cdots$};

    \draw[<->, dashed] (-.6,-.9) -- (-.6,3.6) node[midway, right] {\rotatebox{-90}{$\log k$}};
\end{tikzpicture}
    %\vspace{-8.5mm}
    \caption{Illustration of the $\DnC$ algorithm for finding a near-optimal $\eps$-biased $k$-binning.}
    \label{fig:dnc}
    %\vspace{-5mm}
\end{figure}

\subsubsection{Step 1: the $\DnC$ algorithm for finding a near-optimal solution}
The key idea of our practical algorithm is to quickly identify $k-1$ boundaries that form a near-optimal $\eps$-biased $k$-binning, and use it to establish an upper bound for the local-search step.

We develop a divide-and-conquer algorithm $\DnC$ (Algorithm~\ref{alg:dnc}) that works by recursively dividing its search space.
Specifically, assuming the data is sorted on $x$, $\DnC(l,h,\kappa,\eps)$ partitions the tuples $\dee_{l,h}=\{t_{l+1},\cdots,t_h\}$ into {\em two $\eps$-biased super buckets}.
To do so, it initiates the boundary at the index of the $\lceil \frac{\kappa}{2}\rceil$-th boundary of 
a given $\kappa$-binning of $\dee_{l,h}$.
% the equal-size $\kappa$-binning of $\dee_{l,h}$. 
Particularly, for equal-size $\kappa$-binning, the $\lceil \frac{\kappa}{2}\rceil$-th boundary is
% That is, 
\(i=l+\lceil \frac{\kappa}{2}\rceil\times \frac{h-l}{\kappa}\).
The vertical dashed lines in Figure~\ref{fig:dnc} show the initial boundaries. The top row in the figure depicts the first step of the algorithm, where $\DnC(0,n,k,\eps)$ is divided into two super buckets.

If the two super buckets based on the initial boundary do not satisfy the $\eps$-bias requirement, the algorithm checks if varying $i$ by $\pm 1$ resolves the bias issue. If not, it keeps increasing the exploration radius (the highlighted yellow regions in Figure~\ref{fig:dnc}) until an $\eps$-biased boundary is discovered. Theorem~\ref{th:dnc} shows that if such a boundary does not exist, the problem is infeasible -- i.e., there is no $\eps$-biased $k$-binning for the given bias threshold.

\begin{algorithm}[t]
\caption{the $\DnC$ algorithm for finding a near-optimal $\eps$-biased $k$-binning}
\label{alg:dnc}
\begin{algorithmic}[1]
\Require The tuples $\{t_1,\cdots,t_n\}$ sorted on $x$, group-counts $cnt$.
\Ensure the bucket boundaries.
\Function{$\DnC$}{$l,h,\kappa,\eps, cnt$} \Comment{$cnt_{i,l} = |\Gee_l\cap \{t_{1},\cdots,t_i\}|$}
    \If{$\kappa=1$} 
        {\bf Return} $\{~\}$ 
    \EndIf
    \State $i\gets \lceil\frac{\kappa}{2}\rceil$-th boundary of the $\kappa$-binning of $\dee_{l,h}$ 
    \Comment{$l+ \lceil \frac{\kappa}{2}\rceil\times \frac{h-l}{\kappa}$ for equal-size binning}
    \State found$\gets$ true {\bf if} {\sc is-$\eps$-biased}$(l,i,h,cnt)$ {\bf else} false
    \State $j\gets 1$
    \While{found=false AND $j<\min(i-l,h-i)$}
        \If{{\sc is-$\eps$-biased}$(l,i+j,h,cnt)$}
            \State $i\gets i+j$; found$\gets$true
        \ElsIf{{\sc is-$\eps$-biased}$(l,i-j,h,cnt)$}
            \State $i\gets i-j$; found$\gets$true
        \Else $~j\gets j+1$
        \EndIf
    \EndWhile
    \If{found=false}
        {\bf Return} infeasible
    \EndIf
    \State $S_l\gets \DnC(l,i,\lceil \frac{\kappa}{2}\rceil,\eps,cnt)$
    \State $S_h\gets \DnC(i,h,\lfloor \frac{\kappa}{2}\rfloor,\eps,cnt)$
    \State {\bf Return} $S_l\cup \{i\}\cup S_h$
\EndFunction
\Statex
\Function{is-$\eps$-biased}{$l,i,h,cnt$}
    \State $e_l \gets \max\limits_{j\in[\ell]}\left\vert \frac{cnt_{i,j} - cnt_{l,j}}{i-l} - \frac{|\Gee_j|}{n} \right\vert$
    \State $e_h \gets \max\limits_{j\in[\ell]}\left\vert \frac{cnt_{h,j} - cnt_{i,j}}{h-i} - \frac{|\Gee_j|}{n} \right\vert$
    \State {\bf Return} $(e_l\leq \eps)$ AND $(e_h\leq \eps)$
\EndFunction
\end{algorithmic}
\end{algorithm}

As illustrated in Figure~\ref{fig:dnc}, after partitioning the search space $\dee_{l,h}$ into two super buckets, the algorithm recursively solves the subproblems $\DnC(l,i,\lceil \frac{\kappa}{2}\rceil,\eps)$ and $\DnC(i,h,\lfloor \frac{\kappa}{2}\rfloor,\eps)$ and continues as long as $\kappa>1$.

\begin{theorem}\label{th:dnc}
    The algorithm $\DnC$ (Algorithm~\ref{alg:dnc}) finds an $\eps$-biased binning unless none exists.
\end{theorem}

\begin{proof}
    We provide the proof by contradiction.
    Suppose there exists a valid $\eps$-biased binning $\mathcal{B}$ that partitions $\dee_{l,h}=\{t_{l+1},\cdots,t_h\}$ into $\kappa$ $\eps$-biased buckets, while the $\DnC(l,h,\kappa,\eps)$ could not find a boundary $i$ to partition $\dee_{l,h}$ into the two $\eps$-biased super-buckets.
    Note that $i$ represents the boundary of the  $\lceil \frac{\kappa}{2}\rceil$-th bucket.  
    Let $\{i_1,\cdots,i_{\kappa-1}\}$ be the boundary indices of the $\eps$-biased binning $\mathcal{B}$.
    Let $\mathsf{merge}(j)$ be the operation that merges the buckets $B_j$ and $B_{j+1}$ into the super bucket $B_{j..j+1}$ by removing the $j$-th boundary from $\mathcal{B}$.
    In the following, we first prove that if both $B_j$ and $B_{j+1}$ are $\eps$-biased, then $B_{j..j+1}$ is also $\eps$-biased.

    For every group $\gee_l$, let $r_l =  \frac{|\Gee_l|}{n}$. 
    Also, let 
    \[C_l(j,j') = |\Gee_l\cap \{t_{j+1},\cdots,t_{j'}\}|\]
    Finally, for a bucket $B_j\in\mathcal{B}$ and a group $\gee_l$, let $r_l(j) =  \frac{C_l(i_j,i_{j+1})}{i_{j+1}-i_j}$.

    Since $B_j$ and $B_{j+1}$ are $\eps$-biased, for every $l\in[\ell]$, we have
    \[
    |r_l(j)-r_l|\leq \eps \text{ and } |r_l(j+1)-r_l|\leq \eps
    \]
    Let us consider two cases separately. 
    
    {\em Case 1:} $r_l(j)\leq r_l$ and $r_l(j+1)\leq r_l$ (a similar argument holds for $r_l(j)> r_l$ and $r_l(j+1)> r_l$).
    In this case,
    \begin{align*}
        r_l - r_l(j) &= \eps_1\leq \eps \Rightarrow C_l(i_j,i_{j+1}) = (i_{j+1} - i_{j}) (r_l-\eps_1)\\
        r_l - r_l(j+1) &= \eps_2\leq \eps\Rightarrow C_l(i_{j+1},i_{j+2}) = (i_{j+2} - i_{j+1}) (r_l-\eps_2)\\
    \end{align*}

Then,
    \begin{align*}
        r_l(j..j+1)& = \frac{C_l(i_{j},i_{j+2})}{i_{j+2} - i_{j}}
        \\
        &= \frac{C_l(i_{j},i_{j+1})+C_l(i_{j+1},i_{j+2})}{i_{j+2} - i_{j}}\\
        &= \frac{(i_{j+1} - i_{j}) (r_l-\eps_1)+ (i_{j+2} - i_{j+1}) (r_l-\eps_2)}{i_{j+2} - i_{j}}
    \end{align*}
    Hence, 
    $r_l(j..j+1)$ is a weighted average of $r_l - \eps_1$ and $r_l - \eps_2$, so $r_l -
r_l(j..j+1)$ is a weighted average of $\eps_1$ and $\eps_2$ and is therefore at most $\max\{\eps_1, \eps_2\} <= \eps$.

    As a result, in Case 1, the super bucket $B_{j..j+1}$ is $\eps$-biased.

    \vspace{2mm}
    {\em Case 2:} $r_l(j)> r_l$ and $r_l(j+1)\leq r_l$ (the same argument holds for $r_l(j+1)> r_l$ and $r_l(j)\leq r_l$).
    Hence, 
    \[
        r_l(j) - r_l = \eps_1\leq \eps\quad,\quad
        r_l - r_l(j+1) = \eps_2\leq \eps
    \]
    % \begin{align*}
    %     r_l(j) - r_l &= \eps_1\leq \eps\\
    %     r_l - r_l(j+1) &= \eps_2\leq \eps
    % \end{align*}

    % %\vspace{-4mm}
    In this case, $B_j$ has a surplus $C^+$ of the group $\gee_l$, while $B_{j+1}$ has a shortage $C^-$ of $\gee_l$. 
    Specifically,
    \(C^+ = \eps_1(i_{j+1}-i_j)\) and \(C^- = \eps_2(i_{j+2}-i_{j+1})\).
    % \begin{align*}
    %     C^+ &= \eps_1(i_{j+1}-i_j)\\
    %     C^- &= \eps_2(i_{j+2}-i_{j+1})\\
    % \end{align*}
    If $C^+\geq C^-$, then \(r_l(j..j+1) - r_l \leq \eps_1\leq \eps\).
    This is because after merging the two buckets, the super bucket has a smaller surplus compared to $B_j$ (i.e., $C^+ - C^- < C^+$), while also being larger than $B_j$ (i.e., $i_{j+2}-i_j > i_{j+1}-i_j$).
    Similarly, if $C^+< C^-$, then
    \( r_l - r_l(j..j+1) \leq \eps_2\leq \eps\).
    As a result, in Case 2 also, the super bucket $B_{j..j+1}$ is $\eps$-biased.

    Having proven that merging two neighboring $\eps$-biased buckets results in an $\eps$-biased super bucket, we can now extend the merge operation to merging a chain of buckets $\langle B_i, B_{i+1},\cdots, B_j \rangle$.
    Let us first merge $B_i$ and $B_{i+1}$ to form the super bucket $B_{i..i+1}$. Since $B_i$ and $B_{i+1}$ are $\eps$-biased, the super bucket $B_{i..i+1}$ is $\eps$-biased. Next, merge this super bucket with $B_{i+2}$ to form the super bucket $B_{i..i+2}$, which again is $\eps$-biased as it involves merging two $\eps$-biased buckets $B_{i..i+1}$ and $B_{i+2}$. 
    We can continue this operation until the last step merges $B_{i..j-1}$ with $B_j$ to form the $\eps$-biased super bucket $B_{i..j}$. %Following the same approach $B[i..j]$ is also $\eps$-biased.

    Now, let us apply the merge operation on the first $\lceil \frac{\kappa}{2}\rceil$ buckets of $\mathcal{B}$ to form the super bucket $B_{1..\lceil \frac{\kappa}{2}\rceil}$. Similarly let us merge the last $\lfloor \frac{\kappa}{2}\rfloor$ buckets of $\mathcal{B}$ to form the super bucket $B_{\lceil \frac{\kappa}{2}\rceil+1..n}$.
    Since all buckets of $\mathcal{B}$ have a bias of at most $\eps$, the super buckets $B_{1..\lceil \frac{\kappa}{2}\rceil}$ and $B_{\lceil \frac{\kappa}{2}\rceil+1..n}$ are $\eps$-biased. This would be a valid solution for $\DnC(l,h,\kappa,\eps)$, which contradicts the assumption that there was no feasible solution for it.
    %\vspace{-2mm}
\end{proof}

\begin{lemma}\label{th:dnc:compexity}
    The $\DnC$ algorithm (Algorithm~\ref{alg:dnc}) has the time complexity of $O(n\log k)$.
\end{lemma}

\begin{proof}
    Consider the execution tree of the $\DnC$ algorithm (Figure~\ref{fig:dnc}). At every level, the number of buckets ($\kappa$) is divided by two. As a result, the depth of the tree is $\log k$.
    Since the algorithm keeps {\em partitioning} the search space, each index $i\in [n]$ is checked at most once for every level, while each checking is in constant time.
    As a result, the time complexity of the algorithm is $O(n \log k)$.
\end{proof}

The $\DnC$ algorithm can act as a practical heuristic to quickly identify a near-optimal solution for $\eps$-biased binning. In such a setting, considering the time to sort $\dee$ based on the attribute $x$, the total time complexity of the algorithm is $O(n\log n)$. 

%\vspace{-2mm}
\subsubsection{Step 2: Local Search}
According to Theorem~\ref{th:dnc}, the $\DnC$ algorithm finds a valid solution for the $\eps$-biased $k$-binning problem if one exists. However, it does not guarantee finding the optimal solution.
Therefore, the second step of the algorithm verifies the optimality by 
applying a local search around the boundaries of the initial binning, while using the objective value of the solution found by $\DnC$ to prune the search space.

Let $\mathcal{B}_{dnc}$ be the solution found by the $\DnC$ algorithm, and let
$w_{dnc} = \text{\sc obj}(\mathcal{B}_{dnc})$ be the objective value of $\DnC$ solution. In particular, for equal-size binning, $w_{dnc} = w^\uparrow_{dnc} - w^\downarrow_{dnc}$, where $w^\downarrow_{dnc}$ and $w^\uparrow_{dnc}$ are the minimum and maximum width of its buckets.

$w_{dnc}$ provides an upper-bound threshold for the local search since the optimal objective value $w_{opt}$ is at most equal to $w_{dnc}$, i.e., $w_{opt}\leq w_{dnc}$.
Therefore, any binning whose objective value is larger than $w_{dnc}$ can be safely pruned.

\begin{figure}[!tb]
    \centering
    \begin{tikzpicture}
    \definecolor{highlight}{RGB}{255,220,120}
    
    % Highlighted segments
    \fill[highlight] (1.8,0) rectangle ++(.4,1/3);
    \fill[highlight] (3.8,0) rectangle ++(.4,1/3);
    \fill[highlight] (5.8,0) rectangle ++(.4,1/3);

    % Cells
    \foreach \i in {0,...,3} {
        \draw (2*\i,0) rectangle ++(2,1/3);
    }

    % Arrow (adjusted height)
    \draw[<->] (1.8,0.5) -- (2.2,0.5);
    \node at (2., 0.7) {$w_{dnc}$};
\end{tikzpicture}
    %\vspace{-3mm}
    \caption{Illustration of the local-search around the equal-size buckets with the exploration-window width $w_{dnc}$.}
    \label{fig:ls}
    %%\vspace{-5mm}
\end{figure}

Following this observation, the $\LS$ algorithm \technicalreport{(Algorithm~\ref{alg:ls})} draws a window with width $w_{dnc}$ around each boundary of the equal-size binning (Figure~\ref{fig:ls}).
For each combination of the indices in the boundary windows (i.e., $\varprod_{\kappa=1}^{k-1}w_{dnc}$), the algorithm checks if the corresponding binning is (a) $\eps$-biased and (b) has a smaller objective value than the best-known solution; if so, it updates the best-known solution $\mathcal{B}_{opt}$.
\submit{\textcolor{blue}{The pseudocode of the $\LS$ algorithm is provided in \cite{techrep}.\marginpar{R1.M2}}}
\technicalreport{The pseudocode of the $\LS$ algorithm is provided Algorithm~\ref{alg:ls}.}

\technicalreport{
\begin{algorithm}[t]
\caption{Local-search based algorithm for $\eps$-biased Binning}
\label{alg:ls}
\begin{algorithmic}[1]
\Require The dataset $\dee$, the value $k$, and the value $\eps$.
\Ensure The optimal $\eps$-biased $k$-binning.
\Function{LS}{$\dee, k, \eps$}
    \State $\{t_1,\cdots,t_n\}\gets$ {\sc sort}$(\dee,x)$ \Comment{\small sort $\dee$ on $x$}
    \State $cnt\gets [0]_{n,\ell}$;
    % \State 
    $cnt[1,\gee(t_1)]\gets 1$
    \For{$i\gets 2$ to $n$}
        \State $cnt[i,:]\gets cnt[i-1,:]$
        \State $cnt[i,\gee(t_i)]\gets cnt[i,\gee(t_i)]+1$
    \EndFor
    \State $\mathcal{B}_{dnc}\gets \DnC(0,n,k,\eps,cnt)$
    \If{$\mathcal{B}_{dnc}=$infeasible}
        {\bf Return} infeasible
    \EndIf
    \State $w_{dnc}\gets$ {\sc obj}$(\mathcal{B}_{dnc})$
    \State $bnd_{base} \gets [\max(0,(\frac{i\times n}{k}-\frac{w}{2})$ {\bf for} $i\in[k-1]]$
    \State $\mathcal{B}_{opt}\gets \mathcal{B}_{dnc}$; $w_{opt}\gets w_{dnc}$ 
    \For{combination $\mathcal{C}$ in $\varprod_{\kappa=1}^{k-1}w_{dnc}$}
        \State $\mathcal{B}_{tmp}\gets [bnd_{base}[i]+\mathcal{C}[i]$ {\bf for} $i\in[k-1]]$
        \State $w_{tmp}\gets$ {\sc obj}$(\mathcal{B}_{tmp})$
        \If{$w_{tmp}<w_{opt}$ AND $\beta(\mathcal{B}_{tmp})\leq \eps$}
            \State $\mathcal{B}_{opt}\gets \mathcal{B}_{tmp}$; $w_{opt}\gets w_{tmp}$
        \EndIf
    \EndFor
    \State {\bf Return} $\mathcal{B}_{opt}$
\EndFunction
% \Statex
% \Function{obj}{$\mathcal{B}$}
%     \State $(w_{dnc}^\downarrow,w_{dnc}^\uparrow)\gets (\min,\max)_{B\in\mathcal{B}}|B\cap \dee|$
%     \State {\bf Return} $w_{dnc}^\uparrow - w_{dnc}^\downarrow$
% \EndFunction
% \EndFunction
\end{algorithmic}
\end{algorithm}
}

\begin{lemma}
\label{lem:ls:complexity}
    Let $w_{dnc}$ be the objective value of the solution found by $\DnC$.
    The $\LS$ algorithm \technicalreport{(Algorithm~\ref{alg:ls})} has a space complexity of $O(n)$ and a time complexity of
    \(
        O(n \log n + (w_{dnc})^{k-1}) \leq O(n \log n + n^{k-1})
    \).
\end{lemma}

\begin{proof}
    Since the algorithm needs to maintain the counts matrix $cnt$, it has a space complexity of $O(n)$, given that $\ell$ is a small constant. 
    Sorting $\dee$ on $x$ takes $O(n\log n)$, which dominates the $O(n \log k)$ time complexity of the $\DnC$ algorithm and the $O(n)$ time to compute the $cnt$ values.
    After finding the initial solution by $\DnC$, the 
    total number of combinations that the algorithm checks for finding the optimal solution is $(w_{dnc})^{k-1}$. Therefore, its time complexity is $O(n \log n + (w_{dnc})^{k-1}) \leq O(n^{k-1})$.
\end{proof}

While the time complexity of the $\LS$ algorithm in worst-case is exponential in $k-1$, as we shall observe in our experiments, it is efficient in practice, 
since the adversarial cases where $w_{dnc}$ is large are unlikely.
% except the adversarial cases where $w_{dnc}$ is large -- which usually is not the case since $\DnC$ finds a near-optimal solution. 
In such adversarial cases, the $\DnC$ algorithm can serve as a heuristic for quickly finding a practical solution for $\eps$-biased $k$-binning.

\section{Experiments}\label{sec:exp}
Having introduced our algorithms, in this section, we experimentally evaluate their performance on (a) real-world benchmark datasets and (b) synthetic datasets.

\subsection{Experiments Setup}\label{sec:exp:setup}
The algorithms are implemented using {\sc Python 3.9.6}, and are publicly available on GitHub.
The experiments were conducted on a personal computer with an {\sc Apple M2 Pro} CPU with 12 cores and 32 GB of {\sc DDR5} memory, running {\sc MAC OS 14.5}.

\stitle{Real Datasets}
We use two benchmark datasets in our evaluations:
\begin{enumerate}[leftmargin=*]
    \item {\bf German Credit} dataset\footnote{\url{https://www.kaggle.com/datasets/uciml/german-credit}}: is a widely used benchmark for fairness and credit-scoring models. It contains information on 1,000 individuals, each described by attributes such as credit history, employment, housing, and personal status. This dataset commonly evaluates issues such as {\em gender biases} in credit scoring evaluation.
    We use the attribute \at{\bf Sex} (\at{male}, \at{female}) for grouping, and the attribute \at{\bf Credit amount} for binning. 
    \item {\bf COMPAS} dataset\footnote{\url{https://www.kaggle.com/datasets/danofer/compass}}: includes the records of individuals assessed for their risk of recidivism by the COMPAS tool.
    It is widely used in the literature to examine and mitigate bias in criminal justice algorithms, especially regarding disparities across {\em racial groups}.
    The dataset contains around 6800 entries with features such as age, race, gender, prior offenses, and risk scores. We use the attribute \at{\bf Race} (\at{Caucasian}, \at{African-American}, \at{Hispanic}, \at{Other}) for grouping and the attribute \at{\bf RawScore} for binning.
\end{enumerate}

\stitle{Synthetic Datasets ({\sc syn})}
In addition to the benchmark datasets, we also used synthetic datasets to control and study the impact of data distribution parameters on the algorithms' performance and scalability.
We generate the values of the attribute $x$ for a group $\gee_l\in \Gee$, according to the Gaussian distribution $dist_l = \mathcal{N}(\mu_l,\sigma_l)$, with the mean $\mu_l$ and the standard deviation $\sigma_l$.
We also control the overall ratio of each group $\gee_l$ in the dataset with the probability $p_l$. That is, the probability that a random tuple $t_i\in\dee$ belongs to group $\gee_l$ is $p_l$.
For each setting, specified by the values $\langle n,\ell,\{\mu_l\}^{\ell},\{\sigma_l\}^{\ell},\{p_l\}^{\ell}\rangle$, we generate 30 synthetic datasets, using different random seeds.

% \vspace{-3mm}
\subsection{Validation}\label{exp:validation}
Before evaluating the performance of our algorithms, we use the \at{German Credit} benchmark in this section to illustrate how attribute binning affects both the performance and fairness of downstream models trained for loan approval.

As further discussed in Section~\ref{sec:discussion}, it is important to reiterate that unbiased binning is only a fairness-aware component of the data preparation stage within the broader data analytics pipeline~\cite{jagadish2014big}. By itself, it cannot guarantee fairness in downstream models or analyses. Unbiased binning does not replace the FairML techniques required during model development; rather, it serves as a complementary step that must be used in conjunction with them.

\begin{table}[tb!]
\centering
\small
\caption{Logistic Regression: Performance and Fairness Metrics Before and After Unbiased Binning.}
\label{tab:exp:validation:lr}
\resizebox{\linewidth}{!}{
\setlength{\tabcolsep}{2.5pt} % default is ~6pt
\begin{tabular}{@{}lccccc@{}}
\hline
\textbf{Metric} & \textbf{Baseline} & \multicolumn{3}{c}{\textbf{$\eps$-biased Binning}} & \textbf{Improvement} \\
 & \textbf{Binning} & \textcolor{black}{$0.8I$} & \textcolor{black}{$0.6I$} & $\sim 0.5I$ & \\
\hline
\hline
\textbf{Accuracy} & 0.67 & \textcolor{black}{0.7033} & \textcolor{black}{0.70} & 0.70 & \textcolor{green}{$\uparrow$} Slightly improved \\
\textbf{F1-score} & 0.25 & \textcolor{black}{0.2261} & \textcolor{black}{0.2623} & 0.18 & \textcolor{red}{$\downarrow$} Slightly decreased \\
\hline
\textbf{DI}  & 0.9378 & \textcolor{black}{0.9804} & \textcolor{black}{0.9832} & \textbf{0.9949} & \textcolor{green}{$\checkmark$} Very close to 1 (ideal) \\
\textbf{SPD} & -0.0545 & \textcolor{black}{-0.0180} & \textcolor{black}{-0.0151} & \textbf{-0.0048} & \textcolor{green}{$\checkmark$} Almost perfect parity \\
\textbf{EOD} & -0.0794 & \textcolor{black}{-0.0195} & \textcolor{black}{-0.0215} & \textbf{+0.0116} & \textcolor{green}{$\checkmark$} Major Improvement \\
\textbf{AOD} & -0.0341 & \textcolor{black}{-0.0124} & \textcolor{black}{-0.0113} & \textbf{-0.0050} & \textcolor{green}{$\checkmark$} Almost perfect parity \\
\hline
\textcolor{black}{\bf INDV} & \textcolor{black}{0.974} & \textcolor{black}{0.973} & \textcolor{black}{0.968} & \textcolor{black}{0.958} & \textcolor{red}{$\downarrow$} \textcolor{black}{Slightly decreased} \\
\hline
\multicolumn{6}{p{8.5cm}}{\textcolor{black}{\scriptsize $I=[0.25, 0.093, 0.1]$ is the initial Bias of Baseline (Equal-size) Binning}} \\
\multicolumn{6}{p{8.5cm}}{{\scriptsize {\bf DI}: Disparate Impact; {\bf SPD}: Statistical Parity Difference;}} \\
\multicolumn{6}{p{8.5cm}}{{\scriptsize {\bf EOD}: Equal Opportunity Difference; {\bf AOD}: Average Odds Difference.}} \\
\multicolumn{6}{p{8.5cm}}{\textcolor{black}{\scriptsize {\bf INDV}: Consistency (Individual Fairness).}}\\
\hline
\end{tabular}
}
% \vspace{-4mm}
\end{table}

\begin{table}[tb!]
\centering
\small
\caption{SVM: Performance and Fairness Metrics Before and After Unbiased Binning.}
\label{tab:exp:validation:svm}
% \vspace{-4mm}
\resizebox{\linewidth}{!}{
\setlength{\tabcolsep}{2.5pt} % default is ~6pt
\begin{tabular}{@{}lccccc@{}}
\hline
\textbf{Metric} & \textbf{Baseline} & \multicolumn{3}{c}{\textbf{$\eps$-biased Binning}} & \textbf{Improvement} \\
 & \textbf{Binning} & \textcolor{black}{$0.8I$} & \textcolor{black}{$0.6I$} & $\sim 0.5I$ & \\
\hline
\hline
\textbf{Accuracy} & 0.66 & \textcolor{black}{0.7033} & \textcolor{black}{0.70} & 0.70 & \textcolor{green}{$\uparrow$} Slightly improved \\
\textbf{F1-score} & 0.21 & \textcolor{black}{0.2261} & \textcolor{black}{0.2742} & 0.18 & \textcolor{red}{$\downarrow$} Slightly decreased \\
\hline
\textbf{DI}  & 0.9563 & \textcolor{black}{0.9804} & \textcolor{black}{0.9940} & \textbf{0.9949} & \textcolor{green}{$\checkmark$} Very close to 1 (ideal) \\
\textbf{SPD} & -0.0386 & \textcolor{black}{-0.0180} & \textcolor{black}{-0.0054} & \textbf{-0.0048} & \textcolor{green}{$\checkmark$} Almost perfect parity \\
\textbf{EOD} & -0.0794 & \textcolor{black}{-0.0195} & \textcolor{black}{-0.0145} & \textbf{+0.0116} & \textcolor{green}{$\checkmark$} Major Improvement \\
\textbf{AOD} & -0.0126 & \textcolor{black}{-0.0124} & \textcolor{black}{0.0002} & \textbf{-0.0050} & \textcolor{green}{$\checkmark$} Almost perfect parity \\
\hline
\textcolor{black}{\bf INDV} & \textcolor{black}{0.974} & \textcolor{black}{0.973} & \textcolor{black}{0.962} & \textcolor{black}{0.962} & \textcolor{red}{$\downarrow$} \textcolor{black}{Slightly decreased} \\
\hline
\end{tabular}
}
\end{table}

We considered a scenario in which the attributes \at{age}, \at{Credit amount}, and \at{Duration} (credit duration) were each partitioned into six equal-sized bins. The resulting bucketized attributes exhibited bias levels of 0.25, 0.093, and 0.1, respectively. 
Using these bucketized attributes, we trained two machine learning models -- \at{logistic regression} and \at{SVM} -- to predict the binary attribute \at{Risk}, indicating whether a loan application is considered high risk. Both models achieved reasonable performance in terms of accuracy and F1 score.

We then evaluated the fairness of the models using IBM’s AIF360 toolkit~\cite{bellamy2019ai}, applying several standard fairness metrics and treating \at{sex: female} as the unprivileged group. The corresponding results appear in the first column (``Baseline Binning'') of Tables~\ref{tab:exp:validation:lr} and~\ref{tab:exp:validation:svm}.

Next, we applied $\eps$-biased binning, selecting values of $\eps$ to generate datasets with varying levels of bias.
In particular, we constructed datasets corresponding to 80\%, 60\%, and 50\% of the original bias to systematically evaluate the impact of bias reduction. 
We retrained both models using the newly derived attributes at each bias level and reevaluated their performance and fairness using the same metrics.
The results, reported in Tables~\ref{tab:exp:validation:lr} and~\ref{tab:exp:validation:svm}, show that reducing bias consistently improves fairness metrics across all evaluated models, while accuracy and F1 scores remain largely stable. 
Interestingly, in both models, accuracy increased slightly, although there was a modest price in terms of F1 scores. 
More importantly, reducing the bias levels by half, 
\color{black}
substantially {\em improved all fairness metrics}: DI (Disparate Impact), SPD (Statistical Parity Difference), and AOD (Average Odds Difference) all moved toward their ideal values (1 for DI; 0 for SPD and AOD), while EOD (Equalized Odds Difference) significantly improved from -8\% to +1\% in both models.
These results demonstrate that although fairness-aware data preparation techniques, including unbiased binning, cannot guarantee fairness in downstream tasks, they can nonetheless be highly effective in practice, significantly reducing model unfairness.

\noindent\textbf{Individual fairness.} %\marginpar{R3.D3; Meta3}
In addition to group fairness metrics, we evaluate individual fairness using the notion of consistency, as implemented in AIF360. The results are reported in the last row of Tables~\ref{tab:exp:validation:lr} and~\ref{tab:exp:validation:svm}. We observe that while group fairness metrics improve significantly as bias is reduced, consistency remains largely stable across all settings, with only a negligible decrease.
This suggests that unbiased binning improves group-level fairness without substantially affecting individual fairness.

\begin{table}[tb!]
\centering
\small
\caption{\textcolor{black}{Fair logistic regression using Exponentiated Gradient~\cite{agarwal2019fair}, with $\eps$-biased Binning as input. Fair-opt shows the fairness metric used for optimization.}}
\label{tab:exp:validation:fairML}
\begin{tabular}{@{}l|cc|cccc@{}}
\hline
\textcolor{black}{\textbf{Fair-opt}} & \textcolor{black}{\textbf{Accuracy}} & \textcolor{black}{\textbf{F1-score}} & \textcolor{black}{\textbf{DI}} & \textcolor{black}{\textbf{SPD}} & \textcolor{black}{\textbf{EOD}} & \textcolor{black}{\textbf{AOD}} \\
\hline
\textcolor{black}{\textbf{DP}} & \textcolor{black}{0.7033} & \textcolor{black}{0.2261} & \textcolor{black}{0.9914} & \textcolor{black}{-0.0079} & \textcolor{black}{0.0075} & \textcolor{black}{-0.0054} \\
\textcolor{black}{\textbf{EO}} & \textcolor{black}{0.70} & \textcolor{black}{0.2373} & \textcolor{black}{0.9895} & \textcolor{black}{-0.0095} & \textcolor{black}{-0.0270} & \textcolor{black}{0.0102} \\
\hline
\end{tabular}
\end{table}

\stitle{Integration with fairness-aware learning} %\marginpar{R3.D2; Meta3}
To evaluate the synergy between our data preparation approach and downstream fairness interventions, we combined $\eps$-biased binning with Fair logistic regression using Exponentiated Gradient~\cite{agarwal2019fair}, implemented in AIF360. This method is a popular in-process intervention for FairML, and supports multiple fairness objectives, including demographic parity (DP) and equalized odds (EO).
The results are reported in Table~\ref{tab:exp:validation:fairML} for the dataset generated using $\eps$-biased binning with approximately 50\% of the initial bias. The two rows correspond to fair regression under DP and EO constraints, respectively. We observe that fair regression can be effectively applied on top of $\eps$-biased binning and can further reduce unfairness. At the same time, the gains were not large because, after $\eps$-biased binning, even standard models trained on the $\eps$-biased representation were already near-perfectly fair, leaving limited room for further improvement.

\color{black}

\begin{figure*}[!tb]
\centering
    \begin{subfigure}[t]{0.24\linewidth}
        \includegraphics[width=\linewidth]{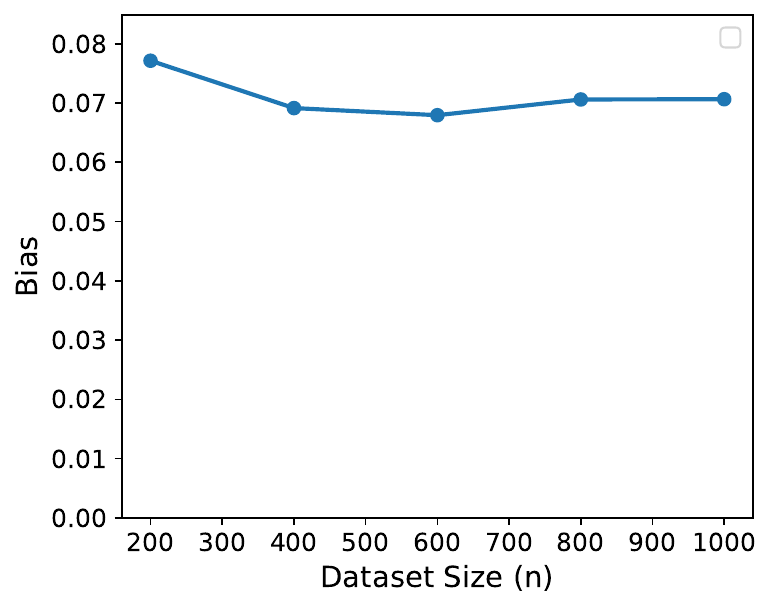}
        %\vspace{-6mm}
        \caption{Bias of equal-size binning}
        \label{fig:gc:bias-n}
        %\vspace{-1.5mm}
    \end{subfigure}
    \hfill
    \begin{subfigure}[t]{0.24\linewidth}
        \includegraphics[width=\linewidth]{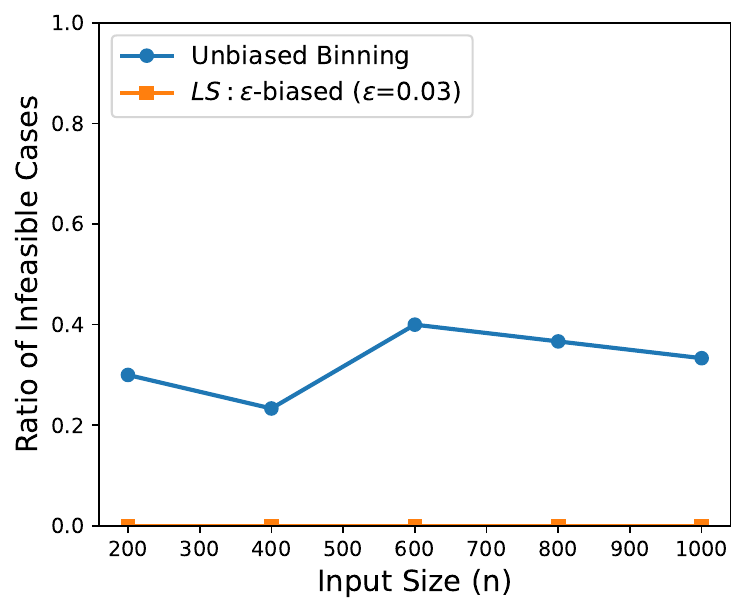}
        %\vspace{-6mm}
        \caption{Infeasible ratio vs Input size}
        \label{fig:gc:invalid-n}
        %\vspace{-1.5mm}
    \end{subfigure}
    \hfill
    \begin{subfigure}[t]{0.24\linewidth}
        \includegraphics[width=\linewidth]{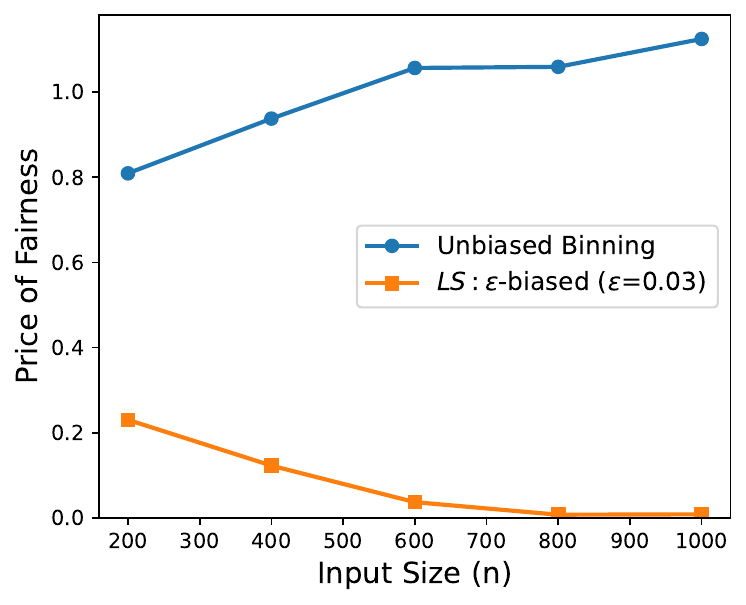}
        %\vspace{-6mm}
        \caption{PoF vs Input size}
        \label{fig:gc:PoF-n}
        %\vspace{-1.5mm}
    \end{subfigure}
    \hfill
    \begin{subfigure}[t]{0.24\linewidth}
        \includegraphics[width=\linewidth]{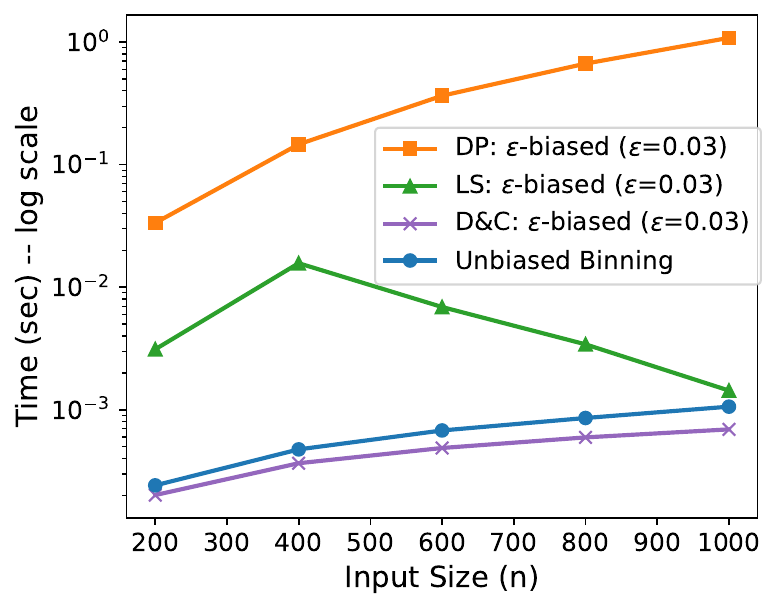}
        %\vspace{-6mm}
        \caption{Time vs Input size}
        \label{fig:gc:time-n}
        %\vspace{-1.5mm}
    \end{subfigure}

    \begin{subfigure}[t]{0.24\linewidth}
        \includegraphics[width=\linewidth]{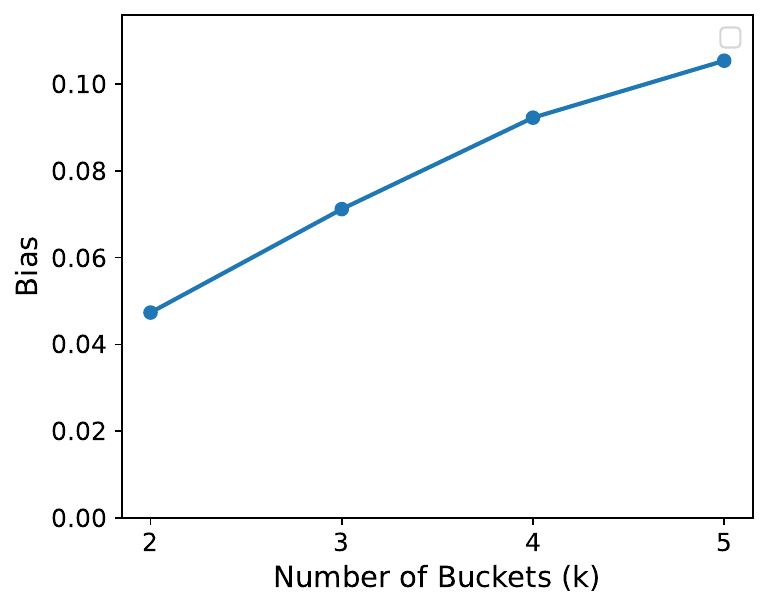}
        %\vspace{-6mm}
        \caption{Bias of equal-size binning}
        \label{fig:gc:bias-k}
        %\vspace{-1.5mm}
    \end{subfigure}
    \hfill
    \begin{subfigure}[t]{0.24\linewidth}
        \includegraphics[width=\linewidth]{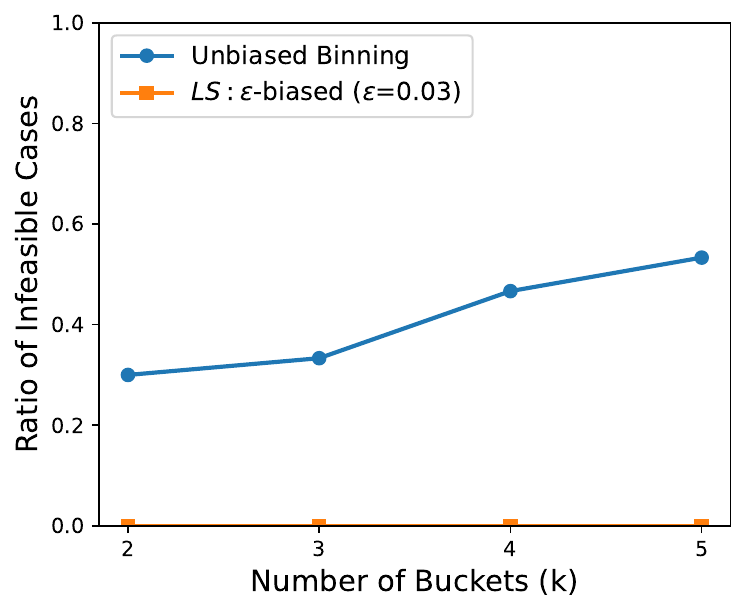}
        %\vspace{-6mm}
        \caption{Infeasible ratio vs \# buckets}
        \label{fig:gc:invalid-k}
        %\vspace{-1.5mm}
    \end{subfigure}
    \hfill
    \begin{subfigure}[t]{0.24\linewidth}
        \includegraphics[width=\linewidth]{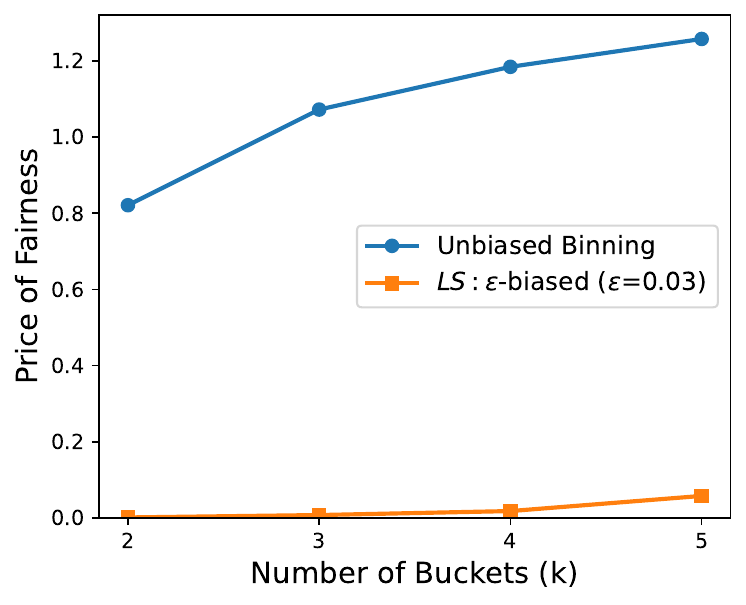}
        %\vspace{-6mm}
        \caption{PoF vs \# buckets}
        \label{fig:gc:PoF-k}
        %\vspace{-1.5mm}
    \end{subfigure}
    \hfill
    \begin{subfigure}[t]{0.24\linewidth}
        \includegraphics[width=\linewidth]{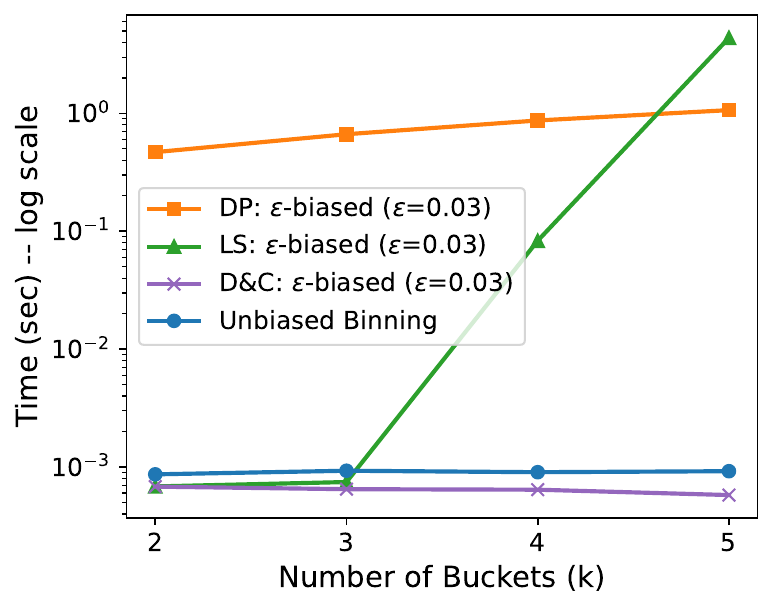}
        %\vspace{-6mm}
        \caption{Time vs \# buckets}
        \label{fig:gc:time-k}
        %\vspace{-1.5mm}
    \end{subfigure}
%\vspace{-3mm}
\caption{Experiment Results on German Credit dataset.}
\label{fig:gc}
%\vspace{-4mm}
\end{figure*}

\begin{figure}[!tb]
\centering
    \begin{subfigure}[t]{0.48\linewidth}
        \includegraphics[width=\linewidth]{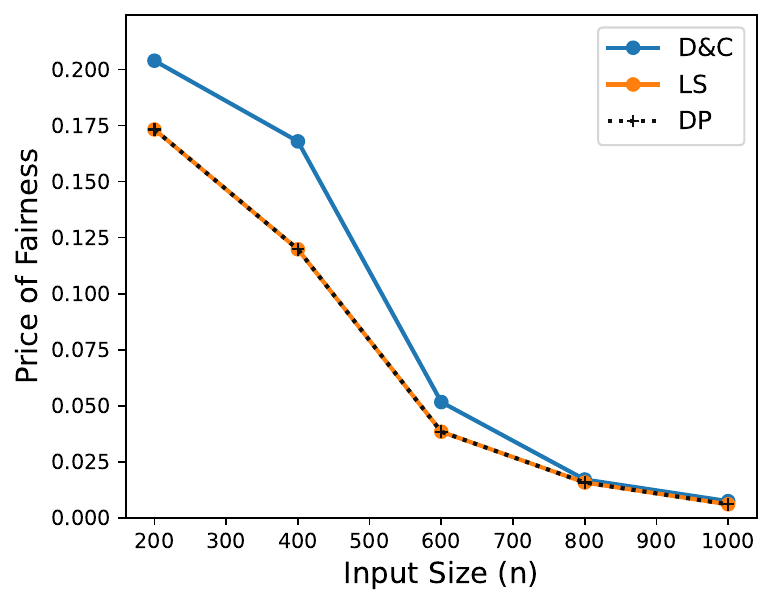}
        %\vspace{-6mm}
        \caption{\textcolor{black}{Varying $n$}}
        \label{fig:gc-dncvsls-n}
        %\vspace{-1.5mm}
    \end{subfigure}
    \hfill
    \begin{subfigure}[t]{0.48\linewidth}
        \includegraphics[width=\linewidth]{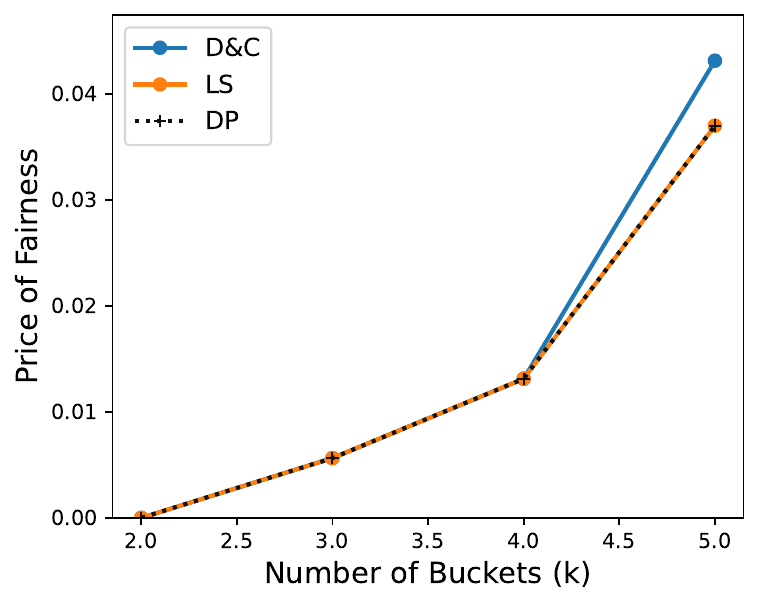}
        %\vspace{-6mm}
        \caption{\textcolor{black}{Varying $k$}}
        \label{fig:gc-dncvsls-k}
        %\vspace{-1.5mm}
    \end{subfigure}
%\vspace{-3mm}
\caption{\textcolor{black}{Near-optimality of $\DnC$ (German Credit)}}
\label{fig:exp:dnc}
%\vspace{-5mm}
\end{figure}  

%\vspace{-3mm}
\subsection{Performance Evaluation}
% \stitle{Evaluation metrics}
In this section, we evaluate the performance of our proposed algorithms across various settings using both real and synthetic datasets.
The evaluation values are presented as an average of multiple runs for each setting.
We measure the efficiency of algorithms using (i) {\em execution time}. For $\eps$-biased binning, we consider our quadratic-time dynamic programming solution $\DP$ as the baseline.
The bias $\beta_\dee(\mathcal{B})$ of a bucketization $\mathcal{B}$ is measured using Equation~\ref{eq:bias} as its (ii) {\em maximum bias} for the groups $\{\gee_1,\cdots,\gee_\ell\}$.
It may not always be possible to generate an unbiased (or $\eps$-biased) $k$-binning for a given setting. Repeating each setting multiple times on randomly selected data, we use the (iii) {\em ratio of infeasible} settings to measure infeasibility likelihood for that setting.
Last but not least, we measure (iv) the {\em price of fairness (PoF)}. 
For each bin $B_j\in \mathcal{B}$, the PoF is its relative size difference from the original bins. Specifically, for 
equal-size binning %(i.e., $n/k$) -- i.e., 
$PoF(B_j) = \left\vert 1- \frac{|B_j|}{n/k}\right\vert$.
The PoF of a $k$-binning $\mathcal{B}$ is the average of the PoF of its bins. That is,
\(PoF(\mathcal{B}) = \frac{1}{k} \sum_{j=1}^k PoF(B_j)\).
The results in this section are presented separately for real and synthetic datasets.

%\vspace{-3mm}
\subsubsection{Benchmark Datasets}\label{sec:exp:validation}
We begin this section by evaluating the performance of our algorithms on real-world benchmark datasets, COMPAS and German Credit.
The results were consistent across the two datasets. \submit{Thus, due to the space constraints, we provide the results on COMPAS in the technical report~\cite{techrep}.}
\technicalreport{The results on COMPAS are provided in the Appendix.}

% %\vspace{-2mm}
The experiment results %on the German Credit dataset 
for various input sizes ($n$) and numbers of buckets ($k$) are presented in Figure~\ref{fig:gc}, using the default parameters $n=800$, $k=3$, and $\eps=0.03$.
First, Figures~\ref{fig:gc:bias-n} and \ref{fig:gc:bias-k} verify that (fairness-unaware) equal-size bucketization of raw credit values failed to generate unbiased binning in various settings. Specifically, {\em the (gender) bias was close to 8\% for various input sizes, and it exceeded 10\% as the number of buckets increased to $k=5$ credit score buckets}.
Next, we applied Algorithm~\ref{alg:step2} to find the unbiased bucketization for each setting.
However, as reflected in Figures~\ref{fig:gc:invalid-n} and \ref{fig:gc:invalid-k}, in approximately one-third of the cases, a fully unbiased binning did not exist. This number increased to nearly 60\% as the number of buckets was increased. This is because the number of boundary candidates $m$ was less than 5 in many cases, while in cases where $k> m$, the problem is infeasible. Similarly, Figures~\ref{fig:gc:PoF-n} and \ref{fig:gc:PoF-k} show a high price of fairness (PoF), sometimes more than 1.0, for unbiased binning.
On the other hand, {\em allowing a small bias threshold ($\eps=0.03$), the PoF dropped to close to zero in most cases}, while there was no infeasible case for any of the $\eps$-biased settings. 
Still, when $n$ is small, the PoF is higher. This can be explained based on the central-limit theorem~\cite{fischer2011history}, as aggregates (ratios) over a larger number of samples (larger buckets) have a smaller variance around the mean (overall ratios). 
The same reasoning explains the small increase in the PoF for the larger value of $k$. 

Figures~\ref{fig:gc:time-n} and \ref{fig:gc:time-k} show the average running time of each algorithm for each setting. 
Among the optimal algorithms, unbiased binning was the fastest, with the running time of the algorithm in all cases dominated by the time required to sort the data on attribute $x$.
On the other hand, the $\DP$ algorithm for $\eps$-biased binning was around $10^3$ times slower than the unbiased binning algorithm in all cases. This verifies the average-case quadratic time complexity of the $\DP$ algorithm, compared to the near-linear time complexity of the unbiased-binning algorithm.
The $\LS$ algorithm, although not theoretically efficient in worst-case scenarios, performed well in practice, as its running times were typically orders of magnitude less than those of the $\DP$ algorithm in most cases.
Yet, its running time significantly increases when its exploration-window width (see Figure~\ref{fig:ls}) and the value of $k$ are not small.
This is observed in Figure~\ref{fig:gc:time-k}, where $\LS$ takes longer than the $\DP$ algorithm for $k=5$.
Nevertheless, the solution found by the $\DnC$ algorithm can serve as a practical alternative in these cases. 

$\DnC$ was the fastest algorithm across all experiments.
Although it does not guarantee optimality, our experiments show that it consistently finds either the optimal or a near-optimal solution. 
This is demonstrated in Figure~\ref{fig:exp:dnc}, where either $\DnC$ found the optimal solution, or the PoF of its solution is close to the optimal.
% \color{blue} \marginpar{R3.D4; Meta3}
Besides, to experimentally confirm the optimality of $\DP$ and $\LS$, we included $\DP$ in the figure, verifying that both algorithms found optimal solutions with the same objective values and PoF.
% \color{black}

\begin{figure}[!tb]
\centering
    \begin{subfigure}[t]{0.48\linewidth}
        \includegraphics[width=\linewidth]{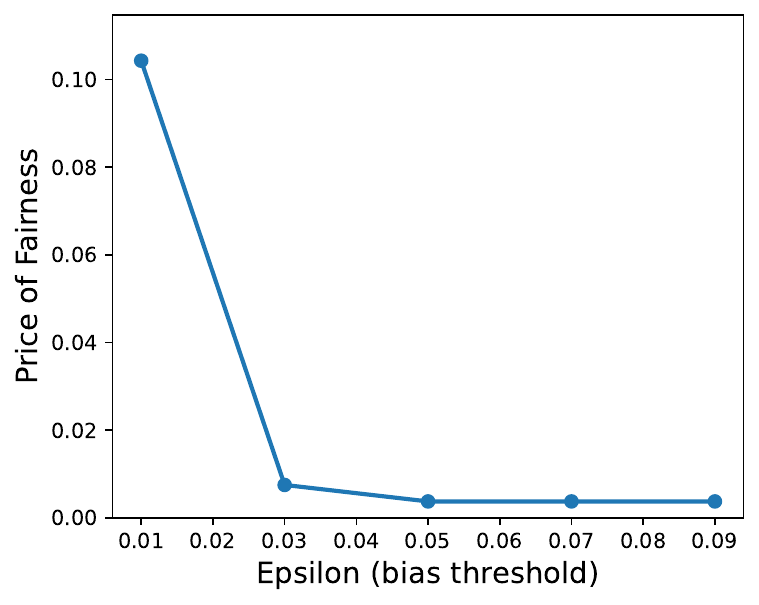}
        % %\vspace{-3mm}
        \caption{German Credit: PoF vs $\eps$}
        \label{fig:gc:PoF-eps}
    \end{subfigure}
    \hfill
    \begin{subfigure}[t]{0.48\linewidth}
        \includegraphics[width=\linewidth]{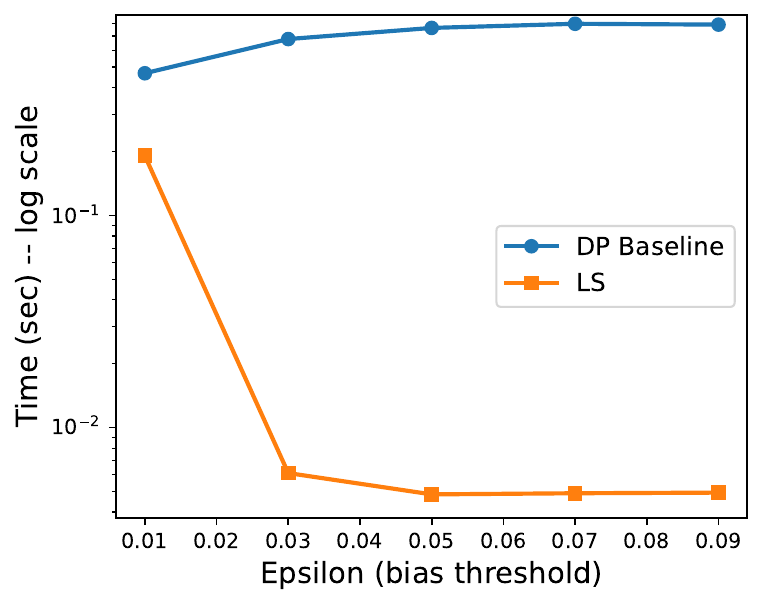}
        % %\vspace{-3mm}
        \caption{German Credit: Time vs $\eps$}
        \label{fig:gc:time-eps}
    \end{subfigure}
%\vspace{-4mm}
\caption{The effect of varying the bias threshold ($\eps$) on time and PoF (German Credit)}
\label{fig:exp:eps}
%\vspace{-4mm}
\end{figure}   

\begin{figure}[!tb]
\centering
\includegraphics[width=.48\linewidth]{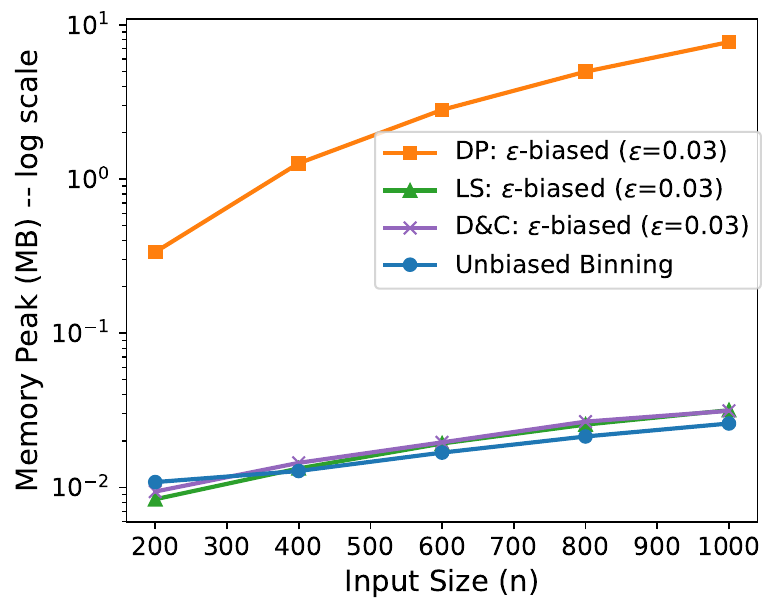}
        %%\vspace{-4mm}
        \caption{\textcolor{black}{German Credit: Space vs Input size}}
        \label{fig:gc-space-n}
%%\vspace{-6mm}
\end{figure}  

Next,
we compared the performance of the $\DP$ and $\LS$ algorithms for different values of bias threshold. The results are provided in Figure~\ref{fig:exp:eps}. 
First, as reflected in Figure~\ref{fig:gc:PoF-eps}, while the price of fairness is large when $\eps$ is 1\%, it suddenly drops for larger values of $\eps$. This confirms that by slightly increasing the max-bias threshold, one can find a binning that is very close to equal-size binning.
Next, from Figure~\ref{fig:gc:time-eps}, we can confirm the different running-time behaviors of the $\DP$ and $\LS$ algorithms when $\eps$ increases.
Increasing the value of $\eps$ slightly increases the running time of $\DP$ as it increases the valid boundaries. On the other hand, while the running time of the $\LS$ algorithm is high for $\eps=0.01$, it quickly decreases for larger values of $\eps$. This major drop can be explained by the significant drop in the PoF when $\eps$ increases. When the PoF is low, the near-optimal solution found by the $\DnC$ algorithm is small, hence the $\LS$ can stop after checking a few combinations.

% \color{blue}
\stitle{Memory Peak (Space)} %\marginpar{R1.D6; Meta3}
Finally, by varying the input size, we evaluate the peak memory usage of each method. The results in Figure~\ref{fig:gc-space-n} are consistent with the space complexities reported in Table~\ref{tab:results}. Specifically, the Unbiased Binning algorithm, along with the $\DnC$ and $\LS$ algorithms for $\eps$-biased binning, have linear space complexity. Consequently, their memory footprint grows only modestly with the input size. In contrast, the $\DP$ algorithm for $\eps$-biased binning has quadratic space complexity, leading to significantly higher memory consumption compared to the other methods.

% \color{black}

\stitle{Entropy-based Binning}
Even though we use equal-size binning for performance evaluation, our local search ($\LS$) and divide-and-conquer ($\DnC$) algorithms are not restricted to this setting and can be applied on top of other binning approaches. To illustrate this, we additionally evaluated the $\DnC$ and $\LS$ algorithms 
% on alternative binning methods: equal-depth binning.
using an entropy-based binning method that maximizes the correlation between the bucketized attribute and a target variable.
Specifically, our goal is to maximize the information gain with respect to a target variable (\at{Risk} in German Credit) by minimizing the entropy of $y$ values in each bucket.
A penalty term is also added to prevent the creation of small buckets.
We then apply the Local search algorithm to reduce the binning bias to $\eps=0.03$.
Figure~\ref{fig:entropy-time} shows the running times of $\DnC$ and $\LS$ for different input sizes. As expected, having a linearithmic time complexity, the $\DnC$ algorithm was efficient for various values of $n$. The running time of $\LS$, however, increases as the width of its exploration window increases.

\begin{figure*}[!tb] 
\centering
    \begin{minipage}[t]{0.24\linewidth}
        \centering
        \includegraphics[width=\textwidth]{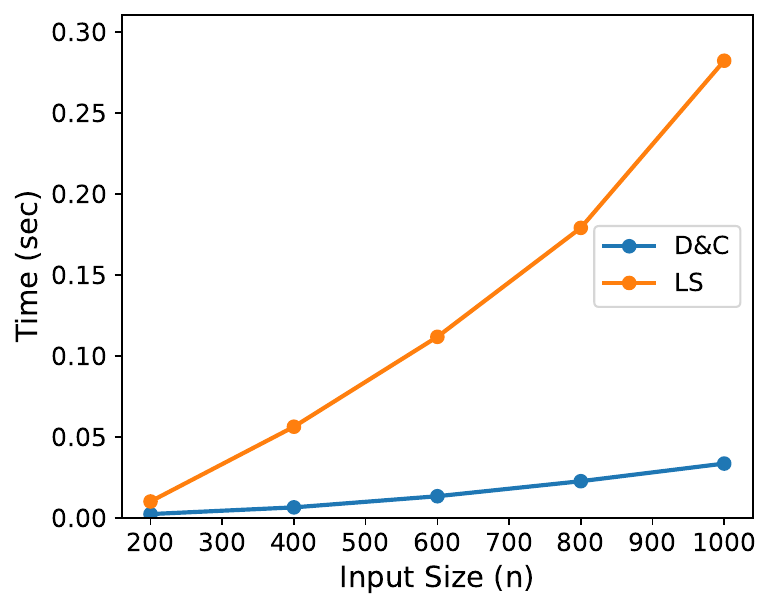}
        %%\vspace{-8mm}
        \caption{Entropy-based Binning (German Credit)}
        % %%\vspace{-2.5em}
        \label{fig:entropy-time}
        %%\vspace{-1.5mm}
    \end{minipage}
    \hfill
    \begin{minipage}[t]{0.24\linewidth}
        \centering
        \includegraphics[width=\textwidth]{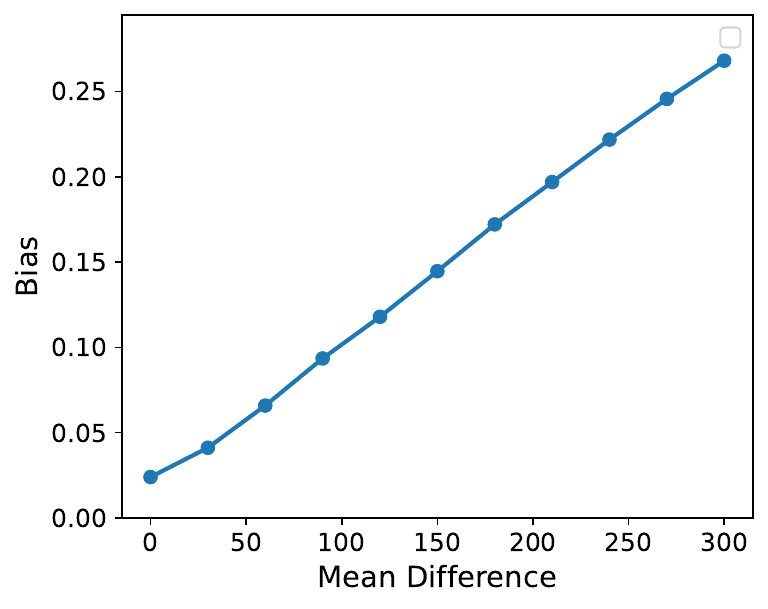}
        %%\vspace{-8mm}
        \caption{{\sc syn}: Bias v.s. Mean difference}
        % %%\vspace{-2.5em}
        \label{fig:bias-mean}
        %%\vspace{-1.5mm}
    \end{minipage}
    \hfill
    \begin{minipage}[t]{0.24\linewidth}
        \centering
        \includegraphics[width=\textwidth]{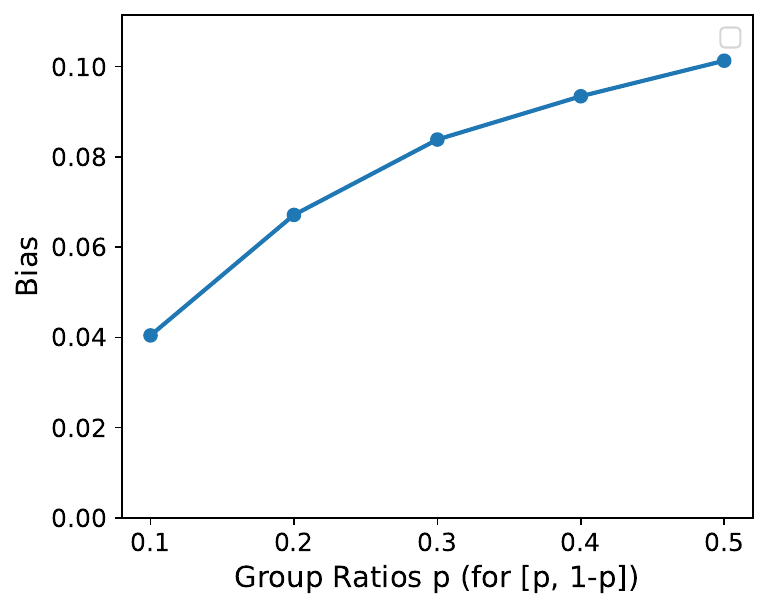}
        %%\vspace{-8mm}
        \caption{{\sc syn}: Bias v.s. group ratios}
        \label{fig:bias-p}
        %%\vspace{-1.5mm}
    \end{minipage}
    \hfill
    \begin{minipage}[t]{0.24\linewidth}
        \centering
        \includegraphics[width=\textwidth]{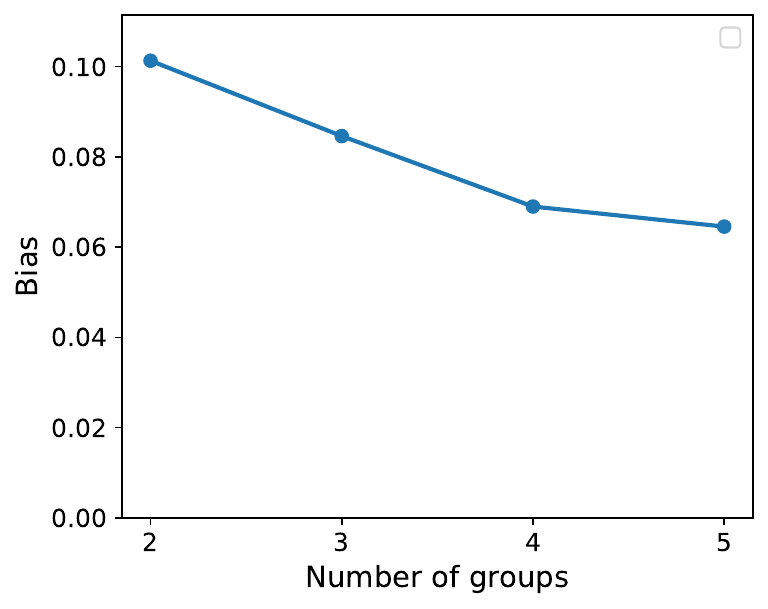}
        %%\vspace{-8mm}
        \caption{{\sc syn}: Bias v.s. number of groups}
        \label{fig:bias-ell}
        %%\vspace{-1.5mm}
    \end{minipage}
    
    \begin{minipage}[t]{0.24\linewidth}
        \centering
        \includegraphics[width=\textwidth]{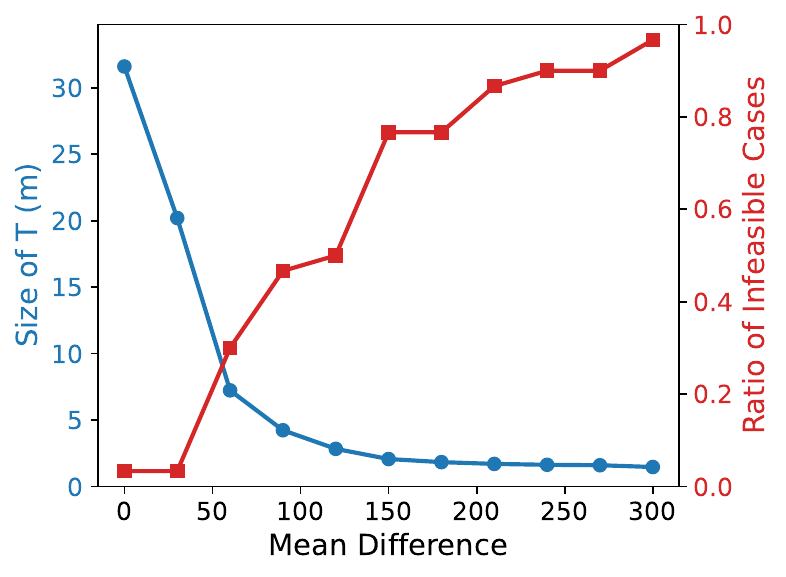}
        %%\vspace{-8mm}
        \caption{{\sc syn}: \# boundary candidates v.s. Mean difference}
        % %\vspace{-2.5em}
        \label{fig:m-mean}
        % %\vspace{-1mm}
    %\vspace{-3mm}
    \end{minipage}
    \hfill
    \begin{minipage}[t]{0.24\linewidth}
        \centering
        \includegraphics[width=\textwidth]{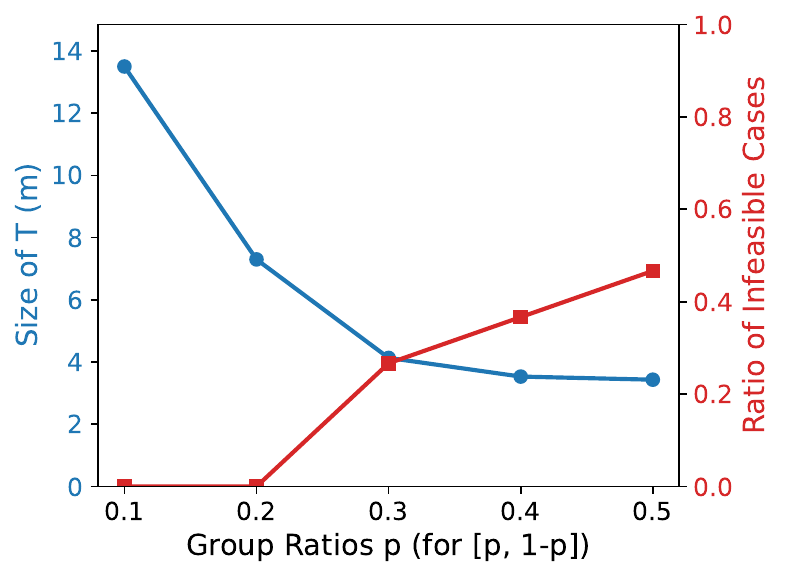}
        %\vspace{-8mm}
        \caption{{\sc syn}: \# boundary candidates v.s. group ratios}
        % %\vspace{-2.5em}
        \label{fig:m-p}
        % %\vspace{-1mm}
    \end{minipage}
    \hfill
    \begin{minipage}[t]{0.24\linewidth}
        \centering
        \includegraphics[width=\textwidth]{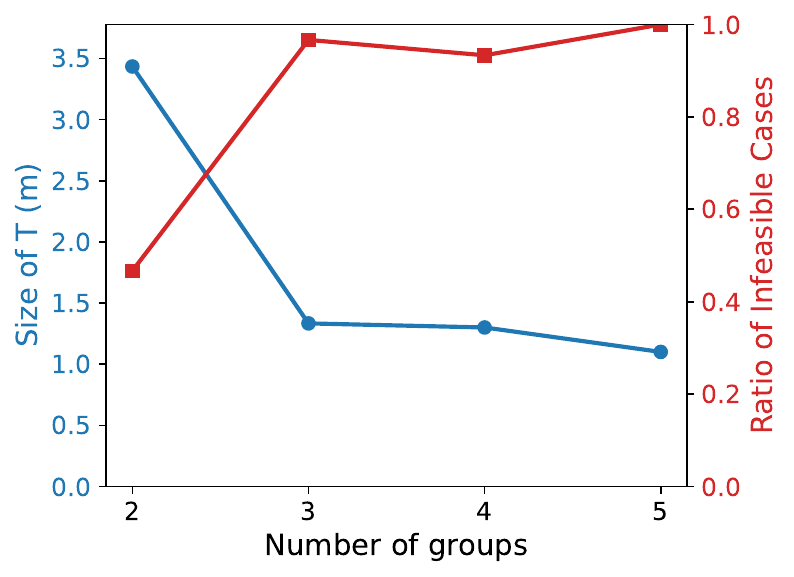}
        %\vspace{-8mm}
        \caption{{\sc syn}: \# boundary candidates v.s. \# groups}
        % %\vspace{-2.5em}
        \label{fig:m-ell}
        % %\vspace{-1mm}
    \end{minipage}
    \hfill
    \begin{minipage}[t]{0.24\linewidth}
        \centering
        \includegraphics[width=\textwidth]{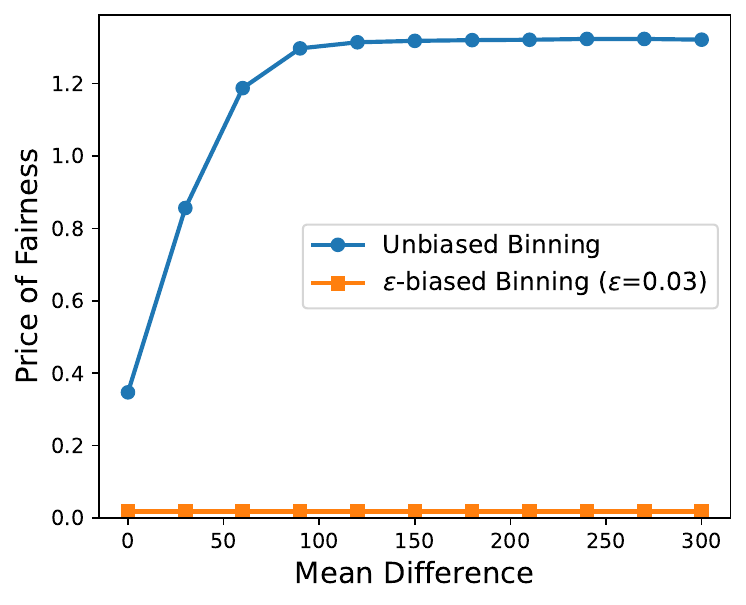}
        %\vspace{-8mm}
        \caption{{\sc syn}: PoF v.s. Mean difference}
        % %\vspace{-2.5em}
        \label{fig:pof-mean}
        % %\vspace{-1mm}
    \end{minipage}
    
    \begin{minipage}[t]{0.24\linewidth}
        \centering
        \includegraphics[width=\textwidth]{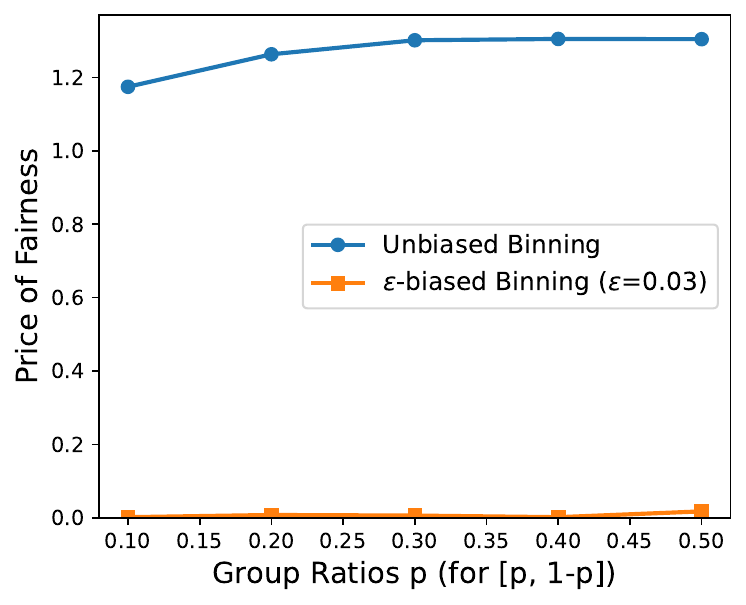}
        %\vspace{-8mm}
        \caption{{\sc syn}: PoF v.s. group ratios}
        % %\vspace{-2.5em}
        \label{fig:pof-p}
        %\vspace{-3.5mm}
    \end{minipage}
    \hfill
    \begin{minipage}[t]{0.24\linewidth}
        \centering
        \includegraphics[width=\textwidth]{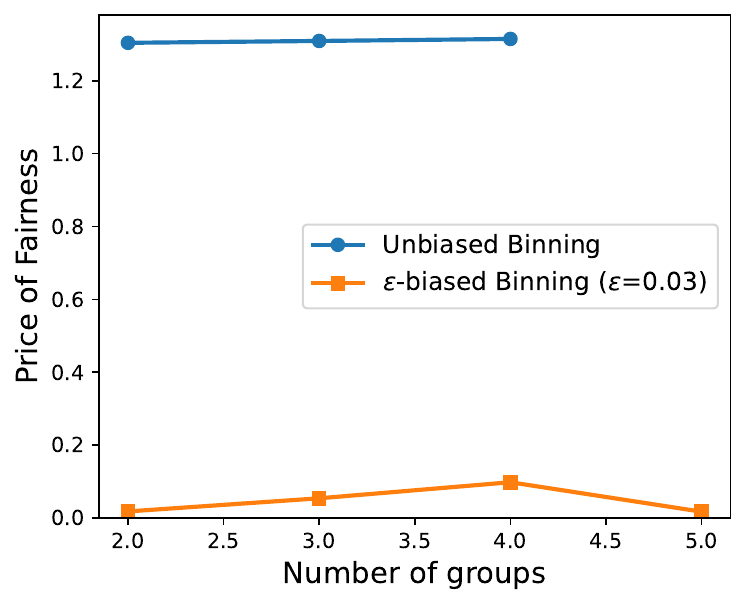}
        %\vspace{-8mm}
        \caption{{\sc syn}: PoF v.s. \# groups}
        % %\vspace{-2.5em}
        \label{fig:pof-ell}
        %\vspace{-3.5mm}
    \end{minipage}
    \hfill
    \begin{minipage}[t]{0.24\linewidth}
        \centering
        \includegraphics[width=\textwidth]{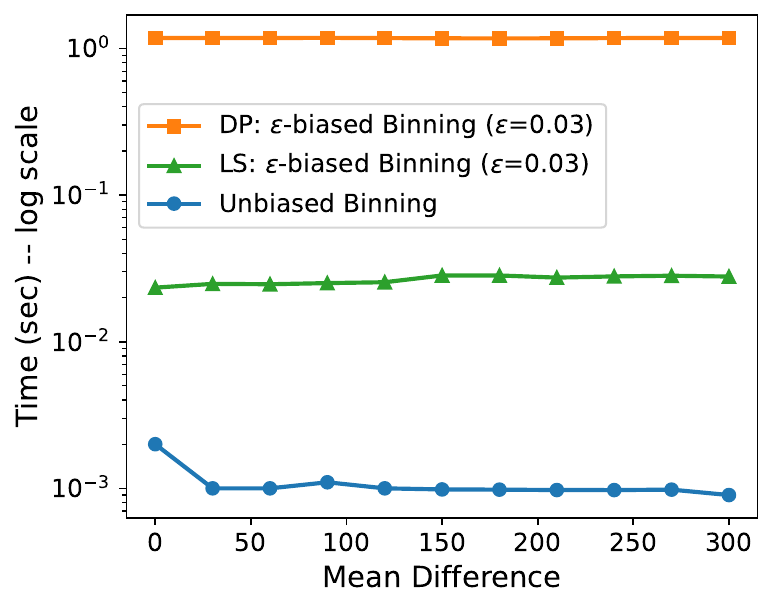}
        %\vspace{-8mm}
        \caption{{\sc syn}: Time v.s. Mean difference}
        % %\vspace{-2.5em}
        \label{fig:time-mean}
        %\vspace{-3.5mm}
    \end{minipage}
    \hfill
    \begin{minipage}[t]{0.24\linewidth}
        \centering
        \includegraphics[width=\textwidth]{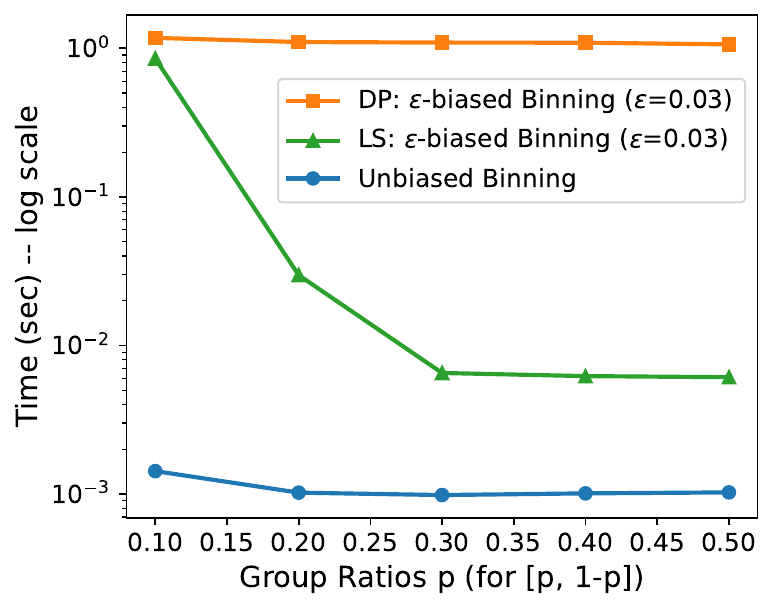}
        %\vspace{-8mm}
        \caption{{\sc syn}: Time v.s. group ratios}
        % %\vspace{-2.5em}
        \label{fig:time-p}
        %\vspace{-3.5mm}
    \end{minipage}
    % %\vspace{-2mm}
\end{figure*}

\begin{figure*}[!tb] 
\centering
    \begin{minipage}[t]{0.24\linewidth}
        \centering
        \includegraphics[width=\textwidth]{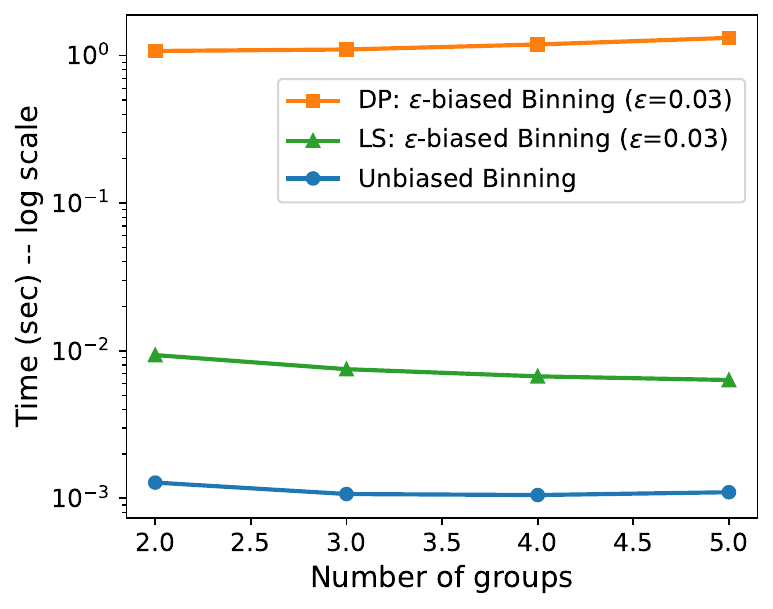}
        %\vspace{-8mm}
        \caption{{\sc syn}: Time v.s. \# groups}
        % %\vspace{-2.5em}
        \label{fig:time-ell}
        %\vspace{-1.5mm}
    \end{minipage}
    \hfill
    \begin{minipage}[t]{0.24\linewidth}
        \centering
        \includegraphics[width=\textwidth]{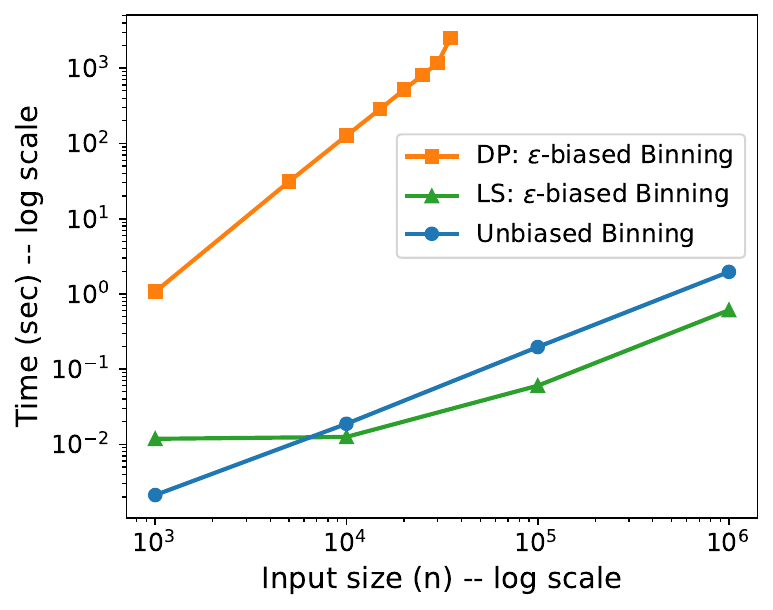}
        %\vspace{-8mm}
        \caption{{\sc syn}: Scalability evaluation.}
        \label{fig:scalability}
        %\vspace{-1.5mm}
    \end{minipage}
    \hfill
    \begin{minipage}[t]{0.24\linewidth}
        \centering
        \includegraphics[width=\textwidth]{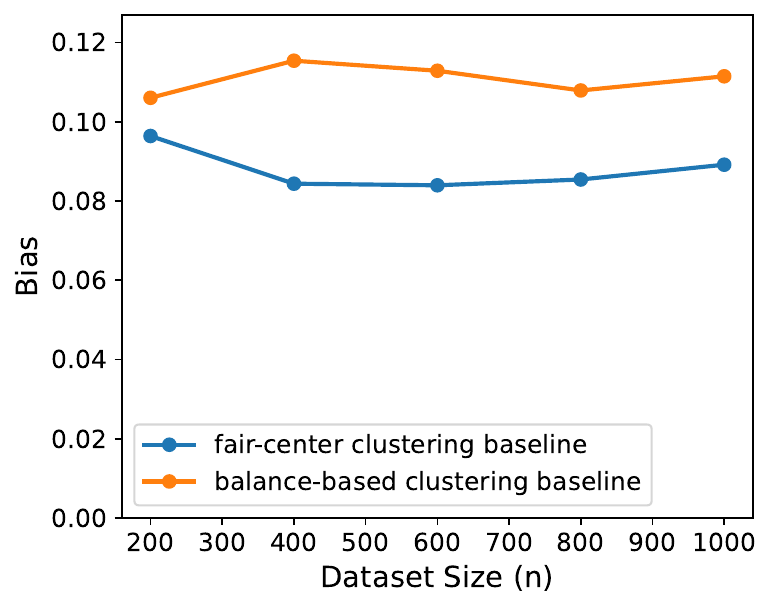}
        %\vspace{-8mm}
        \caption{\textcolor{black}{{\sc gc}: Fair Clustering Baselines: Bias of Binning}}
        % %\vspace{-2.5em}
        \label{fig:fairclustering1}
        %\vspace{-1.5mm}
    \end{minipage}
    \hfill
    \begin{minipage}[t]{0.24\linewidth}
        \centering
        \includegraphics[width=\textwidth]{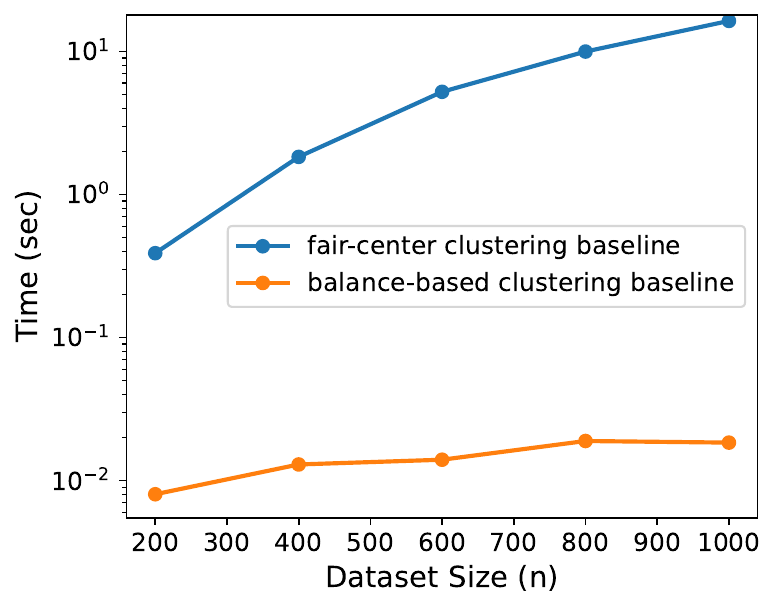}
        %\vspace{-8mm}
        \caption{\textcolor{black}{{\sc gc}: Fair Clustering Baselines: Run time.}}
        \label{fig:fairclustering2}
        %\vspace{-1.5mm}
    \end{minipage}
    %\vspace{-4mm}
\end{figure*}

%\vspace{-3mm}
\subsubsection{Synthetic Datasets}\label{sec:exp:synthetic}
After evaluating the performance of our algorithms on real-world benchmark datasets, we use synthetic datasets with tunable parameters in this section to study the algorithms' performance under various dataset properties.
To do so, we fix the other parameters to their default values and vary one of the parameters, measuring the average performance metrics over 30 runs. To be consistent with our real-world experiments, we set $n=1000$, $k=3$, and $\eps=0.03$. For the odd (resp. even) values of $l\in \ell$, the $x$ values are generated using $dist_l=\mathcal{N}(1000+50,300)$ (resp. $dist_{l'}=\mathcal{N}(1000-50,300)$). Thus, the default value for the distribution-mean difference is $100$. The choices of the distribution means are motivated by the real-world income disparities. The default number of groups is $\ell=2$, and we use equal group ratios (i.e., equal expected number of samples from each group) as the default setting.
The performance results for varying the input size ($n$), number of buckets ($k$), and bias threshold $\eps$ were similar to the real-world datasets. Therefore, in the following, we focus on varying (C1) the distribution-mean difference between the groups, (C2) the (expected) group ratios, and (C3) the number of groups.

\stitle{Bias}
Figures~\ref{fig:bias-mean}, \ref{fig:bias-p}, and \ref{fig:bias-ell} show the results for (C1), (C2), and (C3), respectively.
Notably, when the values are generated using the same distribution (i.e., the mean difference is 0), the bias is close to zero. However, the bias increases (almost linearly) up to 27\% as the mean differences increase.

%\vspace{-1mm}
\stitle{Existence of an unbiased binning}
Next, in Figures~\ref{fig:m-mean}, \ref{fig:m-p}, and \ref{fig:m-ell}, we evaluate the impact of varying distribution mean, group ratios, and number of groups on the likelihood of the existence of an unbiased binning.
In each plot, the left-y-axis shows the number of boundary candidates ($m$) for unbiased binning (see Section~\ref{sec:unbiased:step1}), while the right-y-axis is the ratio of infeasible cases (out of 30 runs) for each setting.
Evidently, when the distribution-mean difference between the groups increases, the number of boundary candidates sharply reduces, which causes a rapid increase in the ratios of infeasible cases. In particular, for the mean difference of 300, only one out of the 30 cases had a valid unbiased binning. Similarly, as the group ratios approach 50\%, and for the larger number of groups, the number of boundary candidates decreases, resulting in an increase in the ratio of infeasible cases. Yet, as also observed in our experiments on COMPAS dataset (Appendix~\ref{app:exp}), increasing the value of $\ell$ (number of groups) notably decreases the probability that a bucket is unbiased -- i.e., it has equal ratios for all groups. As a result, for $\ell>2$, the ratio of infeasible cases increased to nearly 100\%.

% %\vspace{-1mm}
\stitle{Price of Fairness}
Figures~\ref{fig:pof-mean}, \ref{fig:pof-p}, and \ref{fig:pof-ell} show the average PoF for varying distribution-mean difference, group ratios, and number of groups. As expected, the PoF of unbiased binning quickly increases when the gap between the distribution means increases. Similarly, the PoF of unbiased binning was significant for all cases, regardless of the group ratios and the number of groups.
Still, similar to our previous experiments, allowing a small bias threshold (3\%), the PoF dropped close to zero in almost all cases.

% %\vspace{-1mm}
\stitle{Running time}
Next, Figures~\ref{fig:time-mean}, \ref{fig:time-p}, and \ref{fig:time-ell} show the running time of our algorithms for various cases. The time complexity of the $\DP$ algorithm for $\eps$-biased binning does not depend on parameters such as value distributions, group ratios, or the number of groups. This is reflected in our experiments, as its running time was almost the same in all cases. Similarly, the running time of the local-search algorithm $\LS$ did not notably change when varying the mean differences or the number of groups. However, when the ratio of the minority group was 10\%, its running time was almost the same as the $\DP$ algorithm. For such cases, the objective value $w_{dnc}$ (see Figure~\ref{fig:ls}) of the solution found by the $\DnC$ algorithm was larger. This caused a major increase in the number of combinations $\LS$ had to check, hence a significant increase in its running time. The unbiased binning algorithm has a negligible running time across all experiments.

\stitle{Scalability}
Lastly, we evaluated the scalability of our algorithms by exponentially increasing the input size from 1,000 to 1 million.
The running times are provided in Figure~\ref{fig:scalability}.
The $\DP$ algorithm for $\eps$-biased binning did not scale beyond an order of 10K. In particular, for $n=30K$, it took close to 1 hour on average to finish.
In contrast, both the unbiased binning algorithm and the local-search-based algorithm for $\eps$-biased binning could scale to very large input sizes, taking less than 1 second in almost all cases.

\subsubsection{Fair Clustering Baselines}\label{sec:exp:baseline}
We evaluated two representative fair clustering baselines: (i) a fair-center clustering method based on the LS1 heuristic\footnote{\textcolor{black}{The LS1 method performs local search by iteratively swapping centers with points from the same group to improve the clustering cost while maintaining feasibility.}} of~\cite{thejaswi2022clustering}, and (ii) a balance-based clustering approach using fairlet decomposition followed by k-means~\cite{chierichetti2017fair}. 
Since these methods do not produce ordered partitions, we converted their outputs into bins by sorting cluster centers and defining boundaries via midpoints. Figures~\ref{fig:fairclustering1} and~\ref{fig:fairclustering2} report the bias and running time as a function of the dataset size, using the German Credit dataset and averaged over 30 runs. Comparing Figure~\ref{fig:fairclustering1} with Figure~\ref{fig:gc:bias-n}, neither baseline is able to effectively reduce the bias compared to fairness-unaware binning, indicating that enforcing fairness at the clustering level does not translate to unbiased bins. Moreover, the running time results show that Baseline (i) incurs significantly higher computational cost and does not scale with the dataset size. These results show that, because of their fundamentally different objectives, fair clustering methods are not well-suited to addressing the unbiased binning problem.
\color{black}
% \vspace{-3mm}
\section{Discussions}\label{sec:discussion}

\begin{algorithm}[t]
\caption{Binary search for selecting the value of $\eps$}
\label{alg:BSearch}
\begin{algorithmic}[1]
\Require The dataset $\dee$, the value $k$.
\Ensure The $\eps$-biased $k$-binning by automatically selecting $\eps$.
\Function{BinarySearch-Bin}{$\dee, k$}
    \State $\mathcal{B}_{init}\gets$ $k$-binning($\dee, k$)
    \State $\eps_{init}\gets \beta_\dee(\mathcal{B}_{init})$
    \Comment{\small the bias on the initial binning}
    \State $PoF_{l}\gets$ PoF of unbiased binning %(UnbiasedBinning($\dee, k$))
    % \Comment{\small the price of fairness of unbiased binning}
    \State $(\eps_l,\eps_r)\gets (0,\eps_{init});$ $\mathcal{B}_r\gets \mathcal{B}_{init};$ PoF$_r\gets 0$
    \While{$\eps_r - \eps_l\geq \frac{\eps_{init}}{n}$}
        \State $\eps_m \gets \frac{\eps_l + \eps_r}{2}$
        \State $\mathcal{B}_m\gets \eps$-binning($\dee, k,\eps$);
        % \State 
        PoF$_m\gets$ PoF$(\mathcal{B}_m)$
        % \If{PoF$_l-$ PoF$_m>$ PoF$_m-$ PoF$_r$}
        %     \Comment {\small move left}
        %     \State $(\mathcal{B}_r, \eps_r)\gets (\mathcal{B}_m, \eps_m)$
        % \Else \Comment {\small move right}
        %     \State $(\mathcal{B}_l, \eps_l)\gets (\mathcal{B}_m, \eps_m)$
        % \EndIf
        \State {\bf if} PoF$_l-$ PoF$_m>$ PoF$_m-$ PoF$_r$ {\bf then} $(\mathcal{B}_r, \eps_r)\gets (\mathcal{B}_m, \eps_m)$
        \State {\bf else} $(\mathcal{B}_l, \eps_l)\gets (\mathcal{B}_m, \eps_m)$
    \EndWhile
    \State {\bf Return} $\mathcal{B}_{r},\eps_r$
\EndFunction
\end{algorithmic}
\end{algorithm}

% \vspace{-1mm}
\subsection{Tuning the value of $\eps$}
Given an input parameter $\eps\in [0,1]$, $\eps$-biased Binning aims at finding the optimal binning whose maximum bias does not exceed $\eps$. The choice of the bias threshold can be guided by existing regulatory standards, such as the disparate impact law~\cite{barocas2016big}.

Alternatively, the bias threshold may be selected by examining the {\em Price of Fairness} (PoF) tradeoff. Our experiments on both benchmark and synthetic datasets show that although the PoF can be large when $\eps \to 0$, it decreases rapidly as $\eps$ increases slightly (Figure~\ref{fig:exp:eps}). This observation suggests that a reasonable balance between bias and PoF can often be achieved by selecting the knee of the $\eps$–PoF curve (e.g., the knee in Figure~\ref{fig:gc:PoF-eps} occurs around $\eps = 0.03$).

To implement this selection strategy, starting from the initial bias $\eps_{\text{init}}$ of the binning, one can perform a {\bf binary search} over the interval $[0, \eps_{\text{init}}]$ 
\textcolor{black}{to identify the smallest value of $\eps$ that yields a feasible solution without incurring a significant increase in PoF}.
% to identify the smallest value of $\eps$ that does not lead to a significant increase in PoF. 
More concretely, by discretizing the search interval into $\frac{\eps_{\text{init}}}{n}$ values, the binary search runs in $O(\log n)$ time. Thus, the overall runtime of $\eps$-biased binning increases only by a logarithmic factor, which is negligible in practice.

Algorithm~\ref{alg:BSearch} presents the pseudocode for this binary-search–based procedure for selecting $\eps$.

% \vspace{-3mm}
\subsection{Problem Scope and Limitations}\label{sec:limitations}
Attribute binning is applied before sharing a dataset, enabling its use in a variety of data-driven tasks and analytical workflows. While unbiased binning is fairness-aware -- i.e., it ensures that the binning procedure does not itself introduce biased bucketized attributes -- it {\em does not guarantee fairness in downstream tasks}.
Rather, unbiased binning complements algorithmic interventions aimed at developing fair models and data-driven systems~\cite{barocas2023fairness}; it cannot replace them. In other words, unbiased binning is {\em necessary but not sufficient}, and must be used {\em in conjunction} with other fairness-aware data preparation and intervention techniques.
The following limitations of unbiased binning further underscore that it cannot substitute for algorithmic fairness methods:

\begin{enumerate}[leftmargin=*]
    \item {\em Higher-order correlations:} 
    Unbiased binning ensures that each bucketized attribute is individually uncorrelated with demographic groups. However, this does not preclude the presence of higher-order correlations~\cite{draper1998applied}, where a combination of (individually unbiased) attributes has a joint association with the demographic groups. \textcolor{black}{As a result, dependencies that arise from combinations of attributes may still persist after binning.}
    \item {\em Task-unawareness:}
    Unbiased binning is task-agnostic, which is both an advantage and a limitation. On one hand, because the binning is not tailored to a specific task, the resulting dataset remains broadly applicable to diverse data-driven systems and analyses. On the other hand, its task-unaware nature means it cannot guarantee that all downstream analytics performed on the bucketized attributes will themselves be unbiased.
\end{enumerate}

\textcolor{black}{Extending unbiased binning to higher-order correlations and task-aware binning remains an important direction for future work.}

% \color{blue}
% \vspace{-3mm}
\subsection{Fairness Objectives and Generality}
\label{sec:fairness-objectives}
% \marginpar{R2.W1; Meta2}
Our formulation focuses on removing statistical dependence between the bucketized attribute and the sensitive (grouping) attribute.
As a data curation step, unbiased binning operates upstream in the data analytics pipeline and is complementary to downstream fairness-aware modeling techniques.

By eliminating such dependence at the representation level, unbiased binning can facilitate the satisfaction of multiple group-based fairness objectives in downstream tasks.
This observation is consistent with the ``impossibility theorems''~\cite{barocas2023fairness}, stating that if the input data is biased -- i.e., if attributes are statistically dependent on a sensitive attribute -- then it is impossible to simultaneously satisfy all fairness criteria. By construction, unbiased binning removes such dependence between the attributes and the sensitive attributes. Consequently, models trained on unbiased attributes may be able to satisfy multiple fairness metrics simultaneously, as our experiments in Section~\ref{exp:validation} illustrate.

At the same time, our approach has inherent limitations, as discussed in Section~\ref{sec:limitations}.
Moreover, addressing biases at the group level, it does not contribute to improving individual fairness notions.
As a result, in our experiments in Section~\ref{exp:validation}, we observed that, while significantly improving group fairness of the downstream models, unbiased binning did not improve individual fairness. 

% \color{black}
% \vspace{-3mm}
\section{Related Work}\label{sec:relatedwork}
% \vspace{-1mm}

\stitle{Fair ML}
Seminal studies have shown that ML models often inherit or amplify biases present in training data, resulting in disparate impacts across demographic groups~\cite{barocas2016big}. 
In response, numerous fairness definitions have been introduced, and a 
broad range of pre-processing, in-processing, and post-processing interventions have been developed to support the training of fair models~\cite{barocas2023fairness,mehrabi2021survey}.
Unbiased binning operates at a different stage of the data analytics pipeline and must therefore be used in conjunction with FairML techniques. Unbiased binning replaces raw attribute values with unbiased bucketized values to prepare a dataset for sharing; thus, the dataset is the {\em output} of the binning process.
In contrast, FairML methods take a (potentially biased) dataset as their {\em input}, and aim to produce fair models. By reducing biases in the shared dataset upstream, unbiased binning can facilitate the ability of downstream FairML techniques to train models that satisfy multiple fairness definitions simultaneously, with minimal performance degradation.

Among FairML approaches, pre-process interventions are technically the closest to unbiased binning. Pre-process approaches such as disparate impact remover~\cite{feldman2015certifying} and fair representation learning~\cite{xu2023fair,calmon2017optimized} {\em change} attribute values in the dataset -- e.g., by synthetically increasing an individual’s income -- in order to decorrelate attributes from demographic groups.
Consequently, the modified values are no longer strictly factual.
In contrast, unbiased binning bucketizes the raw attribute values without altering them.

\stitle{Responsible Data Management}   
While the major focus of the literature has been on FairML, there is a growing understanding that upstream data management operations can be a critical source of bias~\cite{jagadish2022many}.
In response, the research under the umbrella of responsible data management~\cite{stoyanovich2020responsible} has been dedicated to identifying and mitigating bias across the data analytics pipeline.
Some of the work in this area includes causal data repair techniques for fairness~\cite{salimi2020database,salimi2019interventional}, 
identifying and tuning problematic data slices~\cite{tae2021slice},
fairness-aware query adjustment~\cite{roy2024fairness,zhu2023consistent,shetiya2022fairness,li2023query,accinelli2020coverage},
coverage-based approaches for mitigating representation bias~\cite{shahbazi2023representation},
label repair for individual fairness~\cite{zhang2023iflipper},
mining the minority groups in data~\cite{dehghankar2025mining},
fair and private data generation~\cite{pujol2023prefair},
and fairness-aware 
data imputation~\cite{schelter2021jenga, zhu2024overcoming},
entity matching~\cite{shahbazi2023through,shahbazi2024fairness,moslemi2024threshold},
data cleaning~\cite{tae2019data},
data integration~\cite{nargesian2022responsible, chang2024data},
data augmentation~\cite{sharma2020data,erfanian2024chameleon},
and many more.
Within the related work, FairHash~\cite{shahbazi2024fairhash} shares some high-level similarities with our problem, as it finds a hash map that guarantees uniform distribution at the group level across hash buckets. However, despite its high level of similarity, the objective of FairHash is different from ours: given a set of attributes, FairHash finds a hash function as the combination of the attributes such that an equal-size binning on the hash values satisfies equal-group ratios.
Our problem is different: given an attribute where equal-size binning is biased, we want to partition it into unbiased bins.
As a result, the algorithms (based on ranking and necklace splitting) proposed in~\cite{shahbazi2024fairhash} cannot be adapted for unbiased binning.

% \color{blue}
\stitle{Group Spatial Fairness} %\marginpar{R1D4.a; Meta1}
Group spatial fairness aims to ensure equal model performance across geographic regions.
As a {\em post-process intervention}, 
\cite{kyriakopoulos2025promis} 
proposes a model-agnostic method that adjusts model parameters for predefined spatial regions to reduce unfairness while preserving accuracy.
In contrast, \cite{shaham2024fair} adjusts the spatial regions while training a separate model for each region to ensure calibration fairness across the neighborhoods.
Concretely, it designs a fair KD-tree that recursively splits the space by training a classifier within each cell and selecting split points that minimize calibration unfairness across the resulting partitions.
Similar to other Fair ML approaches~\cite{barocas2023fairness}, these approaches operate at the model development stage of the data analytics pipeline, whereas attribute binning is task-agnostic and occurs before sharing a dataset.

\stitle{Fair Clustering} %\marginpar{R1D4.b; Meta1}
Fair clustering~\cite{chhabra2021overview} has been extensively studied in general metric spaces, where the goal is to partition data into clusters while enforcing various fairness constraints, including balance (demographic parity)~\cite{chierichetti2017fair}, socially-fair clustering~\cite{ghadiri2021socially}, and fair center selection~\cite{thejaswi2022clustering, hotegni2023approximation}.
\cite{thejaswi2022clustering} studies diversity-aware clustering, where given a collection of (possibly overlapping) groups, the selected centers must include at least a specified number from each group while minimizing a clustering objective.
\cite{hotegni2023approximation} studies fair range clustering, where the number of selected centers from each group must lie within specified lower and upper bounds, offering a more flexible alternative to strict fairness constraints.
Conversely, 
\cite{chierichetti2017fair} defines fairness as demographic parity (balance) in each cluster. It introduces fairlet decomposition, a preprocessing approach that partitions the data into small, balanced subsets (fairlets) such that each fairlet satisfies a desired group proportion.
While these approaches are related in that they both partition the space, their objectives are fundamentally different. Clustering methods optimize distance-based objectives (e.g., k-means or k-median), whereas binning partitions an ordered attribute into intervals that satisfy structural constraints such as equal-size buckets. This mismatch makes fair clustering approaches ill-suited for unbiased binning, as confirmed by our experiments (Section~\ref{sec:exp:baseline}).

% \color{black}
\stitle{Attribute Binning} As a step before sharing datasets, attribute binning is popular in settings where privacy preservation, interpretability, or noise reduction is desired~\cite{mironchyk2017monotone,oliveira2008rigorous,zeng2014necessary}. Techniques such as equal-width binning have long been standard tools for transforming continuous features into categorical ones~\cite{kotsiantis2006discretization,boulle2005optimal}.
A variety of methods for attribute discretization have been proposed in the literature, with several surveys offering comprehensive overviews~\cite{kotsiantis2006discretization}.
Even though improper binning can magnify machine bias in downstream data analytics tasks, to the best of our knowledge, this paper is the \underline{\em first} to introduce and study the problem of unbiased binning to address such issues.
% \vspace{-2mm}
\section{Conclusion}\label{sec:conclusion}
In this paper, underscoring the impact of attribute binning on amplifying bias in data-driven decision systems, we introduced the unbiased binning problem to achieve group-parity across discretized buckets.
Next, recognizing practical limitations of unbiased binning, we introduced $\eps$-biased binning for settings where a small bias in the group ratios may be tolerable.
We proposed various solutions to address our problems and conducted extensive experiments on real-world benchmarks and synthetic datasets to evaluate their effectiveness.

\balance
\bibliographystyle{ACM-Reference-Format}
\bibliography{ref}

@inproceedings{agarwal2019fair,
  title={Fair regression: Quantitative definitions and reduction-based algorithms},
  author={Agarwal, Alekh and Dud{\'\i}k, Miroslav and Wu, Zhiwei Steven},
  booktitle={International conference on machine learning},
  pages={120--129},
  year={2019},
  organization={PMLR}
}

@inproceedings{thejaswi2022clustering,
  title={Clustering with fair-center representation: Parameterized approximation algorithms and heuristics},
  author={Thejaswi, Suhas and Gadekar, Ameet and Ordozgoiti, Bruno and Osadnik, Michal},
  booktitle={Proceedings of the 28th ACM SIGKDD Conference on Knowledge Discovery and Data Mining},
  pages={1749--1759},
  year={2022}
}

@article{chierichetti2017fair,
  title={Fair clustering through fairlets},
  author={Chierichetti, Flavio and Kumar, Ravi and Lattanzi, Silvio and Vassilvitskii, Sergei},
  journal={Advances in neural information processing systems},
  volume={30},
  year={2017}
}

@inproceedings{hotegni2023approximation,
  title={Approximation algorithms for fair range clustering},
  author={Hotegni, Sedjro Salomon and Mahabadi, Sepideh and Vakilian, Ali},
  booktitle={International Conference on Machine Learning},
  pages={13270--13284},
  year={2023},
  organization={PMLR}
}

@inproceedings{ghadiri2021socially,
  title={Socially fair k-means clustering},
  author={Ghadiri, Mehrdad and Samadi, Samira and Vempala, Santosh},
  booktitle={Proceedings of the 2021 ACM Conference on Fairness, Accountability, and Transparency},
  pages={438--448},
  year={2021}
}

@article{chhabra2021overview,
  title={An overview of fairness in clustering},
  author={Chhabra, Anshuman and Masalkovait{\.e}, Karina and Mohapatra, Prasant},
  journal={IEEE Access},
  volume={9},
  pages={130698--130720},
  year={2021},
  publisher={IEEE}
}

@inproceedings{kyriakopoulos2025promis,
  title={Promis: A post-processing framework for mitigating spatial bias},
  author={Kyriakopoulos, Dimitris and Sacharidis, Dimitris and Giannopoulos, Giorgos and Gunopulos, Dimitris and Dalamagas, Theodore},
  booktitle={Proceedings of the 33rd ACM International Conference on Advances in Geographic Information Systems},
  pages={488--498},
  year={2025}
}

@inproceedings{shaham2024fair,
  title={Fair spatial indexing: A paradigm for group spatial fairness},
  author={Shaham, Sina and Ghinita, Gabriel and Shahabi, Cyrus},
  booktitle={Advances in database technology: proceedings. International Conference on Extending Database Technology},
  volume={27},
  number={2},
  pages={150},
  year={2024}
}

@article{techrep,
  title={Unbiased Binning: Fairness-aware Attribute Representation},
  author={Asudeh, Abolfazl and Asoodeh, Mila and Asoodeh, Bita and Asudeh, Omid},
  url={https://github.com/asudeh/UnbiasedBinning/UnbiasedBinning.pdf},
  year={2025}
}

@book{draper1998applied,
  title={Applied regression analysis},
  author={Draper, NR},
  year={1998},
  publisher={McGraw-Hill. Inc}
}

@article{jagadish2014big,
  title={Big data and its technical challenges},
  author={Jagadish, Hosagrahar V and Gehrke, Johannes and Labrinidis, Alexandros and Papakonstantinou, Yannis and Patel, Jignesh M and Ramakrishnan, Raghu and Shahabi, Cyrus},
  journal={Communications of the ACM},
  volume={57},
  number={7},
  pages={86--94},
  year={2014},
  publisher={ACM New York, NY, USA}
}

@book{barocas2023fairness,
  title={Fairness and machine learning: Limitations and opportunities},
  author={Barocas, Solon and Hardt, Moritz and Narayanan, Arvind},
  year={2023},
  publisher={MIT press}
}

@article{bellamy2019ai,
  title={AI Fairness 360: An extensible toolkit for detecting and mitigating algorithmic bias},
  author={Bellamy, Rachel KE and Dey, Kuntal and Hind, Michael and Hoffman, Samuel C and Houde, Stephanie and Kannan, Kalapriya and Lohia, Pranay and Martino, Jacquelyn and Mehta, Sameep and Mojsilovi{\'c}, Aleksandra and others},
  journal={IBM Journal of Research and Development},
  volume={63},
  number={4/5},
  pages={4--1},
  year={2019},
  publisher={IBM}
}

@book{fischer2011history,
  title={A history of the central limit theorem: from classical to modern probability theory},
  author={Fischer, Hans},
  volume={4},
  year={2011},
  publisher={Springer}
}

@misc{compas2016,
  author = {Angwin, Julia and Larson, Jeff and Mattu, Surya and Kirchner, Lauren},
  title = {Machine Bias – ProPublica},
  year = {2016},
  howpublished = {\url{https://www.propublica.org/article/machine-bias-risk-assessments-in-criminal-sentencing}},
  note = {Accessed: July 2025},
  institution = {ProPublica}
}

@article{kotsiantis2006discretization,
  title={Discretization techniques: A recent survey},
  author={Kotsiantis, Sotiris and Kanellopoulos, Dimitris},
  journal={GESTS International Transactions on Computer Science and Engineering},
  volume={32},
  number={1},
  pages={47--58},
  year={2006},
  publisher={Citeseer}
}

@article{boulle2005optimal,
  title={Optimal bin number for equal frequency discretizations in supervized learning},
  author={Boulle, Marc},
  journal={Intelligent Data Analysis},
  volume={9},
  number={2},
  pages={175--188},
  year={2005},
  publisher={SAGE Publications Sage UK: London, England}
}

@incollection{machinebias,
  title={Machine bias},
  author={Angwin, Julia and Larson, Jeff and Mattu, Surya and Kirchner, Lauren},
  booktitle={Ethics of data and analytics},
  pages={254--264},
  year={2022},
  publisher={Auerbach Publications}
}

@article{barocas2016big,
  title={Big data's disparate impact},
  author={Barocas, Solon and Selbst, Andrew D},
  journal={Calif. L. Rev.},
  volume={104},
  pages={671},
  year={2016},
  publisher={HeinOnline}
}

@article{obermeyer2019dissecting,
  title={Dissecting racial bias in an algorithm used to manage the health of populations},
  author={Obermeyer, Ziad and Powers, Brian and Vogeli, Christine and Mullainathan, Sendhil},
  journal={Science},
  volume={366},
  number={6464},
  pages={447--453},
  year={2019},
  publisher={American Association for the Advancement of Science}
}

@article{hurley2016credit,
  title={Credit scoring in the era of big data},
  author={Hurley, Mikella and Adebayo, Julius},
  journal={Yale JL \& Tech.},
  volume={18},
  pages={148},
  year={2016},
  publisher={HeinOnline}
}

@article{kassen2013promising,
  title={A promising phenomenon of open data: A case study of the Chicago open data project},
  author={Kassen, Maxat},
  journal={Government information quarterly},
  volume={30},
  number={4},
  pages={508--513},
  year={2013},
  publisher={Elsevier}
}

@article{xu2023fair,
  title={Fair data representation for machine learning at the pareto frontier},
  author={Xu, Shizhou and Strohmer, Thomas},
  journal={Journal of Machine Learning Research},
  volume={24},
  number={331},
  pages={1--63},
  year={2023}
}

@article{calmon2017optimized,
  title={Optimized pre-processing for discrimination prevention},
  author={Calmon, Flavio and Wei, Dennis and Vinzamuri, Bhanukiran and Natesan Ramamurthy, Karthikeyan and Varshney, Kush R},
  journal={Advances in neural information processing systems},
  volume={30},
  year={2017}
}

@inproceedings{feldman2015certifying,
  title={Certifying and removing disparate impact},
  author={Feldman, Michael and Friedler, Sorelle A and Moeller, John and Scheidegger, Carlos and Venkatasubramanian, Suresh},
  booktitle={proceedings of the 21th ACM SIGKDD international conference on knowledge discovery and data mining},
  pages={259--268},
  year={2015}
}

@article{mehrabi2021survey,
  title={A survey on bias and fairness in machine learning},
  author={Mehrabi, Ninareh and Morstatter, Fred and Saxena, Nripsuta and Lerman, Kristina and Galstyan, Aram},
  journal={ACM computing surveys (CSUR)},
  volume={54},
  number={6},
  pages={1--35},
  year={2021},
  publisher={ACM New York, NY, USA}
}

@article{stoyanovich2020responsible,
  title={Responsible data management},
  author={Stoyanovich, Julia and Howe, Bill and Jagadish, Hosagrahar Visvesvaraya},
  journal={Proceedings of the VLDB Endowment},
  volume={13},
  number={12},
  year={2020}
}

@inproceedings{nargesian2022responsible,
  title={Responsible data integration: Next-generation challenges},
  author={Nargesian, Fatemeh and Asudeh, Abolfazl and Jagadish, HV},
  booktitle={Proceedings of the 2022 international conference on management of data},
  pages={2458--2464},
  year={2022}
}

@article{zhang2023iflipper,
  title={iflipper: Label flipping for individual fairness},
  author={Zhang, Hantian and Tae, Ki Hyun and Park, Jaeyoung and Chu, Xu and Whang, Steven Euijong},
  journal={Proceedings of the ACM on Management of Data},
  volume={1},
  number={1},
  pages={1--26},
  year={2023},
  publisher={ACM New York, NY, USA}
}

@article{roy2024fairness,
  title={Fairness in Preference Queries: Social Choice Theories Meet Data Management},
  author={Roy, Senjuti Basu and Schieber, Baruch and Talmon, Nimrod},
  journal={Proceedings of the VLDB Endowment},
  volume={17},
  number={12},
  pages={4225--4228},
  year={2024},
  publisher={VLDB Endowment}
}

@article{salimi2020database,
  title={Database repair meets algorithmic fairness},
  author={Salimi, Babak and Howe, Bill and Suciu, Dan},
  journal={ACM SIGMOD Record},
  volume={49},
  number={1},
  pages={34--41},
  year={2020},
  publisher={ACM New York, NY, USA}
}

@inproceedings{salimi2019interventional,
  title={Interventional fairness: Causal database repair for algorithmic fairness},
  author={Salimi, Babak and Rodriguez, Luke and Howe, Bill and Suciu, Dan},
  booktitle={Proceedings of the 2019 international conference on management of data},
  pages={793--810},
  year={2019}
}

@article{schelter2021jenga,
  title={JENGA: A framework to study the impact of data errors on the predictions of machine learning models},
  author={Schelter, Sebastian and Rukat, Tammo and Biessmann, Felix},
  year={2021}
}

@article{zhu2024overcoming,
  title={Overcoming Data Biases: Towards Enhanced Accuracy and Reliability in Machine Learning.},
  author={Zhu, Jiongli and Salimi, Babak},
  journal={IEEE Data Eng. Bull.},
  volume={47},
  number={1},
  year={2024}
}

@article{zhu2023consistent,
  title={Consistent Range Approximation for Fair Predictive Modeling},
  author={Zhu, Jiongli and Galhotra, Sainyam and Sabri, Nazanin and Salimi, Babak},
  journal={Proc. VLDB Endow.},
  year={2023}
}

@article{pujol2023prefair,
  title={PreFair: Privately Generating Justifiably Fair Synthetic Data},
  author={Pujol, David and Gilad, Amir and Machanavajjhala, Ashwin},
  journal={Proceedings of the VLDB Endowment},
  volume={16},
  number={6},
  pages={1573--1586},
  year={2023},
  publisher={VLDB Endowment}
}

@article{shahbazi2023through,
  title={Through the Fairness Lens: Experimental Analysis and Evaluation of Entity Matching},
  author={Shahbazi, Nima and Danevski, Nikola and Nargesian, Fatemeh and Asudeh, Abolfazl and Srivastava, Divesh},
  journal={Proceedings of the VLDB Endowment},
  volume={16},
  number={11},
  pages={3279--3292},
  year={2023},
  publisher={VLDB Endowment}
}

@article{shahbazi2023representation,
  title={Representation bias in data: A survey on identification and resolution techniques},
  author={Shahbazi, Nima and Lin, Yin and Asudeh, Abolfazl and Jagadish, HV},
  journal={ACM Computing Surveys},
  volume={55},
  number={13s},
  pages={1--39},
  year={2023},
  publisher={ACM New York, NY}
}

@inproceedings{shahbazi2024fairness,
  title={Fairness-aware data preparation for entity matching},
  author={Shahbazi, Nima and Wang, Jin and Miao, Zhengjie and Bhutani, Nikita},
  booktitle={2024 IEEE 40th International Conference on Data Engineering (ICDE)},
  pages={3476--3489},
  year={2024},
  organization={IEEE}
}

@article{shahbazi2024fairhash,
  title={Fairhash: A fair and memory/time-efficient hashmap},
  author={Shahbazi, Nima and Sintos, Stavros and Asudeh, Abolfazl},
  journal={Proceedings of the ACM on Management of Data},
  volume={2},
  number={3},
  pages={1--29},
  year={2024},
  publisher={ACM New York, NY, USA}
}

@inproceedings{moslemi2024threshold,
  title={Threshold-independent fair matching through score calibration},
  author={Moslemi, Mohammad Hossein and Milani, Mostafa},
  booktitle={Proceedings of the Conference on Governance, Understanding and Integration of Data for Effective and Responsible AI},
  pages={40--44},
  year={2024}
}

@article{chang2024data,
  title={Data distribution tailoring revisited: cost-efficient integration of representative data},
  author={Chang, Jiwon and Cui, Bohan and Nargesian, Fatemeh and Asudeh, Abolfazl and Jagadish, HV},
  journal={The VLDB Journal},
  volume={33},
  number={5},
  pages={1283--1306},
  year={2024},
  publisher={Springer}
}

@inproceedings{tae2021slice,
  title={Slice tuner: A selective data acquisition framework for accurate and fair machine learning models},
  author={Tae, Ki Hyun and Whang, Steven Euijong},
  booktitle={Proceedings of the 2021 International Conference on Management of Data},
  pages={1771--1783},
  year={2021}
}

@inproceedings{tae2019data,
  title={Data cleaning for accurate, fair, and robust models: A big data-AI integration approach},
  author={Tae, Ki Hyun and Roh, Yuji and Oh, Young Hun and Kim, Hyunsu and Whang, Steven Euijong},
  booktitle={Proceedings of the 3rd international workshop on data management for end-to-end machine learning},
  pages={1--4},
  year={2019}
}

@inproceedings{shetiya2022fairness,
  title={Fairness-aware range queries for selecting unbiased data},
  author={Shetiya, Suraj and Swift, Ian P and Asudeh, Abolfazl and Das, Gautam},
  booktitle={2022 IEEE 38th International Conference on Data Engineering (ICDE)},
  pages={1423--1436},
  year={2022},
  organization={IEEE}
}

@article{jagadish2022many,
  title={The many facets of data equity},
  author={Jagadish, H and Stoyanovich, Julia and Howe, Bill},
  journal={ACM Journal of Data and Information Quality},
  volume={14},
  number={4},
  pages={1--21},
  year={2022},
  publisher={ACM New York, NY}
}

@article{li2023query,
  title={Query refinement for diversity constraint satisfaction},
  author={Li, Jinyang and Moskovitch, Yuval and Stoyanovich, Julia and Jagadish, HV},
  journal={Proceedings of the VLDB Endowment},
  volume={17},
  number={2},
  pages={106--118},
  year={2023},
  publisher={VLDB Endowment}
}

@article{dehghankar2025mining,
  title={Mining the Minoria: Unknown, Under-represented, and Under-performing Minority Groups},
  author={Dehghankar, Mohsen and Asudeh, Abolfazl},
  journal={Proceedings of the VLDB Endowment},
  year={2025},
  publisher={ACM}
}

@article{erfanian2024chameleon,
  title={Chameleon: Foundation Models for Fairness-Aware Multi-Modal Data Augmentation to Enhance Coverage of Minorities},
  author={Erfanian, Mahdi and Jagadish, HV and Asudeh, Abolfazl},
  journal={Proceedings of the VLDB Endowment},
  volume={17},
  number={11},
  pages={3470--3483},
  year={2024},
  publisher={VLDB Endowment}
}

@inproceedings{accinelli2020coverage,
  title={Coverage-based Rewriting for Data Preparation.},
  author={Accinelli, Chiara and Minisi, Simone and Catania, Barbara},
  booktitle={EDBT/ICDT Workshops},
  year={2020}
}

@inproceedings{sharma2020data,
  title={Data augmentation for discrimination prevention and bias disambiguation},
  author={Sharma, Shubham and Zhang, Yunfeng and R{\'\i}os Aliaga, Jes{\'u}s M and Bouneffouf, Djallel and Muthusamy, Vinod and Varshney, Kush R},
  booktitle={Proceedings of the AAAI/ACM Conference on AI, Ethics, and Society},
  pages={358--364},
  year={2020}
}

@inproceedings{oliveira2008rigorous,
  title={Rigorous Constrained Optimized Binning for Credit Scoring},
  author={Oliveira, I and Chari, M and Haller, S},
  booktitle={SAS Global Forum},
  year={2008}
}

@article{zeng2014necessary,
  title={A necessary condition for a good binning algorithm in credit scoring},
  author={Zeng, Guoping},
  journal={Applied Mathematical Sciences},
  volume={8},
  number={65},
  pages={3229--3242},
  year={2014}
}

@article{mironchyk2017monotone,
  title={Monotone optimal binning algorithm for credit risk modeling},
  author={Mironchyk, Pavel and Tchistiakov, Viktor},
  journal={Utr. Work. Pap},
  year={2017}
}
\appendix
\section*{APPENDIX}

\section{Dynamic Programming (DP) Algorithm for $\eps$-biased Binning}\label{sec:dp2}
The solution proposed for unbiased binning does not directly extend to $\eps$-biased binning, mainly since Theorem~\ref{th:cboundary} no longer holds for $\eps$-biased binning.
Therefore, in the following, we propose a new idea that enables the development of $\DP$ for this setting.

Index the tuples $\{t_1,\cdots,t_n\}$ in a sorted order based on $x$.
The first step of our algorithm is to construct an upper-triangular table $\tee$ (Figure~\ref{fig:ebiased}).
For each pair $i<j\leq n$, the value of $\tee[i,j]$ is one if a bucket with $t_i$ and $t_j$ as its boundaries satisfies the $\eps$-bias requirement.
\begin{align}\label{eq:ebiasreq}
\tee[i,j]=\begin{cases}
    1 & \text{if }\max\limits_{l\in[\ell]}\left\vert
        \frac{|\Gee_l\cap \{t_{i+1},\cdots,t_j\}|}{j-i} - \frac{|\Gee_l|}{n}
    \right\vert \leq \eps \\[1.7ex]
    0 & \text{otherwise}
\end{cases}
\end{align}

Having introduced the table $\tee$, we are ready to develop the $\DP$ algorithm.

\stitle{the DP formulation} $\opt(j,\kappa)$ is the optimal solution for the $\eps$-biased binning problem with $\kappa$ buckets, on the tuples $\{t_1,\cdots,t_j\}$. %Hence, the optimal solution for $\eps$-biased binning is $\opt(n,k)$.
\(\opt(j, \kappa)\) returns $w^\downarrow$ and $w^\uparrow$, the minimum and maximum widths of its buckets.

\stitle{The boundary conditions}
A subproblem $\opt(j,\kappa)$ is infeasible if for all $i<j$, $\tee[i,j]=0$. We use $(0,\infty)$ as the output of an infeasible subproblem.
When $\tee[0,j]=1$, the bucket from the beginning to $t_j$ is $\eps$-biased. Hence, $\opt(j,1)=(j,j)$, if $\tee[0,j]=1$.
Also, $\kappa\leq j$.

\begin{figure}[t]
\centering
\begin{tikzpicture}[x={(0.49\textwidth/8,0)}, y={(0,-0.8cm)}] 

    % row 1
    \filldraw[fill=blue!30, draw=black] (0,0) rectangle (1,1);
    \node at (0.5,-.2) {$1$};
    \filldraw[fill=blue!30, draw=black] (1,0) rectangle (2,1);
    \node at (1.5,-.2) {$2$};
    \filldraw[fill=blue!30, draw=black] (2,0) rectangle (3,1);
    \node at (2.5,-.2) {$3$};
    \filldraw[fill=blue!30, draw=black] (3,0) rectangle (4,1);
    \node at (3.5,-.2) {$4$};
    \filldraw[fill=blue!30, draw=black] (4,0) rectangle (6,1);
    \node at (5,-.2) {$\cdots$};
    \filldraw[fill=blue!30, draw=black] (6,0) rectangle (7,1);
    \node at (6.5,-.2) {$n$};
    \node at (7.2,.5) {$0$};
    
    % row 2
    \node at (7.2,1.5) {$1$};
    \filldraw[fill=blue!30, draw=black] (1,1) rectangle (2,2);
    \filldraw[fill=blue!30, draw=black] (2,1) rectangle (3,2);
    \filldraw[fill=blue!30, draw=black] (3,1) rectangle (4,2);
    \filldraw[fill=blue!30, draw=black] (4,1) rectangle (6,2);
    \filldraw[fill=blue!30, draw=black] (6,1) rectangle (7,2);

    % row 3
    \node at (7.2,2.5) {$2$};
    \filldraw[fill=blue!30, draw=black] (2,2) rectangle (3,3);
    \filldraw[fill=blue!30, draw=black] (3,2) rectangle (4,3);
    \filldraw[fill=blue!30, draw=black] (4,2) rectangle (6,3);
    \filldraw[fill=blue!30, draw=black] (6,2) rectangle (7,3);

    % row 4
    \node at (7.2,3.5) {$3$};
    \filldraw[fill=blue!30, draw=black] (3,3) rectangle (4,4);
    \filldraw[fill=blue!30, draw=black] (4,3) rectangle (6,4);
    \filldraw[fill=blue!30, draw=black] (6,3) rectangle (7,4);

    % row 5
    \node at (7.2,5) {$\vdots$};
    \filldraw[fill=blue!30, draw=black] (4,4) rectangle (6,6);
    \filldraw[fill=blue!30, draw=black] (6,4) rectangle (7,6);
    \node at (5,5) {$\ddots$};

    % row 6
    \node at (7.35,6.5) {$n-1$};
    \filldraw[fill=blue!30, draw=black] (6,6) rectangle (7,7);

    \node at (4.2,-.4) {$\mathbf{j}$};
    \node at (7.4,4.2) {$\mathbf{i}$};
    \node at (6.5,0.5) {$\checkmark$};
    \node at (6.5,2.5) {$\checkmark$};
    \node at (6.5,3.5) {$\checkmark$};
    \node at (3.5,1.5) {$\checkmark$};
    \node at (3.5,2.5) {$\checkmark$};
    \node at (2.5,.5) {$\checkmark$};
    
\end{tikzpicture}
\caption{Illustration of the table $\tee$, representing the potential buckets that are $\eps$-biased.}
\label{fig:ebiased}
\end{figure}

\stitle{Recursive Formula}
The valid ($\eps$-biased) buckets ending at an index $j$ are specified by the cells $\tee[i,j]=1$. The subproblem $\opt(j,\kappa)$ must select its last bucket among such valid buckets.

For each index $i$ where $\tee[i,j]=1$, the last bucket is $\big(t_i[x],t_j[x]\big]$. Hence, $|B_\kappa\cap \dee |=j-i$.
Recursively solving the next subproblem, let $(w^\downarrow_i,w^\uparrow_i) = \opt(i,\kappa-1)$.
Then after adding $B_\kappa$,
\begin{align*}
    w_{i}^\uparrow &= \max (w_p^\uparrow~,~ j-i)\\
    w_{i}^\downarrow &= \min (w_p^\downarrow~,~ j-i)
\end{align*}
Among the possible options, the one with the minimum objective value is selected.
Hence,
\begin{align}\label{eq:dp2}
    \opt(j,\kappa) = (w^\downarrow_{i^*},w^\uparrow_{i^*}) \text{, where } i^*=\argmin\limits_{\mathcal{T}[i,j]=1} {w_{i}^\uparrow - w_{i}^\downarrow}
\end{align}

\begin{algorithm}[ht]
\caption{$\eps$-biased Binning}
\label{alg:ebiased}
\begin{algorithmic}[1]
\Require The dataset $\dee$, the value $k$, and the value $\eps$.
\Ensure The optimal $\eps$-biased $k$-binning.
% \Function{\sc UnbiasedBinning}{$\dee,k$} 
    \State $\{t_1,\cdots,t_n\}\gets$ {\sc sort}$(\dee,x)$ \Comment{\small sort $\dee$ on $x$}
    
    \For{$i\gets 0$ to $n-1$} \Comment{\small filling the table $\tee$}
        \For{$j\gets i+1$ to $n$}
            \State $\tee[i,j]\gets$ Equation~\ref{eq:ebiasreq}
        \EndFor
    \EndFor

    \State $M\gets [\langle(0,\infty),null \rangle]_{n\times k}$\Comment{\small initializing $M$}
    \For{$j\gets 1$ to $n$} \Comment{\small the boundary conditions}
        \If{$\tee[0,j]=1$}
            $M[j,1]\gets \langle(j,j),1 \rangle$
        \EndIf
        % \State $M[i,j]\gets \langle(j,j),0 \rangle$ {\bf if} $\tee[0,j]=1$ {\bf else} $\langle(0,\infty),null \rangle$
    \EndFor

    \Statex {\small /*Filling the matrix $M$*/}
    \For{$\kappa\gets 2$ to $k$}
    \For{$j\gets k$ to $n$}
        \State $\langle (w^\downarrow,w^\uparrow),i^*\rangle\gets \langle(0,\infty),null \rangle$
        \For{$i\gets 1$ to $j-1$}
            \State $\langle (w_i^\downarrow,w_i^\uparrow),\square \rangle\gets M[i,\kappa-1]$
            \If{$(w_i^\uparrow - w_i^\downarrow)<(w^\uparrow - w^\downarrow)$}
                \State $\langle (w^\downarrow,w^\uparrow),i^*\rangle\gets \langle(w_i^\downarrow,w_i^\uparrow),i \rangle$
            \EndIf
        \EndFor
        $M[j,\kappa]\gets \langle (w^\downarrow,w^\uparrow),i^*\rangle$
    \EndFor
    \EndFor

    \Statex {\small /*Tracing back the output*/}
    \If{$M[n,k] = \langle(0,\infty),null \rangle$}
        {\bf Return} Infeasible
    \EndIf
    \State $S\gets [~]$ \Comment{\small initialize the output}
    \State $j\gets n; \kappa \gets k$
    \While{$\kappa>0$}
        \State $\langle(w_i^\downarrow,w_i^\uparrow),i \rangle \ gets M[j,\kappa]$
        \State $S.push(t_{i}[x])$ \Comment{\small add $t_{i}[x]$ to the beginning of $S$}
            \State $j\gets i$; $\kappa \gets \kappa-1$
    \EndWhile
    \State {\bf Return} $S$
% \EndFunction
\end{algorithmic}
\end{algorithm}

\subsubsection{Implementation based on Matrix Filling}
Algorithm~\ref{alg:ebiased} shows the implementation of our $\DP$ algorithm for the $\eps$-biased binning problem.
The algorithm defines a matrix $M_{n\times k}$, where each cell $M[j,\kappa]$ includes $\langle\opt(j,\kappa),i^*\rangle$. $i^*$ is the selected index according to Equation~\ref{eq:dp2}. The matrix $M$ is then filled based on the DP formulation discussed above.

\vspace{2mm}
{\sc Lemma~\ref{lem:ebiased-complexity}.}
{\it
% \begin{lemma}\label{lem:ebiased-complexity}
    Algorithm~\ref{alg:ebiased} has a time complexity of $O(n^2k)$ and space complexity of $O(n^2)$.
% \end{lemma}
}

\begin{proof}
{\em Space complexity:}
Storing the table $\tee$ requires a space complexity of $O(n^2)$. Therefore, although the matrix $M$ has a size of $n\times k$, the overall space complexity is $O(n^2)$.

\vspace{2mm}\noindent{\em Time complexity:}
\begin{itemize}[leftmargin=*]
    \item Sorting $\dee$ based on the attribute $x$ is done in $O(n\log n)$.
    \item The table $\tee$ is filled in $O(n^2)$, for a constant value of $\ell$.
    \item Iterating through the indices $i<j$ to fill each cell of matrix $M[j,\kappa]$ is in $O(n)$. Therefore, filling the matrix $M$ is done in $O(n^2k)$
    \item Finally, the trace back in the matrix $M$ to find the optimal binning is in $O(k)$.
\end{itemize}
Therefore, putting everything together, Algorithm~\ref{alg:ebiased} has a time complexity of $O(n^2 k)$.
\end{proof}

Note that, by setting $\eps=0$, one can use Algorithm~\ref{alg:ebiased} for solving the unbiased binning problem. In other words, this algorithm also serves as a (slower) baseline for Problem~\ref{problem}.

\begin{figure}[!tbh]
\centering
    \begin{subfigure}[t]{0.7\linewidth}
        \includegraphics[width=\linewidth]{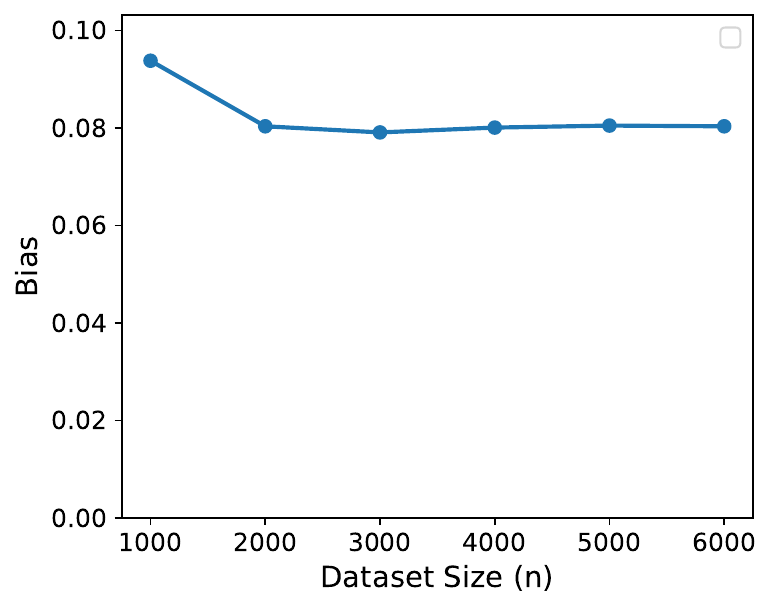}
        \caption{}
        \label{fig:compas:bias-n}
    \end{subfigure}

    \begin{subfigure}[t]{0.7\linewidth}
        \includegraphics[width=\linewidth]{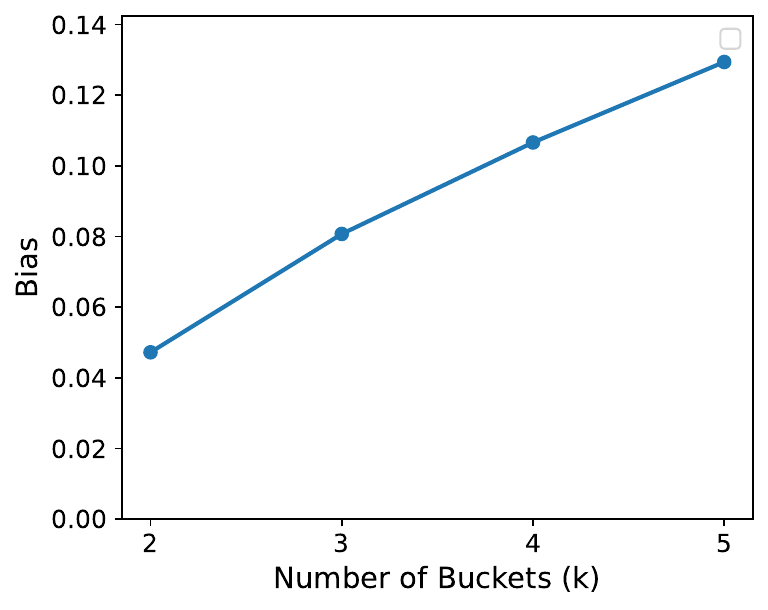}
        \caption{}
        \label{fig:compas:bias-k}
    \end{subfigure}
\vspace{-4mm}
\caption{COMPAS dataset: Bias of equal-size binning}
\label{fig:compas}
% \vspace{-5mm}
\end{figure}

\begin{figure}[!tbh]
\centering
    \begin{subfigure}[t]{0.7\linewidth}
        \includegraphics[width=\linewidth]{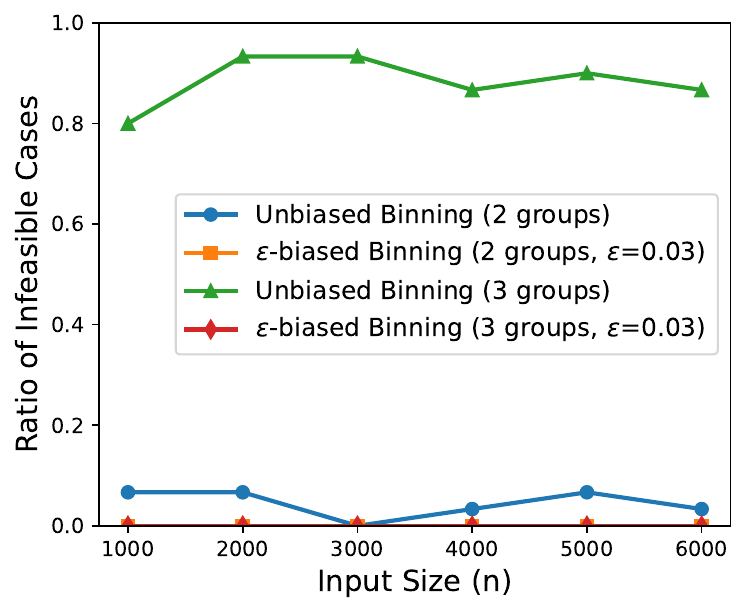}
        \caption{Input size}
        \label{fig:compas:invalid-n}
    \end{subfigure}
    
    \begin{subfigure}[t]{0.7\linewidth}
        \includegraphics[width=\linewidth]{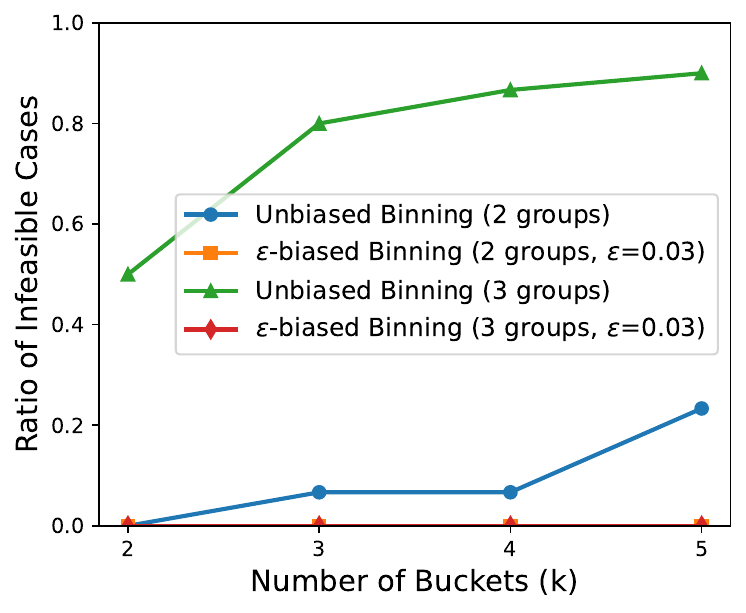}
        \caption{\# buckets}
        \label{fig:compas:invalid-k}
    \end{subfigure}
\vspace{-4mm}
\caption{COMPAS dataset: Infeasible ratio.}
\label{fig:compas}
% \vspace{-5mm}
\end{figure}

\begin{figure}[!tbh]
\centering
    \begin{subfigure}[t]{0.7\linewidth}
        \includegraphics[width=\linewidth]{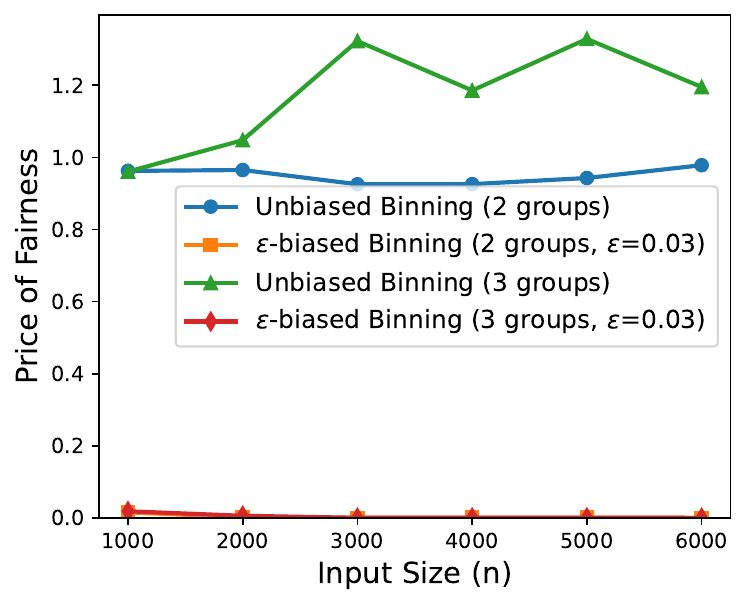}
        \caption{Input size}
        \label{fig:compas:PoF-n}
    \end{subfigure}
    
    \begin{subfigure}[t]{0.7\linewidth}
        \includegraphics[width=\linewidth]{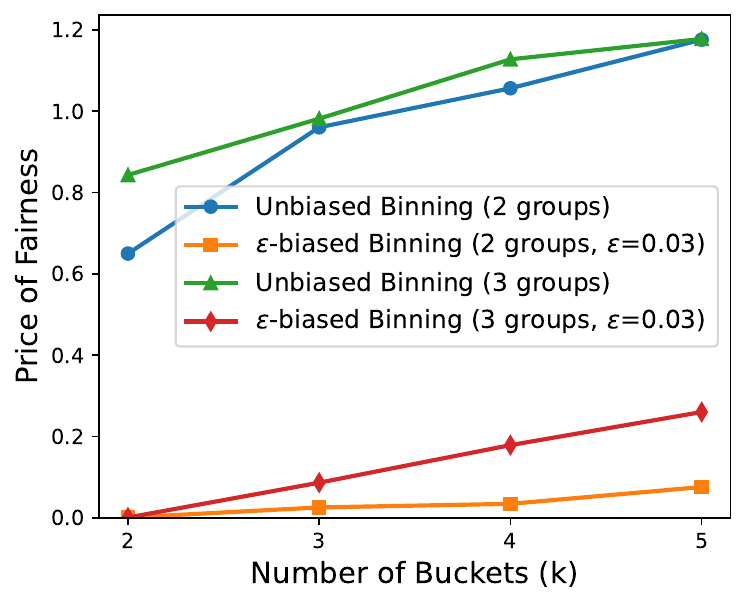}
        \caption{\# buckets}
        \label{fig:compas:PoF-k}
    \end{subfigure}
\vspace{-4mm}
\caption{COMPAS dataset: Price of Fairness.}
\label{fig:compas}
% \vspace{-5mm}
\end{figure}

\begin{figure}[!tbh]
\centering
    \begin{subfigure}[t]{0.7\linewidth}
        \includegraphics[width=\linewidth]{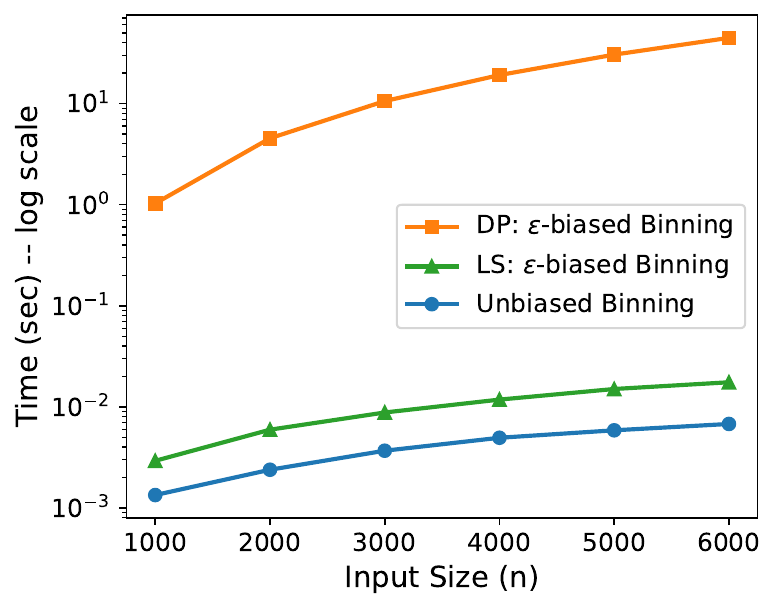}
        \caption{Input size}
        \label{fig:compas:time-n}
    \end{subfigure}

    \begin{subfigure}[t]{0.7\linewidth}
        \includegraphics[width=\linewidth]{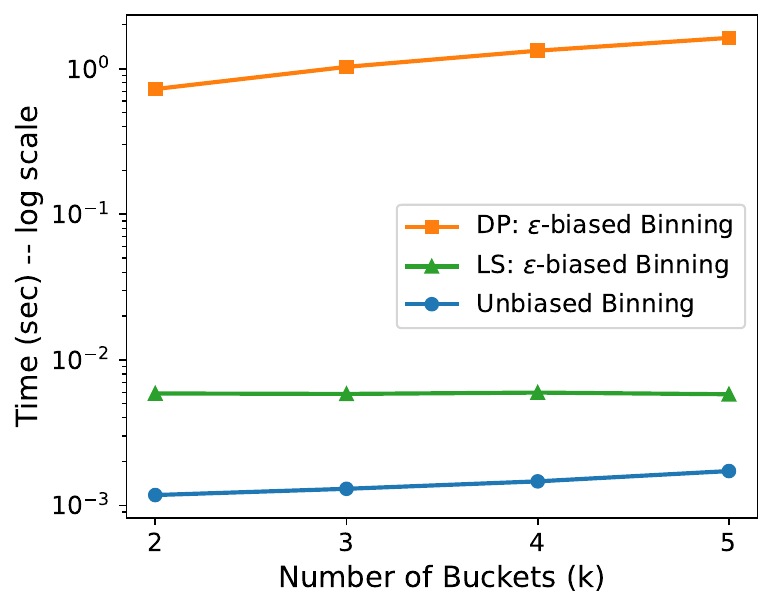}
        \caption{\# buckets}
        \label{fig:compas:time-k}
    \end{subfigure}
\vspace{-4mm}
\caption{COMPAS dataset: Time.}
\label{fig:compas}
% \vspace{-5mm}
\end{figure}

\begin{figure}[!tb]
\centering
    \begin{subfigure}[t]{0.7\linewidth}
        \includegraphics[width=\linewidth]{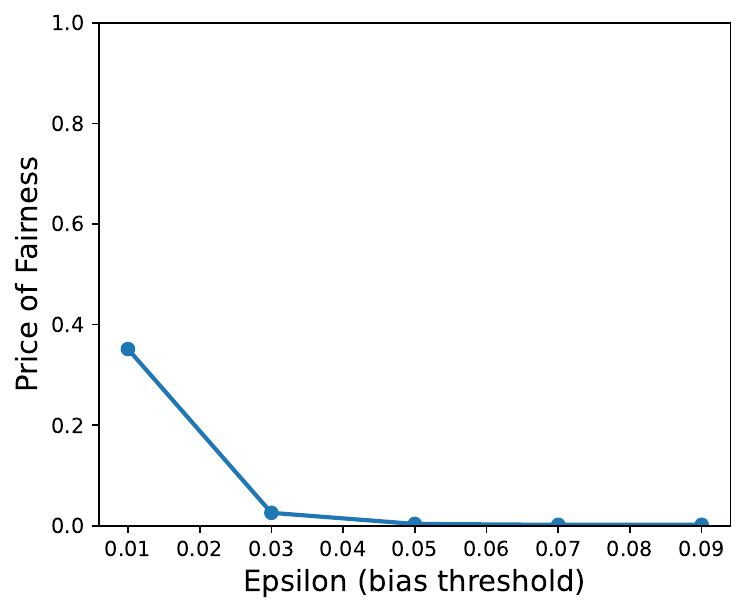}
        % \vspace{-3mm}
        \caption{PoF vs $\eps$}
        \label{fig:compas:PoF-eps}
    \end{subfigure}
    
    \begin{subfigure}[t]{0.7\linewidth}
        \includegraphics[width=\linewidth]{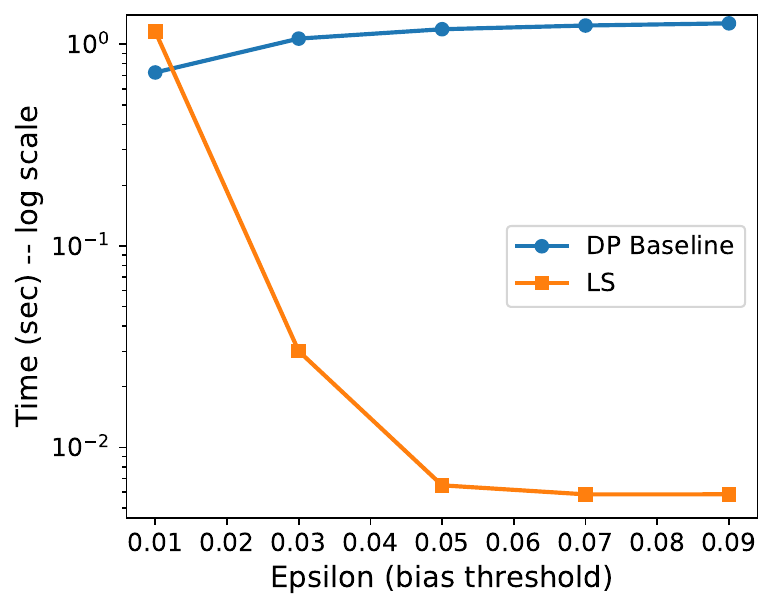}
        % \vspace{-3mm}
        \caption{Time vs $\eps$}
        \label{fig:compas:time-eps}
    \end{subfigure}
    \vspace{-4mm}
\caption{The effect of varying the bias threshold ($\eps$) on time and PoF (COMPAS)}
\label{fig:exp:eps}
\end{figure} 

\section{Extended Experiment Result on COMPAS Dataset}\label{app:exp}
Our experiment results on the COMPAS dataset for various input sizes and the number of buckets are provided in Figure~\ref{fig:compas}. The default values are $n=1000$, $k=3$, and $\eps=0.03$.
COMPAS recidivism scores are well-known for machine bias, highlighting the bias against the racial group \at{Black} compared to the group of \at{White} individuals~\cite{compas2016}.
We also observed this in Figures~\ref{fig:compas:bias-n} and \ref{fig:compas:bias-k}, where the equal-size binning of the raw recidivism scores resulted in up to 10\% bias for varying input sizes, while this value increased to around 14\% for $k=5$ risk-score buckets.

While the experiment results for the COMPAS dataset were consistent with the German Credit dataset, the ratio of infeasible cases (Figures~\ref{fig:compas:invalid-n} and \ref{fig:compas:invalid-k}) for unbiased binning for the two groups \at{black} and \at{white}, was less (between 0 to 20\%). Still, the PoF was significant for unbiased binning. On the other hand, allowing a small bias threshold of 3\% could reduce the PoF to almost zero in all cases.
Next, to evaluate the impact of the number of groups ($\ell$), we repeated our experiments for the three groups of \at{Black}, \at{Hispanic}, and \at{White}. The ratio of infeasible cases surged to more than 80\% in Figures~\ref{fig:compas:invalid-n} and \ref{fig:compas:invalid-k}. This is because the chance of finding the boundaries with unbiased ratios on both \at{Black} and \at{Hispanic} is less than only considering the \at{Black} group. Therefore, the expected number of boundary candidates for the three groups is less than the binary groups. Similarly, the PoF of unbiased binning is higher for considering three groups.
However, allowing a maximum bias of 3\%, the PoF significantly dropped, while there was no infeasible case in all settings. Still, the PoF increased to around 25\% for $k=5$ buckets.

The running times of the algorithms were nearly identical for two and three groups. Therefore, in Figures~\ref{fig:compas:time-n} and \ref{fig:compas:time-k}, we only provided the running times for binary-group experiments. Similar to our experiments on the German credit dataset, the unbiased-binning algorithm was the fastest in all cases, while the $\DP$ algorithm for $\eps$-biased binning was $10^3$ times slower than it. Also, the $\LS$ algorithm proved to be practically efficient in all cases.

Next,
we compared the performance of the $\DP$ and $\LS$ algorithms for different values of bias threshold. The results are provided in Figure~\ref{fig:exp:eps}. First, as reflected in Figures~\ref{fig:compas:PoF-eps}, while the price of fairness is large when $\eps$ is 1\%, it suddenly drops for larger values of $\eps$. This confirms that by slightly increasing the max-bias threshold, one can find a binning that is very close to equal-size binning.
Next, looking at Figures~\ref{fig:compas:time-eps}, we can confirm the different running-time behaviors of the $\DP$ and $\LS$ algorithms when $\eps$ increases.
Increasing the value of $\eps$ slightly increases the running time of $\DP$ as it increases the valid boundaries. On the other hand, while the running time of the $\LS$ algorithm is high for $\eps=0.01$, it quickly decreases for larger values of $\eps$. This major drop can be explained by the significant drop in the PoF when $\eps$ increases. When the PoF is low, the near-optimal solution found by the $\DnC$ algorithm is small, hence the $\LS$ can stop after checking a few combinations.
\end{document}